\documentclass[twoside]{article}
\usepackage[left=2cm,right=2cm,top=2cm,bottom=2cm]{geometry}
\usepackage{multicol}
\usepackage{algorithm,algorithmic,multirow}
\usepackage{bbm, relsize}
\usepackage{amsmath}
\usepackage{amsfonts}
\usepackage{amssymb}
\usepackage{graphicx}
\usepackage{color}
\usepackage[normalem]{ulem}
\usepackage{nomencl}
\makenomenclature

\newcommand{\yadi}{\nomenclature}
\newenvironment{proof}{\noindent{\sc Proof.}}{\qed}
\newtheorem{theorem}{Theorem}[section]
\newtheorem{lemma}{Lemma}[section]

\newtheorem{remark}{Remark}[section]
\newtheorem{definition}{Definition}[section]

\newcommand{\qed}{$\blacksquare$}

\def\bhag#1{\noindent
\setcounter{equation}{0}
\section{#1}
}

\def\RR{{\mathbb R}}
\def\CC{{\mathbb C}}
\def\ZZ{{\mathbb Z}}

\def\TT{\mathbb T}

\usepackage{listings}

\definecolor{dkgreen}{rgb}{0,0.6,0}
\definecolor{gray}{rgb}{0.5,0.5,0.5}
\definecolor{mauve}{rgb}{0.58,0,0.82}

\def\be{\begin{equation}}
\def\ee{\end{equation}}
\def\bea{\begin{eqnarray}}
\def\eea{\end{eqnarray}}

\def\disp{\displaystyle}

\def\donchitre#1#2{\vskip 6.5cm\noindent
\parbox[t]{1in}{\special{eps:#1.eps x=6.5cm y=5.5cm}}
\hbox to 7cm{}\parbox[t]{0.0cm}{\special{eps:#2.eps x=6.5cm y=5.5cm}}}

\def\gs{\gtrsim}
\def\ls{\lesssim}

\def\hati{\hat{\mathbf{i}}}

\title{Localized kernel method for separation of linear chirps}
\author{
S.~Kitimoon\thanks{Institute of Mathematical Sciences, Claremont Graduate University, Claremont, CA 91711, U.S.A..\\
\textsf{email:} sippanon.kitimoon@cgu.edu},\
E. Mason\thanks{
Hawkeye 360, Herndon, VA 20170, U.S.A.
\textsf{email:} eric.mason@he360.com},\
 H.~N.~Mhaskar\thanks{
Institute of Mathematical Sciences, Claremont Graduate University, Claremont, CA 91711, U.S.A.. The research of HNM was supported in part by  ONR grants N00014-23-1-2394, N00014-23-1-2790.
\textsf{email:} hrushikesh.mhaskar@cgu.edu}
}

\begin{document}
\maketitle
\begin{abstract}
The task of separating a superposition of signals into its individual components is a common challenge encountered in various signal processing applications, especially in domains such as audio  and radar signals. 
A previous paper by Chui and Mhaskar proposes a method called Signal Separation Operator (SSO) to find the instantaneous frequencies and amplitudes of such superpositions where both of these change continuously and slowly over time.
In this paper, we amplify and modify this method in order to separate chirp signals in the presence of crossovers, a very low SNR, and discontinuities.
We give a theoretical analysis of the behavior of SSO in the presence of noise to examine the relationship between the minimal separation, minimal amplitude, SNR, and sampling frequency.
Our method is illustrated with a few examples, and numerical results are reported on a simulated dataset comprising 7 simulated signals.
 
\end{abstract}


\vspace{1em}

\section{Introduction}\label{bhag:intro}

The task of separating a superposition of signals into its individual components is a common challenge encountered in various signal processing applications, especially in domains such as audio \cite{makino2018,Nugraha2016,yuan2025} and radar \cite{Joshil2020,Filice2013,Dutt2019,Wang2022,Nuhoglu2023,Yanchao2018}. Due to the extensive set of applications and the frequent occurrence of signal source separation challenges, methodologies are often tailored to the specific field in question, as well as the level of prior knowledge of both the component signals and the mixing mechanism.  

In these applications, the RF system makes measurements of the electromagnetic spectrum with a large instantaneous bandwidth, typically above $1$ GHz and continuing to increase with continual improvements in hardware capabilities.
The resulting data may consist of tens to hundreds of signals mixed together after just a few seconds to minutes of recording, many overlapping in time and/or frequency and measured at low signal-to-noise (SNR).
Most radars transmit a signal that consists of a sequence of linear frequency modulated (LFM) pulses, commonly referred to as \emph{chirp signals}.

A linear chirp is a signal (pulse) of the form $A(t)\exp(\hati\phi(t))$, where $\hati=\sqrt{-1}$ and $\phi$ is a quadratic polynomial, so that the instantaneous frequency $\phi'(t)$ is a linear function.
If $\gamma$ is the starting time of the pulse, $d$ is its duration, $\omega$ is the initial frequency, $B$ is the bandwidth, then the phase of the pulse is given by
\be\label{eq:pulsephasedef}
\phi(t)=\phi(\gamma,  \omega, B, d; t)=\omega(t-\gamma)+(B/d)(t-\gamma)^2.
\ee
\yadi{$\omega$, $\gamma$, $B$, $d$}{Chirp parameters} The pulse is assumed to repeat periodically in a sequence of $m$ pulses with a pulse repetition interval PRI being the time between the subsequent pulses or bursts of energy.
Thus, for a signal $s$ starting at $t=t_0^*$ its $n$-th pulse is given by
\be\label{eq:pulsebursts}
s_n(t)=s_n(A,  \omega, B, t_0^*, d, m, \mbox{PRI}; t)= \begin{cases}
\begin{aligned}
&\disp A(t)\exp\left(\hati (\omega(t-n*\mbox{PRI})+(B/d)(t-n*\mbox{PRI})^2)\right),&\\
&\mbox{ if $t_0^*+n*\mbox{PRI}\le t\le t_0^*+n*\mbox{PRI}+d$, \ $n=0,\cdots, m-1$,}\\
&0, &\mbox{otherwise}.
\end{aligned}
\end{cases}
\ee
\yadi{PRI}{pulse repetition interval}
\yadi{$\hati$}{$\sqrt{-1}$}
 Naturally, the chirp signal has the form $s=\sum_n s_n$.

The first step in processing the electromagnetic spectrum is the detection and parameterization step, where the time and frequency locations of individual pulses are determined, and their parameters estimated. 
The second step in the EW receiver is deinterleaving, where groups of pulses are associated with each other to obtain the different radar signals, from which further insights are extracted.  
Here, we focus on the first problem of the detection, separation, and parameter estimation of individual radar pulses.

Thus, the problem which we intend to solve is the following. 
We observe a signal of the form
\be\label{eq:observations}
F(t)=f(t)+\epsilon(t),
\ee
where $\epsilon$ is a sub-Gaussian random noise, and $f$ is the ground truth signal
\be\label{eq:groundtruth}
f(t)=\sum_{j=1}^K s(A_j, \omega_j, B_j, t_{0,j}^*, d_j, m_j, \mbox{PRI}_j;t), \qquad t\in [0,T].
\ee
From equidistant observations $F(\ell\delta)$, $\ell=0,\cdots, N$, we wish to determine the number of pulses at any time $t$ as well as the parameters $A_j, \omega_j, B_j, t_{0,j}^*, d_j, m_j, \mbox{PRI}_j$ of the pulses.
In this paper, we will assume that all $A_j$'s are $\equiv 1$. 

The instantaneous frequency of the chirp $s_n$ in \eqref{eq:pulsebursts} at time $t$ is defined by $\omega+(2B/d)(t-n*\mbox{PRI})$.
Our approach is to assume (following \cite{ingrid2011, bspaper}) that on a sufficiently small interval of the form $[t-\Delta,t+\Delta]$, this frequency is approximated by a constant frequency at the center of the interval. 
We will then use specially constructed localized kernels to solve the problem for each of these intervals as described in \cite{loctrigwave, bspaper}.
A more elaborate subdivision is necessary when the instantaneous frequencies of different pulses cross each other.
We will demonstrate experimentally that our method is very stable, capable of working with a very small SNR if the sampling rate is sufficiently high.
Thus, the only limitation is the hardware capability, not a theoretical limitation.
We will analyze theoretically the connection between $\Delta$ and the various parameters of the pulses and the accuracy.

\color{black}
Even though the problem is formulated in terms of a continuous variable, we only have discrete samples of the form $F(\ell \delta)$, $|\ell|<N$.
Therefore, we cannot use techniques based on the Fourier transform, including short term Fourier transform. 
This would require a discretization, resulting in many errors: discretization error, aliasing, etc. 
One could think of extending $F$ as a $2N\delta$ periodic function, and using harmonic analysis on the group $\ZZ_N$.
This would restrict the instantaneous frequencies to be lattice points.

Our approach is simpler.
First,  we will show in Theorem~\ref{theo:reduction} that for any sufficiently small interval of the form $[t^*-\Delta, t^*+\Delta]$, $F(\ell\delta)$ can be approximated by a quantity of the form $\tilde{\mu}(\ell)=\sum_{k=1}^{K(t^*)}A_k(t^*)\exp(-\hati \lambda_k(t^*)\ell)+\mbox{noise}$, where $\lambda_k(t^*)$'s are (unitless) instantaneous frequencies at $t^*$, scaled by the sampling parameter $\delta$.
For a sufficiently large sampling rate (corresponding to a small $\delta$), we can ensure that each $\lambda_k(t^*)\in [0,\pi]$. 

We then revert our thinking about the physical notions of time and frequency. 
The points $t^*$ and the end points $t^*\pm \Delta$ need to be chosen so as to be among the grid points $\ell\delta$ at which the signal is sampled.
In our reverse thinking, we now identify the sampling point $\ell\delta$ in the interval $[t^*-\Delta, t^*+\Delta]$ with the  ``frequency'' $\ell$ in the sense of trigonometric Fourier series.
This has the advantage that we can use the samples directly without having to go through any transformations to convert them into the ``Fourier domain''. 
In particular, this bypasses problems like discretization, aliasing, circularity, etc. which would arise if we thought in terms of Fourier transform.

The instantaneous frequencies $\lambda_k(t^*)$ are now viewed as unknown points in the ``time domain'' $\RR/(2\pi\ZZ)$ to be determined.
In particular, we no longer need to require $\lambda_k(t^*)$'s to be grid points.

Then $\tilde{\mu}(\ell)$ are seen as Fourier coefficients of a $2\pi$-periodic (complex) measure $\mu$ associating the mass $A_k(t^*)$ with each point $\lambda_k(t^*)$. 
If the samples $F(\ell\delta)$ were available without noise for all $\ell\in\ZZ$, we would have formally: 
$$
\mu(x)=\sum_{\ell\in\ZZ}\tilde{\mu}(\ell)\exp(i\ell x)=\sum_{\ell\in\ZZ}F(\ell\delta)\exp(i\ell x), \qquad x\in\RR/(2\pi\ZZ).
$$
Even then, this Fourier series will not converge, and in any case, we have the values of $F(\ell\delta)$ only for $|\ell|<N$ for some $N$.
In the classical theory of trigonometric approximation, certain summability factors are introduced to obtain an approximation of a function given finitely many of its Fourier coefficients.
We  have constructed a highly localized summability kernel for this purpose.
This enables us both to give quantitative estimates on the error in the  estimation of the instantaneous frequencies, as well achieve a substantial noise reduction.

\color{black}
                 
The main contributions of this paper are the following.
\begin{enumerate}
\item This paper is the first of its kind that illustrates the use of SSO in the presence of a very high noise, crossover frequencies, and various components of the signal having different start and end points.
\item In particular, we describe how the SSO algorithm needs to be modified to detect the start and stop points of the various signal components and crossover frequencies.
\item We examine the theory behind the SSO to highlight the connection between the noise level, the sampling frequency, and the length of the snippet required for SSO to work properly.
\item {Our algorithms are fully automated, assuming only a few hyperparameters, listed in Table~\ref{tab:tunable_parameters}, and discussed in detail in Section~\ref{bhag:tunable}.}
\item Our algorithm utilizes FFT, and hence, is very fast.
\end{enumerate}

We will describe some related work in Section~\ref{bhag:related}. 
The mathematical foundation behind our work is explained in Section~\ref{bhag:foundations}, including the intuition and definition of the Signal Separation Operator (SSO) introduced in \cite{bspaper}. 
The algorithms and methodology to use SSO for the chirp separation problem are described in Section~\ref{bhag:algorithm} with step-by-step illustrations on one example.
In particular, we demonstrate the connection between the sampling rate, noise level, minimal separation among the instantaneous frequencies.
We have tested our algorithm on a database of 7 examples. 
The process to generate this data is described in Section~\ref{bhag:datagen}.
Section~\ref{bhag:tunable} describes some heuristics on the choice of the tunable parameters listed in Table~\ref{tab:tunable_parameters}.
A full report on   the entire dataset which we worked on is discussed in Section~\ref{bhag:numerical_results}, and the results summarized in the appendix.

\bhag{Related works}\label{bhag:related}

There are several general ways to approach the problem of radar signal detection and separation.  

In the case when the frequencies are constant, the problem is known as parameter estimation in exponential sums. 
This problem apparently goes back to the 1795 paper by Prony \cite{prony_original}.
Mathematically, the problem can be stated as follows.
Let $K\ge 1$ be an integer, $\{\lambda_j\}_{k=1}^K\subset (-\pi,\pi]$, $\mu=\sum_{j=1}^K A_j\delta_{\lambda_j}$ for some complex numbers $\{A_j\}_{k=1}^K$, where $\delta_x$ denotes the Dirac delta supported at $x$. 
Given noisy variants of the Fourier coefficients
\be\label{eq:stationary}
\hat{\mu}(\ell)=\sum_{j=1}^K A_j\exp(-\hati \ell \lambda_j), \qquad |\ell|<N,
\ee
for some integer $N\ge 1$ (where $\hati=\sqrt{-1}$), determine the locations $\lambda_j$ and the coefficients $A_j$.
\yadi{$\lambda_j$}{unitless, constant frequencies}
Of course, finding $\lambda_j$'s is a bigger challenge; once we obtain these, the $A_j$'s can be found using a least square fit for the data \eqref{eq:stationary}.
We will not go into a survey of the very large amount of literature on this problem, except to point out that our current work is based on ideas developed in \cite{loctrigwave, bspaper, mhaskar2024robust}.

If the radar pulses have sufficiently high SNR, then they can be detected using an energy-based envelope detector operating on the complex valued time-series data, known as In-phase and Quadrature (IQ) data \cite{lehtomaki2005,Juliano2020,Mello2018}.
It is actually more common for modern solid state radars to transmit low probability of intercept (LPI) waveforms that utilize LFM pulses at low and negative SNR \cite{Xie2025, pace2009detecting}; in this case, the envelope detector will no longer be suitable.
In the low power setting, there are two broad classes of methods, the first operates on the IQ data directly, while the second utilizes time-frequency transforms approaching detection and separation as a 2D signal processing problem. 
In the time domain, there are analytic methods that operate on the IQ data directly, such as variations of independent components algorithm (ICA) \cite{Joshil2020,Lin2020,Jiang2015,Shao2020}, with extensions to nonlinear chirps \cite{PENG2024,zhou2018}, and incorporation of simple neural networks \cite{pang2022}.
A more common approach is to utilize time-frequency transforms and process the result as an image instead of operating on the IQ data directly \cite{Jin2022,Swiercz2020,Hou2020,Gongming2019,Liu2015,Su2022}.

A general framework that includes both the problems of chirp separation and parameter estimation in exponential sums is the following.
One observes a signal of the form
 \be\label{eq:nonstationary}
f(t)=\sum_{j=1}^K f_j(t)=\sum_{j=1}^K A_j(t)\exp(\hati \phi_j(t)), \qquad 0\le t\le T,
\ee
where the instantaneous amplitudes (IA) $A_j$'s are complex valued, and $\phi_j$'s are differentiable functions of time.
\yadi{$A_j(t)$, $\phi_j(t)$}{Instantaneous complex amplitudes and phases}

The quantity $\phi_j'(t)$ is known as the instantaneous frequency (IF) of the component $f_j$.
The problem is to estimate the IFs and the IAs.
This is also a very important and old problem, starting with a paper by a Nobel Laureate, Gabor \cite{gabor1946}.
One of the first approaches to solve the problem is the Hilbert-transform based empirical mode decomposition (EMD) by Huang \cite{huang1998empirical}. 
However, there is no mathematical theory behind this approach.
A mathematically rigorous approach was given by Daubeschies and her collaborators \cite{daubechies1996nonlinear, ingrid2011}.
These seminal papers have given rise to a lot of variations and improvements.
A complete survey is both beyond the scope of this paper as well as not relevant here.
We refer to \cite{li2022chirplet} for  more references on this topic.
In our paper, we will use the algorithm given in \cite{thakur2013synchrosqueezing} for comparing our method with SST.

In \cite{bspaper}, we have applied the technique based on localized trigonometric kernels to define a signal separation operator (SSO) to solve this problem under conditions weaker than those in \cite{ingrid2011}.
Our conditions do assume that all the components are present in the interval $[0,T]$; i.e., the number of signals $f_j$ is the same throughout the interval, and there is no discontinuity in any of these.
Our method finds the IFs by looking at the peaks of a power spectrum rather than an elaborate reference curve mechanism of SST. 
Likewise, the IAs are estimated by a simple substitution compared to the limiting process required in SST. 
Finally, our method can be implemented very fast using FFT.

The paper \cite{li2022chirplet} deals specifically with the separation of chirplets, based on a modification of a filter described in \cite{bspaper}.
It is not clear when the various conditions in this paper are satisfied.
Also, the method requires a three dimensional search.
In this paper, we will demonstrate how the methods developed in \cite{bspaper} can be adapted for the chirp separation problem, including the case of crossover frequencies.



\section{Mathematical foundations}\label{bhag:foundations}

In this section, we start in Section~\ref{bhag:stationary} by explaining the basic idea and theorem behind our methods in the case of signals with constant amplitudes and frequencies. 
In Section~\ref{bhag:reduction}, we explain  how the general problem of blind source signal separation can be reduced to this case at each point under certain assumptions as in \cite{ingrid2011, bspaper}. 
A theoretical analysis of the  performance of SSO in the presence of noise is given in Section~\ref{bhag:sso}.

\subsection{Constant amplitudes and frequencies}\label{bhag:stationary}
In the sequel, we write $\TT=\RR/(2\pi\ZZ)$; i.e., $\TT$ is the quotient space of $\RR$ where $x$ and $y$ are identified if $x-y$ is an integer multiple of $2\pi$.
We describe first the case when the amplitudes and frequencies are constant; i.e., the problem of parameter estimation with the information of the form
\be\label{eq:stationarynoisy}
\tilde{\mu}(\ell)=\sum_{j=1}^K A_j\exp(-\hati \ell \lambda_j)+\epsilon_\ell, \qquad |\ell|<N,
\ee
where $\epsilon_\ell$'s are independent realizations of a sub-Gaussian random variable.
This discussion summarizes the corresponding results in \cite{mhaskar2024robust}.
Following \cite[Section~2.3]{boucheron2013concentration}, we recall that a real-valued random variable $X$ is called sub-Gaussian with parameter $V$ ($X\in \mathcal{G}(V)$) if $\mathbb{E}(X)=0$ and $\log \mathbb{E}(\exp(tX))\le tV^2/2$ for all $t>0$; e.g., if $X$ is a mean zero normal random variable with variance $V$, then $X\in \mathcal{G}(V)$. 
Likewise, any bounded random variable is sub-Gaussian.
We define a complex valued random variable $X$ to be in $\mathcal{G}(V)$ if its real and imaginary parts are in $\mathcal{G}(V)$.
\yadi{$V$}{Variance parameter of a sub-Gaussian random variable}
It is not difficult to see that if $X_1,\cdots,X_n$ are i.i.d., complex valued variables all in $\mathcal{G}(V)$, $\mathbf{a}=(a_1,\cdots,a_n)\in\CC^n$, $|\mathbf{a}|_n^2=\sum_{\ell=1}^n |a_\ell|^2$, then $\sum_{\ell=1}^n a_\ell X_\ell \in \mathcal{G}(|\mathbf{a}|_nV)$, and hence, that
\begin{equation}\label{eq:subgaussian_sum_tail}
\mathsf{Prob}\left(\left|\sum_{\ell=1}^n a_\ell X_\ell\right| >t\right)\le 4\exp\left(-\frac{t^2}{4|\mathbf{a}|_n^2V^2}\right), \qquad t>0.
\end{equation}

A crucial role in our theory is played by certain localized trigonometric polynomial kernels, which we now define.
Let $H:\RR\to [0,1]$ such that $H(t)=H(-t)$ for all $t\in\RR$, $H(t)$ is a constant if $|t|\le 1/2$, and $H(t)=0$ if $|t|\ge 1$. 
The actual choice is not important for our theory, but we will take $H$ to be fixed in the following discussion.
\yadi{$H$}{low pass filter}
\yadi{$\Phi_n$}{Localized trigonometric kernel}
\yadi{$\sigma_n$}{Reconstruction operator}

We define for 
\be\label{eq:lockerndef}
\hbar_n=\left\{\sum_{|\ell|<n}H\left(\frac{|\ell|}{n}\right)\right\}^{-1}\!\!, \quad\ \Phi_n(x)=\hbar_n\sum_{|\ell|<n}H\left(\frac{|\ell|}{n}\right)e^{\hati\ell x}, \qquad x\in\TT, \ n>0.
\ee
We note that the normalizing factor $\hbar_n$ is chosen so that
\begin{equation}\label{eq:phimax}
\max_{x\in \TT}|\Phi_n(x)|=\Phi_n(0)=1.
\end{equation}
An important property of $\Phi_n$ is the following localization estimate: For every $S\ge 2$, there exists $L=L(H,S)\ge 1$ such that 
\begin{equation}\label{eq:locest}
|\Phi_n(x)| \le \frac{L}{\max(1, (n|x|)^S)}, \qquad x\in\TT,\ n>0.
\end{equation}
Explicit expressions for $L$ in terms of $H$ and $S$ are given in \cite{singdet}.
The estimate \eqref{eq:locest} implies that $\hbar_n^{-1}\Phi_n(x)$ is an approximation to the Dirac delta $\delta_0$.

The fundamental idea behind our approach is the following. 
We define 
\be\label{eq:sigmadef}
\sigma_n(\tilde{\mu})(x)=\hbar_n\sum_{|\ell|<n}H\left(\frac{|\ell|}{n}\right)\tilde{\mu}(\ell)e^{\hati\ell x}, \quad \sigma_n(\mu)(x)=\hbar_n\sum_{|\ell|<n}H\left(\frac{|\ell|}{n}\right)\hat{\mu}(\ell)e^{\hati\ell x}, \qquad x\in\TT.
\ee
Then it is not difficult to verify that
\be\label{eq:funda_idea}
\sigma_n(\tilde{\mu})(x)=\sum_{j=1}^K A_j\Phi_n(x-\lambda_j)+E_n(x)=\sigma_n(\mu)(x)+E_n(x),
\ee
where 
\be\label{eq:noisesigma}
E_n(x)=\hbar_n\sum_{|\ell|<n}H\left(\frac{|\ell|}{n}\right)\epsilon(\ell)e^{\hati\ell x}, \qquad x\in\TT.
\ee
Because of the localization of $\Phi_n$, one can deduce that $\sigma_n(\mu)(x)$ is small when $x$ is away from each of the $\lambda_j$'s, and when it is large then $x$ must be close to exactly one of the $\lambda_j$'s. 
Moreover, $E_n(x)$ being a weighted average of the noise realizations is small for large enough $n$.
These observations are made precise in the following Theorem~\ref{theo:stationary_main}, proved in \cite{mhaskar2024robust}.
\color{black}
\begin{remark}\label{rem:stft}
{\rm
In analogy to Section~\ref{bhag:intro}, it is customary to define  the signal in the form
\be\label{eq:constsignal}
f(t)=\sum_{j=1}^K A_j\exp(-\hati t \nu_j)+\epsilon(t), \qquad t\in\RR,
\ee
and consider $\tilde{\mu}(\ell)$ to be samples of this signal at $\ell/R$, where $R$ is the sampling rate, and $\lambda_j=\nu_j/R$. 
It is then tempting to use the short term Fourier transform (STFT) with different windows, and think of \eqref{eq:sigmadef} to be discretization of the same.
From a purely mathematical point of view, this strategy ignores several errors. 
STFT is defined as an integral over $\RR$.
So, this approach would involve a discretization error as well as an aliasing error.
In \eqref{eq:constsignal}, the frequencies $\nu_j$ can be any real numbers, in \eqref{eq:sigmadef} with the notation as just described, one requires the $\lambda_j$'s to be in $(-\pi,\pi]$ - two $\lambda_j$'s which are equal modulo $2\pi$ cannot be distinguished. 
We may think of the procedure called periodization to get from the STFT to \eqref{eq:sigmadef}.
However, rather than getting into these complications, we have worked directly in the context of the theory of Fourier series, so that both of these errors simply don't arise.
\qed}
\end{remark}
\color{black}

In this paper, we will encounter several constants whose values are hard to estimate, but not important to describe the results. 
We make the following convention.

\vspace{2ex}

\noindent\textbf{Constant convention}

\vspace{2ex}

\textit{
The letters $c, c_1,\cdots$ will denote generic positive constants depending on $H$ and $S$ alone. Their values might be different at different occurrences, even within a single formula. 
The notation $A\lesssim B$ means $A\le c B$, $A\gtrsim B$ means $B\lesssim A$, and $A\sim B$ means $A\lesssim B\lesssim A$. 
Constants denoted by capital letters, such as $L$, $C$, etc. will retain their values.
}
\vspace{2ex}

Before stating the main theorem, we estimate the noise term $E_n(x)$ in the following lemma \cite[Lemma~4.1]{mhaskar2024robust}.
\begin{lemma}\label{lemma:noiselemma}
Let $\delta\in (0,1)$. 
There exist positive constants $C_1, C_2, C_3$, depending only on $H$ such that for $n\ge C_1(\ge 1)$, we have
\begin{equation}\label{eq:noiseest}
\mathsf{Prob}\left(\max_{x\in\TT}|E_n(x)| \ge C_2V\sqrt{\frac{\log (C_3n/\delta)}{n}}\right)\le \delta.
\end{equation}
\end{lemma}

In order to state the main theorem of this section, we need some further notation.
Let
\begin{equation}\label{eq:notation}
\mathbf{M}=\sum_{j=1}^K |A_j|, \ \ \mathfrak{m}=\min_{1\le j\le K}|A_j|, \ \ \eta=\min_{\ell\not=k} |\lambda_j-\lambda_\ell|.
\end{equation}
\yadi{$\mathbf{M}$}{Sum of absolute values of amplitudes}
\yadi{$\mathfrak{m}$}{minimal absolute value of amplitudes}
\yadi{$\eta$}{minimal separation among frequencies}
Our main theorem in this section can now be stated as follows.
\begin{theorem}\label{theo:stationary_main}
Let $0<\delta<1$,
\begin{equation}\label{eq:levelsetdef}
\mathbb{G}=\{x\in [-\pi,\pi] : |\sigma_n(\tilde{\mu})(x)|\ge \mathfrak{m}/2\},
\end{equation}
and (cf. \eqref{eq:locest}, \eqref{eq:notation})
\begin{equation}\label{eq:thresholdCdef}
C =  \left(\frac{16\mathbf{M}L}{\mathfrak{m}} \right)^{1/S}.
\end{equation}
Let $n$ be large enough so that the estimate \eqref{eq:noiseest} holds, and in addition, $n\ge 4C/\eta$.
Then
each of the following statements holds with probability exceeding $1-\delta$.
\begin{itemize}
\item (\textbf{Disjoint union condition}) \\
The set $\mathbb{G}$ is a disjoint union of exactly $K$ subsets $\mathbb{G}_\ell$,
\item (\textbf{Diameter condition}) \\
For each $\ell=1,\cdots, K$,   $\mathsf{diam}(\mathbb{G}_\ell) \le 2C/n$,
\item (\textbf{Separation}) \\
$\mathsf{dist}(\mathbb{G}_\ell, \mathbb{G}_j) \ge \eta/2$ for $\ell \neq k$,
\item (\textbf{Interval inclusion}) \\
For each $\ell=1,\cdots, K$,  $I_\ell=\{x\in\TT: |x-\lambda_\ell|\le 1/(4n)\}\subseteq \mathbb{G}_\ell$.
 
\end{itemize}
Moreover, if
\begin{equation}\label{eq:lambda_estimator}
\hat{\lambda}_\ell =\arg\max_{x\in \mathbb{G}_\ell}|\sigma_n(x)|,
\end{equation}

then

\begin{equation}\label{eq:lambdaerr}
|\hat{\lambda}_\ell-\lambda_\ell| \le 2C/n.
\end{equation}

\end{theorem}

\subsection{Reduction to the constant parameter case}\label{bhag:reduction}

In this section, we describe certain conditions under which the data of the form \eqref{eq:nonstationary} can be reduced to the constant parameter data as in Section~\ref{bhag:stationary}.
Thus, we are interested in finding the instantaneous frequencies $\phi_j'(t^*)$ and amplitudes $A_j(t^*)$ in the signal
\be\label{eq:nonstationarybis}
f(t)=\sum_{j=1}^K f_j(t)=\sum_{j=1}^K A_j(t)\exp(\hati\phi_j(t)),
\ee
for some point $t^*\in\RR$, where $A_j$'s are continuous and $\phi_j$'s are continuously differentiable functions on an interval of the form $[t^*-\Delta, t^*+\Delta]$.
We introduce the notation (abusing the notation introduced in \eqref{eq:notation})
\be\label{eq:notation2}
\mathbf{M}(t^*)=\sum_{j=1}^K |A_j(t^*)|, \qquad B(t^*)=\max_{1\le j\le K}|\phi_j'(t^*)|.
\ee
The following theorem is a reformulation of \cite[Theorem~4.2]{bspaper}. 
We will reproduce the proof for the sake of completeness.
\yadi{$\Delta$}{half-length of time interval for snippet}
\yadi{$t_k$}{center of the time interval for snippet}
\yadi{$I_k$}{time interval for snippet $[t_k-\Delta,t_k+\Delta]$}
\yadi{$R$}{sampling frequency}

\begin{theorem}\label{theo:reduction}
Let $t^*\in\RR$, $\alpha, \Delta>0$, and for $|u|\le \Delta$,
\be\label{eq:slowvary}
|A_j(t^*+u)-A_j(t^*)| \le \alpha |u||A_j(t^*)|, \qquad |\phi_j'(t^*+u)-\phi_j'(t^*)| \le \alpha |u||\phi_j'(t^*)|.
\ee 
Then we have
\be\label{eq:reduction}
\left|f(t^*+u)-\sum_{j=1}^K f_j(t^*)\exp\left(\hati \phi_j'(t^*)u\right)\right| \le \alpha \mathbf{M}(t^*)(B(t^*)\Delta +1)\Delta, \qquad |u|\le |\Delta|.
\ee
\end{theorem}

\begin{proof}\ 
In this proof, we write $\mathbf{M}=\mathbf{M}(t^*)$ and $B=B(t^*)$. Let $|u|\le \Delta$. 
We observe that
\be\label{eq:pf1eqn1}
\begin{aligned}
\left|\sum_{j=1}^K A_j(t^*+u)\right.&\left.\exp(\hati \phi_j(t^*+u))-\sum_{j=1}^K A_j(t^*)\exp(\hati \phi_j(t^*))\exp(\hati u \phi_j'(t^*))\right|\\
&\le \left|\sum_{j=1}^K \left(A_j(t^*+u)-A_j(t^*)\right)\exp(\hati \phi_j(t^*+u))\right|\\
&\qquad +\left|\sum_{j=1}^K A_j(t^*)\left(\exp(\hati \phi_j(t^*+u))-\exp(\hati \phi_j(t^*))\exp(\hati u \phi_j'(t^*))\right)\right|
\end{aligned}
\ee
In view of the fact that
$$
|e^{ix}-e^{iy}|=2|\sin((x-y)/2)| \le |x-y|,
$$
and \eqref{eq:slowvary}, we deduce that
 for each $k$,
\be\label{eq:pf1eqn2}
\begin{aligned}
\left|\exp(\hati \phi_j(t^*+u))\right.&\left.-\exp(\hati \phi_j(t^*))\exp(\hati u \phi_j'(t^*))\right|\le \left|\phi_j(t^*+u)-\phi_j(t^*)-u \phi_j'(t^*)\right|\\
& \le \left|\int_0^u \left(\phi_j'(t^*+v)-\phi_j'(t^*)\right)dv\right|\le \alpha B\left|\int_0^u |v|dv\right|\le \alpha B|u|^2\le \alpha B\Delta^2.
\end{aligned}
\ee
Using \eqref{eq:slowvary} again, we see that
$$
\left|\sum_{j=1}^K \left(A_j(t^*+u)-A_j(t^*)\right)\exp(\hati \phi_j(t^*+u))\right|\le \mathbf{M}\Delta.
$$
Together with \eqref{eq:pf1eqn2} and \eqref{eq:pf1eqn1}, this leads to \eqref{eq:reduction}.
\end{proof}

\subsection{Signal separation operator}\label{bhag:sso}
In this section, we fix $t^*\in\RR$.
The quantities denoted below by $\alpha$, $\Delta$, $R$ are independent of $t^*$, and so are the constants involved in $\ls$, $\gs$, and $\sim$, but the other quantities will depend upon $t^*$ without its mention in the notation.

Let $R >0$ denote the sampling frequency.  We write
\be\label{eq:reducedfreq}
\lambda_j=\phi_j'(t^*)/R, \quad n=\lfloor R\Delta\rfloor, \quad \mu=\sum_{j=1}^K f_j(t^*)\delta_{\lambda_j}, \quad \hat{\mu}(\ell)=\sum_{j=1}^K f_j(t^*)\exp(-\hati \ell\lambda_j).
\ee
Then using Theorem~\ref{theo:reduction}, it is easy to deduce that (cf. \eqref{eq:sigmadef})
\be\label{eq:ssointro}
\left|\hbar_n\sum_{|\ell|<n} H\left(\frac{|\ell|}{n}\right)f(t^*-\ell/R)e^{\hati \ell x}-\sigma_n(\mu)(x)\right|\le \alpha \mathbf{M}(t^*)(B(t^*)\Delta +1)\Delta.
\ee
This leads to the following definition (cf. \cite[Definition~2.3]{bspaper})
\begin{definition}\label{def:ssodef}
Let $F :\RR\to \CC$, $t^*\in\RR$, $R>0$.
 We define the \textbf{\textit{Signal Separation Operator}} SSO by
\be\label{eq:ssodef}
\mathcal{T}_{n,R}(F)(t^*;x)=\hbar_n\sum_{|\ell|<n} H\left(\frac{|\ell|}{n}\right)f(t-\ell/R)e^{\hati \ell x}, \qquad x\in\RR.
\ee
\yadi{$\mathcal{T}_{n,R}$}{Signal separation operator}
\end{definition}
We will use the operator to separate the components $f_j(t^*)$ and the instantaneous frequencies $\phi_j'(t^*)$ based on the noisy observations
\be\label{eq:nonstationary_noisy}
F(t)=\sum_{j=1}^K f_j(t) +\epsilon(t)=\sum_{j=1}^K A_j(t)\exp(i\phi_j(t))+\epsilon(t),
\ee
where for each $t$, $\epsilon(t)$ is a realization of a complex sub-Gaussian random variable in $\mathcal{G}(V)$.
Theorem~\ref{theo:stationary_main} can be translated into Theorem~\ref{theo:main} below, where we abuse the notation to make the comparison easier. 
We note that  Theorem~\ref{theo:main} itself is a restatement of \cite[Theorems~2.4,4.1]{bspaper} in a somewhat modified form taking the noise into account.
In order to state this theorem, we abuse the notation in Section~\ref{bhag:stationary} again, and write (in addition to \eqref{eq:notation2})
\begin{equation}\label{eq:notation3}
 \ \mathfrak{m}(t^*)=\min_{1\le \ell\le K}|f_\ell(t^*)|, \ \ \eta(t^*)=\min_{\ell\not=k} |\phi_\ell'(t^*)-\phi_k'(t^*)|, \ \eta^*=\eta(t^*)/R.
\end{equation}
For brevity, we will omit the mention of $t^*$ from our notation, unless we feel that this might cause some confusion.
\begin{theorem}\label{theo:main}
Let $0<\delta<1$, $t^*\in\RR$, $R, \Delta>0$, and the assumptions of Theorem~\ref{theo:reduction} be satisfied. 
We assume that  $n=\lfloor R\Delta\rfloor$ satisfies the condition \eqref{eq:noiseest},  and in addition, $n\ge 4C/\eta^*$. 
We assume further that
\be\label{eq:Deltacond}
\alpha \mathbf{M}(B\Delta +1)\Delta\le \mathfrak{m}/4.
\ee
We define
\be\label{eq:levelsetbis}
\mathbb{G}(t^*)=\{x\in\TT : |\mathcal{T}_{n,R}(F)(t^*;x)|\ge 3\mathfrak{m}/4\}.
\ee
Then
each of the following statements holds with probability exceeding $1-\delta$.
\begin{itemize}
\item (\textbf{Disjoint union condition}) \\
The set $\mathbb{G}(t^*)$ is a disjoint union of exactly $K$ subsets $\mathbb{G}_\ell(t^*)$,
\item (\textbf{Diameter condition}) \\
For each $\ell=1,\cdots, K$,   $\mathsf{diam}(\mathbb{G}_\ell(t^*)) \le 2C/n$,
\item (\textbf{Separtion}) \\
$\mathsf{dist}(\mathbb{G}_\ell(t^*), \mathbb{G}_j(t^*)) \ge \eta/2$ for $\ell \neq k$,
\item (\textbf{Interval inclusion}) \\
For each $\ell=1,\cdots, K$,  $I_\ell=\{x\in\TT: |x-\lambda_\ell|\le 1/(4n)\}\subseteq \mathbb{G}_\ell(t^*)$.
\end{itemize}
Moreover, if
\begin{equation}\label{eq:lambda_estimator_nonstationary}
\hat{\lambda}_\ell =\arg\max_{x\in \mathbb{G}_\ell(t^*)}|\mathcal{T}_{n,R}(F)(t^*;x)|,\quad \widehat{\phi_\ell'(t^*)}=R\hat{\lambda}(\ell),
\end{equation}
then
\begin{equation}\label{eq:lambdaerr_nonstationary}
|\hat{\lambda}_\ell-\lambda_\ell| \le 2C/n.
\end{equation}
\end{theorem}

\begin{remark}\label{rem:ncond}
{\rm
For the convenience of the reader, we summarize the conditions on the sampling frequency $R$ below.
We assume  that for every $t^*$ of interest,
\be\label{eq:rcond1}
R\Delta \ge \max\left(C_1, 4C/\eta^*\right)
\ee
and $R$ is large enough so that with probability $>1-\delta$,
\be\label{eq:recond2}
|E_n(x)|\le C_2V\sqrt{\frac{(\log(C_2R\Delta)/\delta)}{R\Delta}}
\ee
We note that the condition \eqref{eq:rcond1} puts a lower bound on $\Delta$, while the condition \eqref{eq:Deltacond} is an upper bound on the Lipschitz constant $\alpha$. 
In particular, the choice of $\Delta$ is a theoretically delicate one.
In our experiments, this was done by experimentation.
\qed}
\end{remark}
\begin{remark}\label{rem:chirps}
{\rm
In the case of a chirp
$$
\phi'(t)=\begin{cases}
\omega+\frac{B}{d}(t-\gamma), &\mbox{ if $t_0\le t\le t_0+d$}\\
0, &\mbox{otherwise},
\end{cases}
$$
 we have for $t^*\in [t_0, t_0+d]$
$$
|\phi'(t^*+u)-\phi'(t^*)|=\left|\frac{B}{d}\right||u|, \qquad |\phi'(t^*)|=\left|\omega+\frac{B}{d}(t^*-\gamma)\right|
$$
So, we may take 
$$
\alpha=\frac{|B/d|}{\min_{t\in [t_0, t_0+d]}|\omega+(B/d)(t-\gamma)|}.
$$

\qed}
\end{remark}

\bhag{Algorithms}\label{bhag:algorithm}
In this section, we describe how the SSO can be used for signals of the form \eqref{eq:groundtruth}.
The basic idea behind our algorithms is to apply Theorem~\ref{theo:main} to different snippets of the signal. 
Thus, we choose a small $\Delta$ and divide the signal duration into $D$ overlapping time  intervals $I_k=[t_k-\Delta, t_k+\Delta]$, $k=1,\cdots, D$.
The $t_k$'s are chosen such that $I_k \cap I_{k+1} \neq \emptyset$. This is to ensure that all the sampling data are considered in this experiment.
Theorem~\ref{theo:main} gives guidelines on how to select $\Delta$. 
In practice, one has to treat this as a tunable parameter.
In our experiments reported here, we chose $\Delta$ by experimentation so that it is large enough to eliminate the noise effect, but small enough to give the accurate estimation. 
A large value of $\Delta$ implies that there are more data points to average out the noise, but this has higher chance that the signal will start or end within the interval which leads to inaccurate estimation.

  The signal restricted to $I_k$ has the form (cf. \eqref{eq:groundtruth}):
\be\label{eq:ground_truth_reformed}
F_k(t)=\sum_{j=1}^{J_k}f_{j,k}(t)+\epsilon(t), \qquad f_{j,k}(t)=A_{j,k}(t)\exp(\hati\phi_{j,k}(t)), 
\ee
 where $\epsilon(t)$ represents the noise and
 \be\label{eq:phase_bis}
  \phi_{j,k}(t)=\omega_{j,k}(t-\gamma_{j,k})+a_{j,k}(t-\gamma_{j,k})^2/2.
\ee
\yadi{$\omega_{j,k}$, $\gamma_{j,k}$, $a_{j,k}$}{Chirp parameters for chirp $j$ in $I_k$}
\yadi{$\Lambda_{j,k}$}{estimated IF for chirp $j$ at $t_k$}
\yadi{$\tau$}{threshold for SSO, estimate for $3\mathfrak{m}/4$}
The SSO will be used with each $F_k$ to obtain the parameter values for each $f_{j,k}$. 
In our experiments, we have focused on IFs, assuming all the amplitudes to be $1$.
Taken together, the set $(t_k,\widehat{\phi_{j,k}'(t_k)})$ (where $\widehat{\phi_{j,k}'(t_k)}$ is the estimated value for $\phi_{j,k}'(t_k)$) is called the \textbf{SSO diagram}.
Once we determine the IFs, the amplitudes can be determined by simply substituting the estimated IFs in SSO and more accurately by solving an appropriate least square problem. 
We postpone a more detailed investigation of this question to a future work.


The clusters defined in Theorem~\ref{theo:main} are found, using a clustering algorithm.  For example, in this paper we used DBSCAN \cite{ester1996density}.
The notation we use in our algorithm is 
\be\label{eq:dbscabdef}
([\mbox{scatter data}], [\mbox{label}]) \leftarrow\mbox{DBSCAN}([\mbox{scatter data}], [\mbox{neighborhood radius}], [\mbox{minimum number of neighbors}]).
\ee
Here, $[\mbox{scatter data}]$ is a data matrix,  and the algorithm attaches a label $1$ for the cluster of interest (and $-1$ otherwise), based on the number of points neighboring the so called core points in a ball of radius denoted above by $[\mbox{neighborhood radius}]$.

Our algorithm needs a few more tunable parameters as summarized in Table \ref{tab:tunable_parameters}.
In Section~\ref{bhag:tunable}, we describe several heuristics for the selection of these parameters, although we point out that it would be impossible to determine these parameters from a single observation of a totally blind source signal without any insight into the ground truth parameters involved.
In this section, we illustrate the basics of the algorithms using a somewhat ad hoc choice of these parameters.
\color{black}
\begin{table}[H]
\begin{center}
\begin{tabular}{ |c|l| } 
 \hline
 Parameters & Description \\
 \hline

 {$\eta$} & {Minimal separation among frequencies}\\
 {$B_{rec}$} &{ Receiver bandwidth=$\max_{j,k,t}\phi_{j,k}'(t)-\min_{j,k,t}\phi_{j,k}'(t)$}\\
 \color{black}
 $\Delta$ & Time interval width centered at $t_k$, which is chosen based on Theorem \ref{theo:main}\\
 $D$ & Number of intervals centered at $t_k$, which is chosen so that interval $I_k$'s are overlapping\\
 $D_1$ & Minimum number of neighbors, which can be tuned based on \\
 & the number of points in the SSO diagram in the radius given by $B_{\mbox{\scriptsize{rec}}}$.\\
 $D_2$ & Minimum number of neighbors, which can be tuned based on \\
 & the number of points in the smallest signal cluster\\
$M$ & Number of partitions when signal crossover detected, which can be tuned so that \\
& some partitions contain sufficiently long linear chirps \\
 \hline
\end{tabular}
\end{center}
 	\caption{The tables above shows a list of tunable parameters.} \label{tab:tunable_parameters}
\end{table}

Algorithm~1 is designed to find the IFs. 
A challenge here is to determine the correct threshold related to the unknown minimal amplitude $\mathfrak{m}$ in the presence of noise.
A value which is too small would indicate superfluous frequencies, a value which is too large would miss some frequencies. This effect is discussed in our paper \cite{mhaskar2024robust}.
There are other details such as detecting the clusters and find the right peaks which are described below.
Part 2 will estimate the parameters of the instantaneous frequencies obtained by part 1 using least squares.
The challenge here is to figure out where the different signal components begin and end, and which part of the signal to use in order to set up the least square problem in order to obtain reliable solutions.
SSO does not work directly when there are crossover frequencies.
In fact, the set up itself does not make sense in general in such cases; e.g., whether the crossover of the form $X$ represents two linear frequencies crossing or two $V$-shaped frequencies touching each other.
 Part 3 is the refinement process where we use the facts that SSO is inherently unable to resolve crossovers and that we are dealing with linear chirps
  to detect the presence of such frequencies.
 If sufficiently long linear chirps are present on both sides of the crossover, then we can detect those using the previous two algorithms with a smaller value of $\Delta$ and interpolate to determine where the crossover happens and which chirps cross each other.

 Figure~\ref{fig:flowchart1} and  \ref{fig:flowchart2} give flowcharts of how the algorithms are related to each other.
 The whole process will be illustrated step-by-step with one example in this section, and analyzed further in Section~\ref{bhag:numerical_results} with a few other examples. 
 
We illustrate our method using the signal with parameters as described in Table~\ref{tab:signal_table}. 
In the table, the sign of the chirp term is denoted by $\{+, -\}$ to indicate up-chirp and down-chirp, respectively.  
Furthermore, signals with a 0 PRI refer to a signal with only a single pulse in the analyzed time window. 

\begin{table}[ht]
\begin{center}
\begin{tabular}{ |c|c|c|c|c|c|c|c|c| } 
 \hline
 Signal$_j$ & $A_j$ & $\omega_j$ & $B_j$ & $d_j$ & $t^*_{0,j}$ & $\mbox{PRI}_j$ & Chirp Sign & Number of Pulses \\
 \hline
 1 & 1 & 1080000000 & 15000000 & 3.0e-05 & 1.0e-05 & 5.0e-05 & + & 2 \\
 2 & 1 & 1360000000 & 5000000 & 1.0e-05 & 1.0e-05 & 1.5e-05 & + & 5 \\
 3 & 1 & 1540000000 & 20000000 & 3.0e-05 & 1.0e-05 & 0 & - & 1 \\
 4 & 1 & 1510000000 & 50000000 & 7.0e-05 & 1.5e-05 & 0 & + & 1 \\
 5 & 1 & 1480000000 & 15000000 & 5.0e-05 & 3.0e-05 & 0 & - & 1 \\
 6 & 1 & 1040000000 & 15000000 & 3.0e-05 & 1.5e-05 & 4.0e-05 & - & 2 \\
 \hline
\end{tabular}
\end{center}
 	\caption{The tables above shows the signal ground truth parameters, noting that the 7 realizations used for the numerical experiments use the same signal parameters with varying $\omega$.} \label{tab:signal_table}
\end{table}


In this example, we use $\epsilon(t)$ to be a Gaussian noise at -10 dB SNR, and the sampling rate of 1GHz.

\begin{figure}[H]
\begin{center}
\includegraphics[width=.65\textwidth]{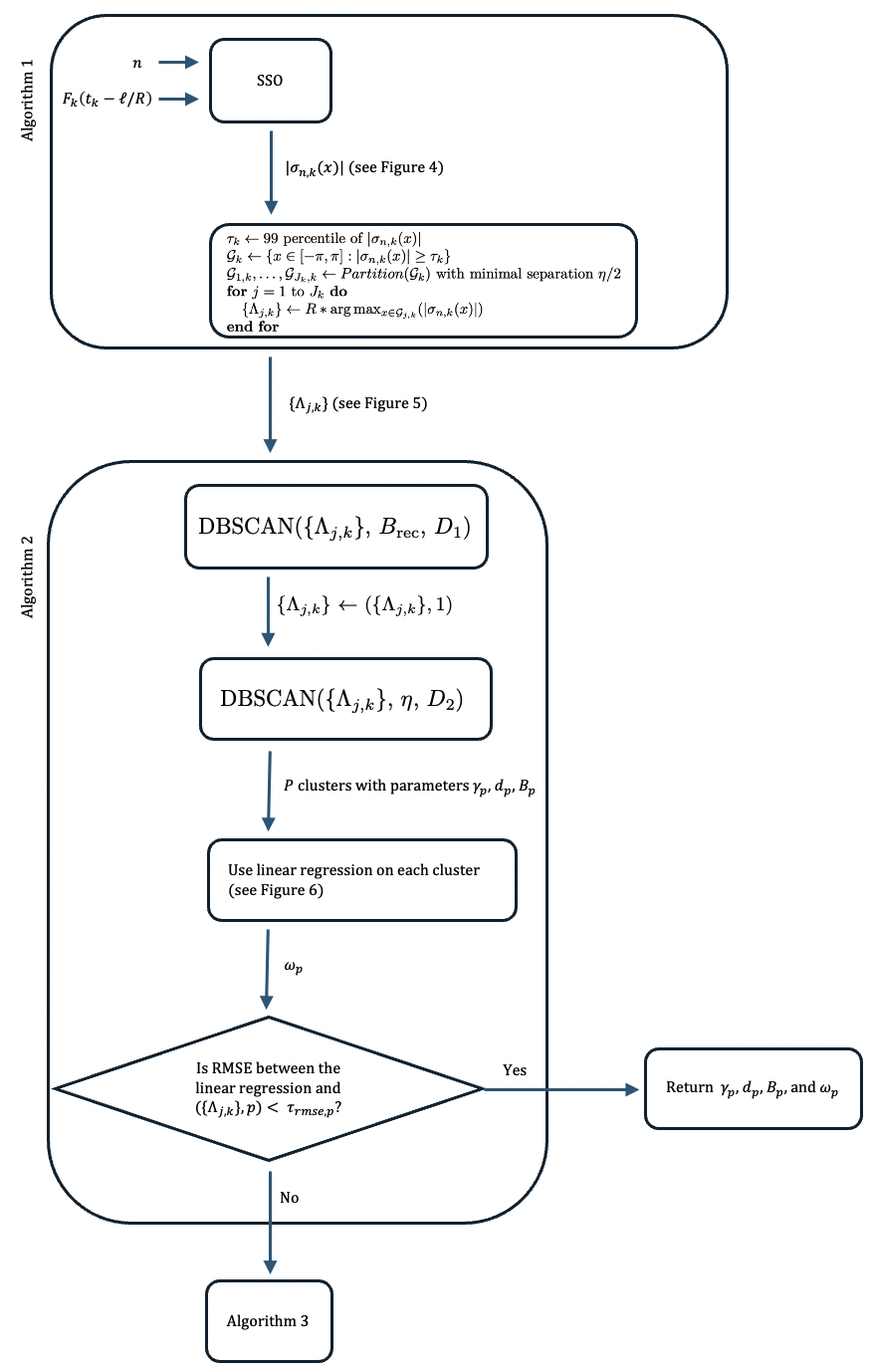}
\end{center}
\caption{The flowchart diagram for our algorithm scheme.}
\label{fig:flowchart1}
\end{figure}

\begin{figure}[H]
\begin{center}
\includegraphics[width=.65\textwidth]{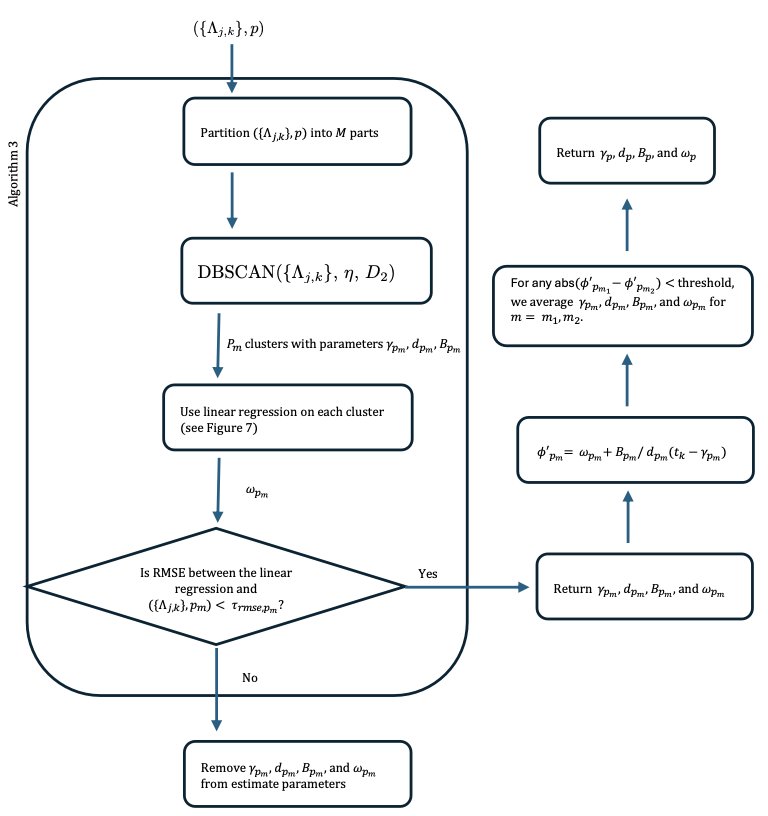}
\end{center}
\caption{The flowchart diagram for our algorithm scheme.}
\label{fig:flowchart2}
\end{figure}

 The visualizations of  the raw signal $f(t)$ and its instantaneous frequency $\phi'(t)$ corresponding to the parameters in Table~\ref{tab:signal_table} are shown in Figure \ref{fig:signal_ground_truth}.
\begin{figure}[ht]
\begin{center}
\includegraphics[scale=.14]{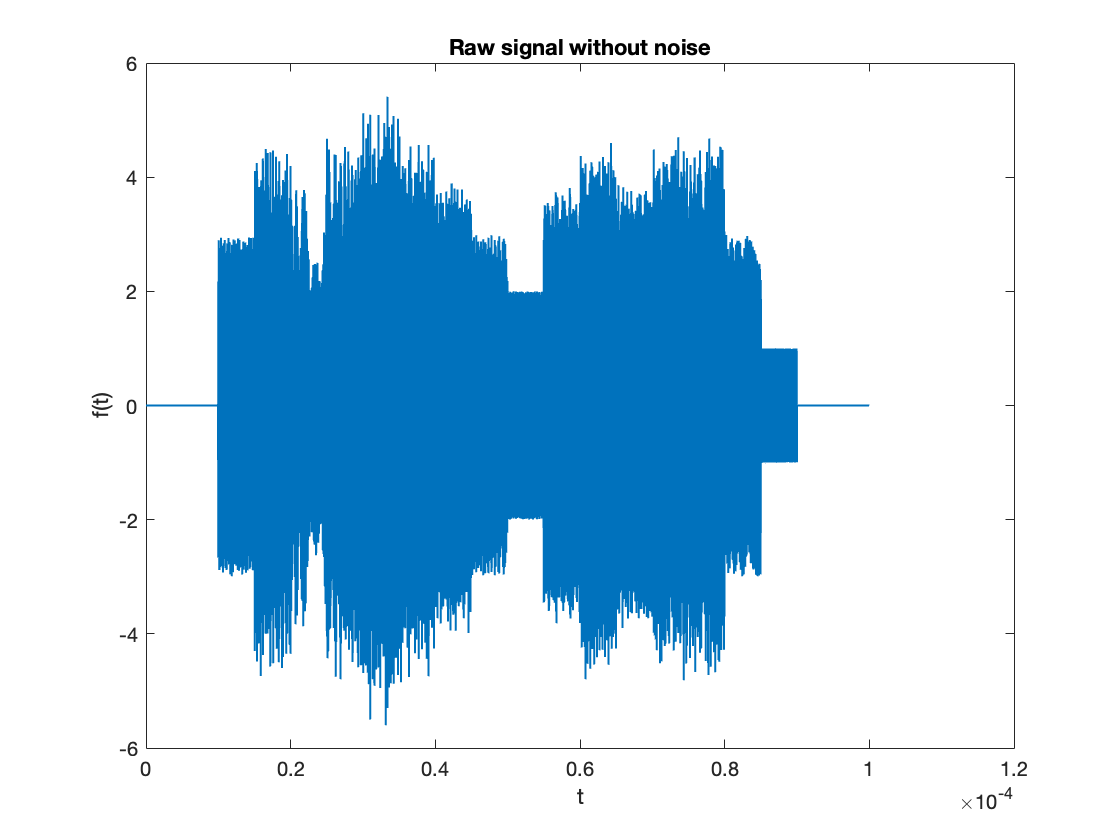}
\includegraphics[scale=.14]{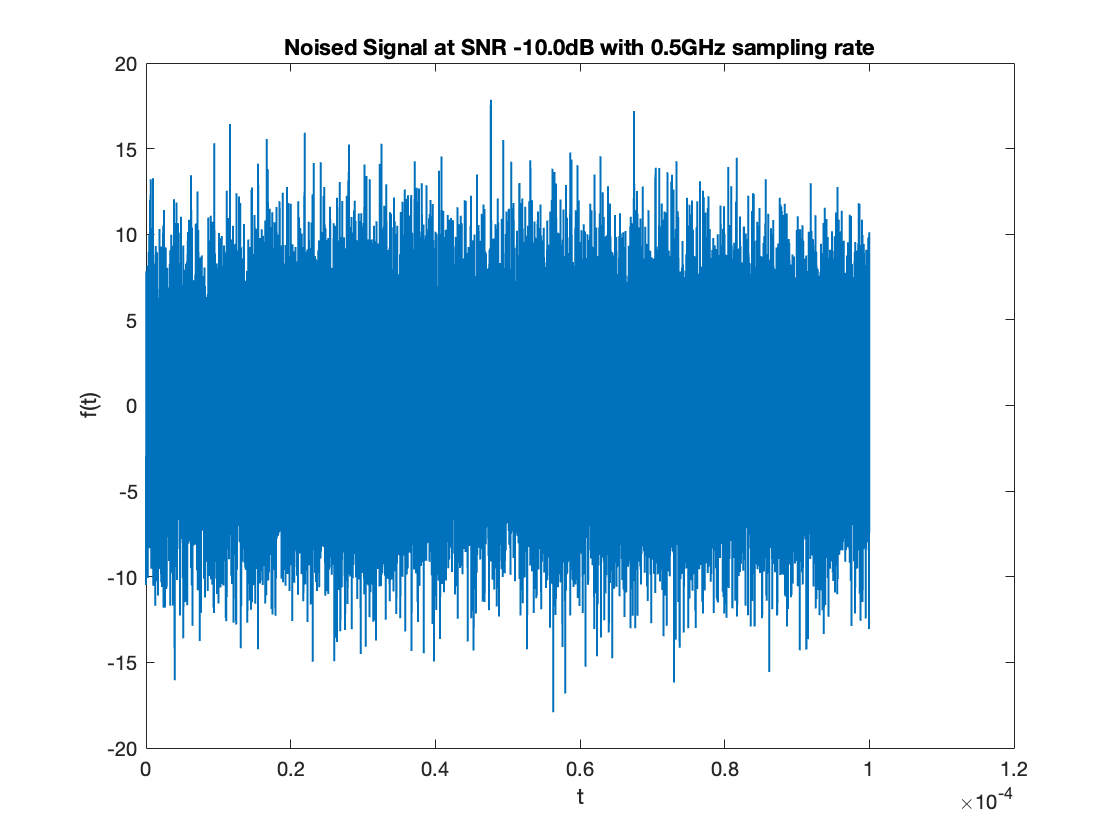}
\includegraphics[scale=.14]{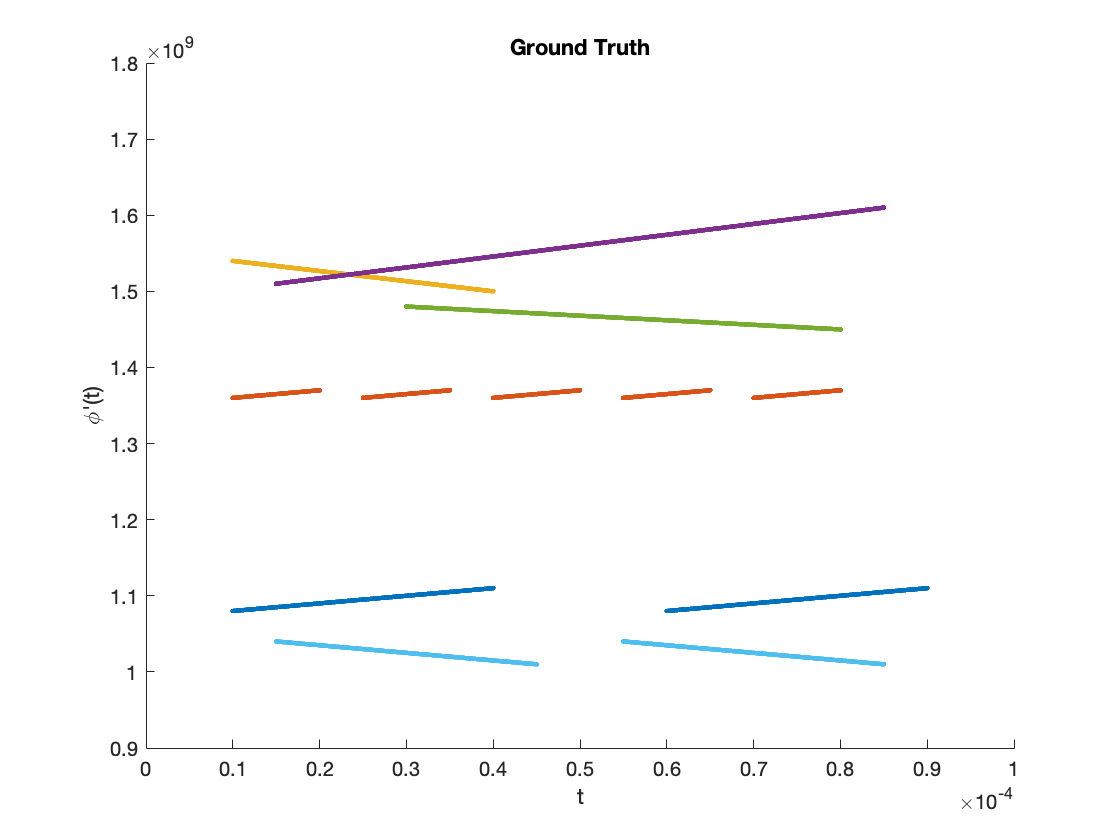}
\end{center}
\caption{(Left) The raw signal $f(t)$ as given in \eqref{eq:nonstationary} at sampling rate $1$ GHz. (Middle) The noised signal $F(t)=f(t)+\epsilon(t)$ at -10 dB SNR. (Right) The ground truth of $\phi'(t)$ from \eqref{eq:pulsephasedef}.}
\label{fig:signal_ground_truth}
\end{figure}

%

\vspace*{-0.5cm}
\subsection{Finding instantaneous frequencies}\label{bhag:alg1}

In this section, we describe the first stage of the algorithm, Algorithm~1, illustrated in Figure \ref{fig:flowchart1}, where we obtain the piecewise constant approximation to the the instantaneous frequencies in each snippet {$\{F_k(t_k - \ell/R)\}$}  for $k=1,\ldots,D$ (cf. \eqref{eq:ground_truth_reformed}) using the SSO for each such snippet.
{We note that the point $t_k$ and $\Delta$ are chosen so that all the points $t_k-\ell/R$ are the points at which the signal is sampled.} 
The main challenge here is to determine the threshold in the definition \eqref{eq:levelsetbis}.
Since the value of the minimum amplitude $\mathfrak{m}$ is unknown, we run SSO on {$\{F_k(t_k - \ell/R)\}$} and set the threshold $\tau_k$ to be $99.9$ percentile of the power spectrum.
The percentile is the same for every interval $I_k$, but the actual value of the threshold $\tau_k$ will be different depending on the sampling frequency and the SNR.
 \begin{algorithm}[H]
 \begin{algorithmic}[1]
 \item[{\rm a)}] \textbf{Input:} Minimal separation $\eta$ and signal {$\{F_k(t_k - \ell/R)\}$}.
 \item[{\rm b)}] \textbf{Output:} $\{\Lambda_{j,k}\}$
 \STATE $\hbar_n \gets \left\{\sum_{|\ell|<n}H\left(\frac{|\ell|}{n}\right)\right\}^{-1}$
 \STATE $\sigma_{n,k}(x) \gets \hbar_n \sum_{|\ell|<n} H\left(\frac{|\ell|}{n}\right)F_k(t_{k}-\ell/R)e^{i\ell x}$
 \STATE $\tau_k \gets 99.9 \mbox{ percentile of } |\sigma_{n,k}(x)|$
 \STATE $\mathcal{G}_k \gets \{x\in [-\pi,\pi] : |\sigma_{n,k}(x)|\ge \tau_k\}$
 \STATE $\mathcal{G}_{1,k}, \ldots, \mathcal{G}_{J_k,k} \gets Partition(\mathcal{G}_k) \mbox{ with minimal separation } \eta/2$
  \FOR{$j=1$ to $J_k$}
 \STATE $\{\Lambda_{j,k}\} \gets R * \arg\max_{x\in \mathcal{G}_{j,k}} (|\sigma_{n,k}(x)|)$
 \ENDFOR
 \STATE \textbf{Return: } $\{\Lambda_{j,k}\}$ for $j=1,\ldots,J_k$
 \caption{Signal separation operator (SSO)}
 \end{algorithmic}
 \label{alg:univariate}
 \end{algorithm}

\begin{figure}[H]
\begin{center}
\includegraphics[scale=.14]{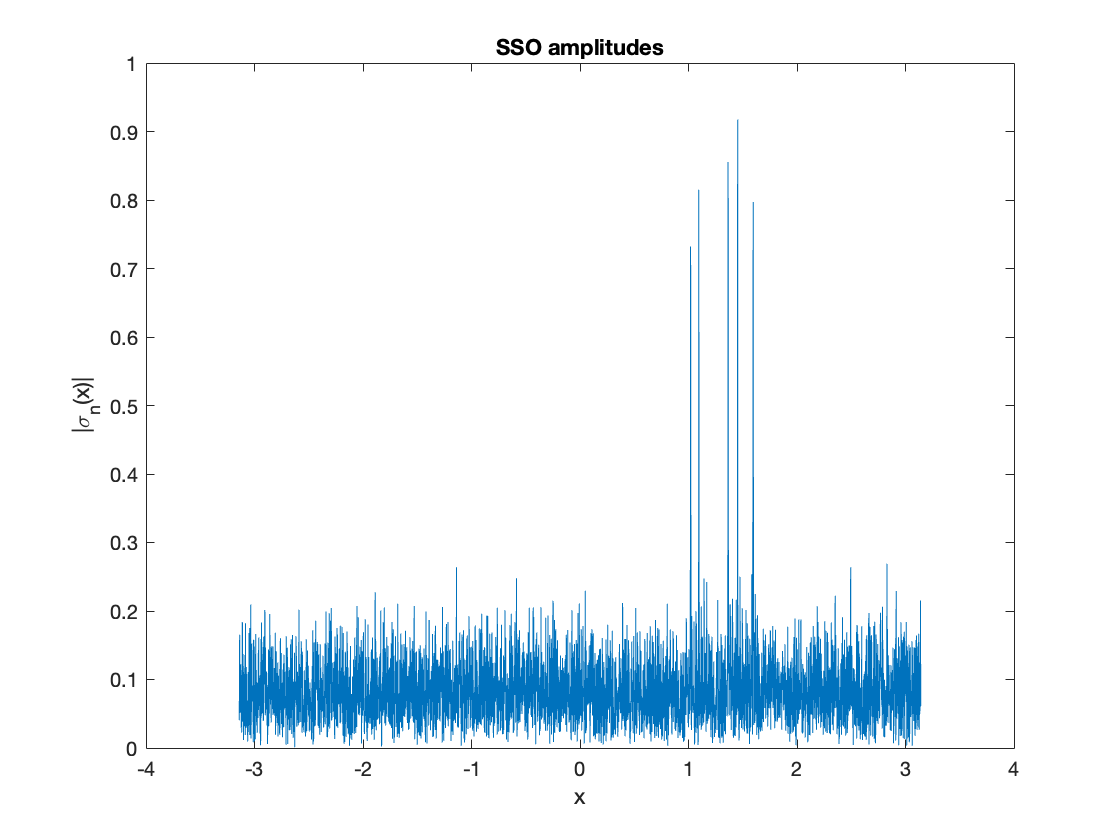}
\includegraphics[scale=.14]{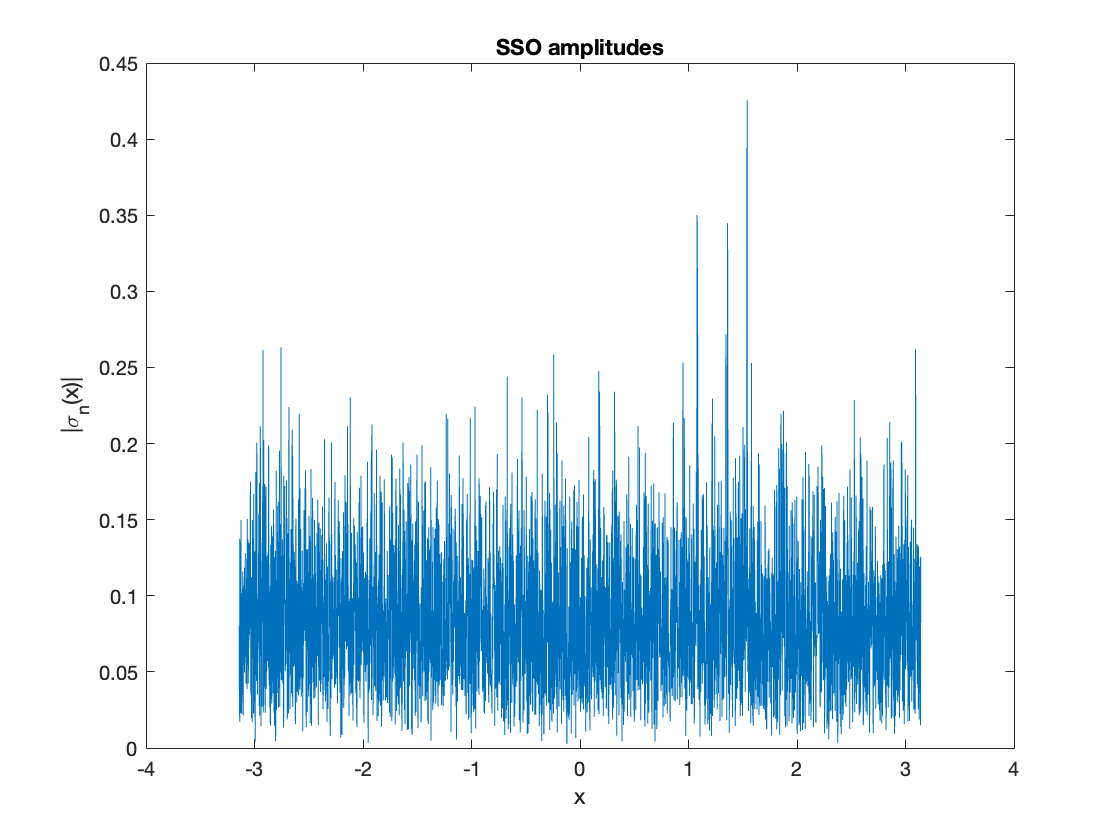}
\includegraphics[scale=.14]{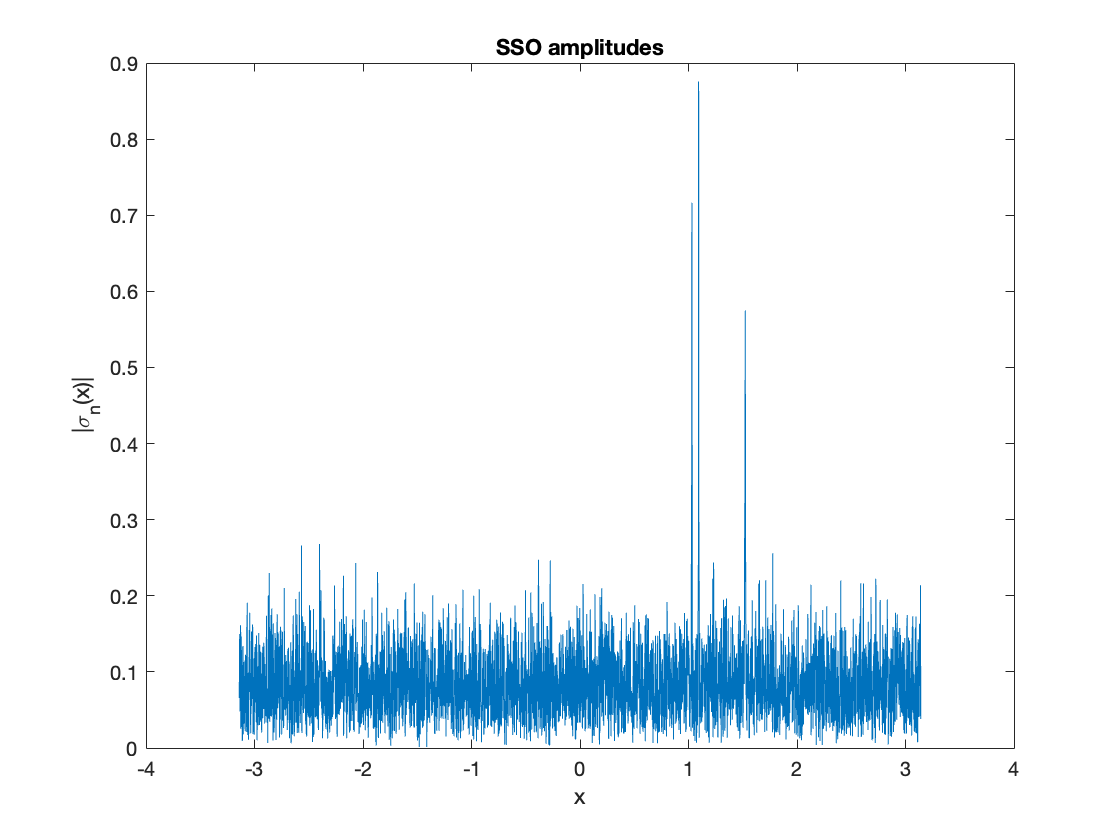}
\includegraphics[scale=.14]{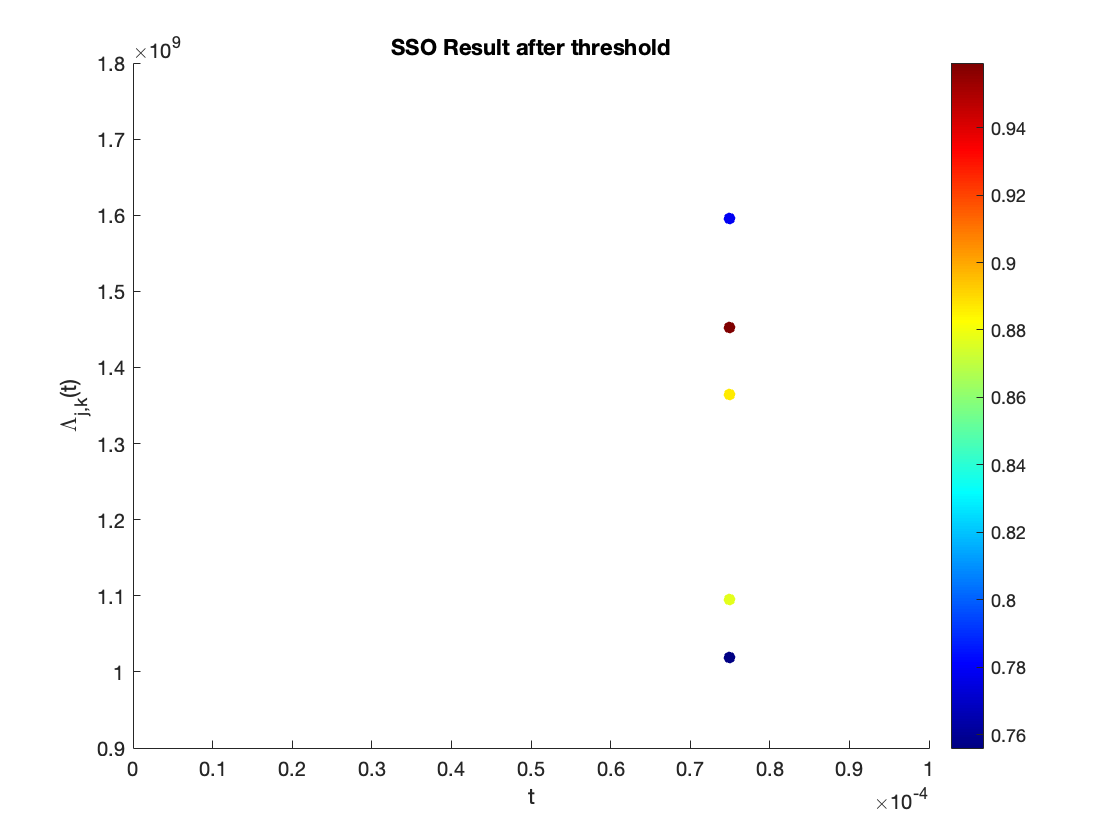}
\includegraphics[scale=.14]{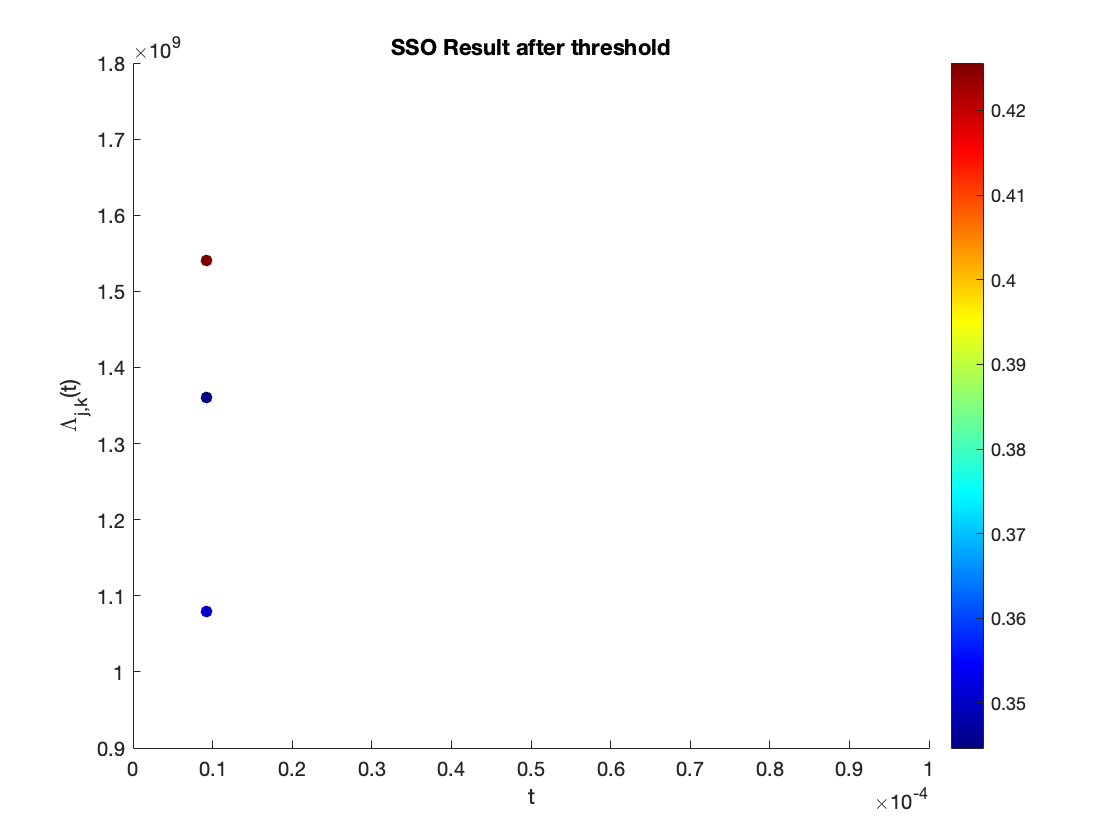}
\includegraphics[scale=.14]{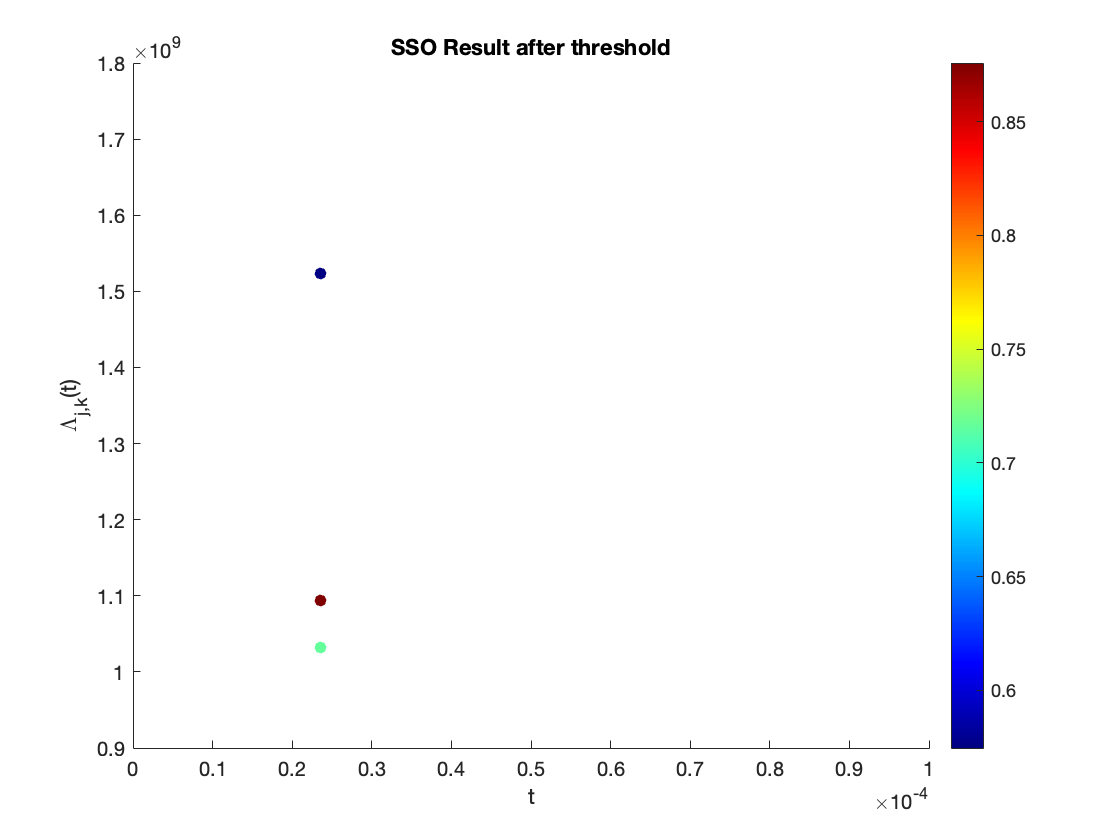}
\end{center}
\caption{(Top) $|\sigma_{n,k}(x)|$ for the interval $I_k$. (Bottom) SSO results at $t_k$.
(Left) The plots for $t_k =7.5 \times 10^{-5}$ where the signals pass from the beginning through the end of the interval $I_k$. (Middle) The plots for $t_k =1 \times 10^{-5}$ where the signal starts within the interval $I_k$. (Right) The plots for $t_k =2.35 \times 10^{-5}$ where there are 2 signals crossover within the interval $I_k$. }
\label{fig:freqattk}
\end{figure}

The output of this step is the estimated values of the instantaneous frequencies at $t_k$:
$$
\Lambda_{j,k}=\omega_{j,k}+a_{j,k}(t_k-\gamma_{j,k}), \qquad j=1,\cdots, J_k.
$$
This is illustrated in  Figure~\ref{fig:freqattk}.
We see that the frequencies are determined quite accurately when there is no discontinuity in the signals, and the minimal separation is large, but there is a problem when a signal starts/stops somewhere within $I_k$, resulting in a discontinuity or when there is a cross-over frequency, resulting in a small minimal separation.

\subsection{Parameter estimation}\label{bhag:paramest}
The starting point of this step is the results of all the instantaneous frequencies in all the snippets.  The resulting diagram is called  the raw SSO diagram as shown in Figure~\ref{fig:rawsso} (Left). 
Knowing the receiver bandwidth $B_{\mbox{\scriptsize{rec}}}$, the DBSCAN on line~1 of Algorithm~2 then isolates  on the relevant frequency part of the SSO diagram, which shows a piecewise constant approximation to the frequencies in the signal (see Figure \ref{fig:rawsso} (Right)).
If the signals were in two different bands, this step would isolate both of them separately. 
We will illustrate this in Section~\ref{bhag:thresholdselect}.

 \begin{algorithm}[ht]
\begin{algorithmic}[1]
\item[{\rm a)}] \textbf{Input:} the receiver bandwidth $B_{\mbox{\scriptsize{rec}}}$, the minimum separation $\eta, \{\Lambda_{j,k}\}$ for $j=1,\ldots,J_k$, $k=1,\ldots,D$, and minimum number of neighbors $D_1, D_2$.
\item[{\rm b)}] \textbf{Output:} Estimation of $\omega_p, B_p, d_{p}$ and $\gamma_p$ for $p = 1, \ldots, P$.
\STATE $(\{\Lambda_{j,k}\},-1)$, $(\{\Lambda_{j,k}\},1)$ $\gets$ DBSCAN($\{\Lambda_{j,k}\}$, $B_{\mbox{\scriptsize{rec}}}$, $D_1$) to find the part of the $\{\Lambda_{j,k}\}$ relevant for our signal.
\STATE $\{\Lambda_{j,k}\}$ $\gets$ $(\{\Lambda_{j,k}\},1)$
\STATE $(\{\Lambda_{j,k}\}, -1)$, $(\{\Lambda_{j,k}\}, 1)$, $\ldots$, $(\{\Lambda_{j,k}\}, P)$ $\gets$ DBSCAN($\{\Lambda_{j,k}\}$, $\eta$, $D_2$) to separate the $\{\Lambda_{j,k}\}$ into $P$ clusters.
\FOR{$p=1$ to $P$}
\STATE $\Gamma_p \gets [(t_k, y_{j,k}) \mbox{ for } y_{j,k} \in (\{\Lambda_{j,k}\}, p)]$
\STATE $\gamma_p \gets$ $\min(t_k)$ for $(t_k, y_{j,k}) \in \Gamma_p$.
\STATE $d_p \gets$ $\max(t_k) - \min(t_k)$ for $(t_k, y_{j,k}) \in \Gamma_p$
\STATE $B_p \gets$ $(\max(y_{j,k}) - \min(y_{j,k}))/2$ for $(t_k, y_{j,k}) \in \Gamma_p$
\STATE Use linear regression on 50\% of $\Gamma_p$ that has highest $|\sigma_{n,k}(x)|$ to obtain the parameter estimation for $\omega_p$.
\STATE $\tau_{\mbox{\scriptsize{rmse}},p}$ $\gets$ $1\%$ of $\mbox{mean}(y_{j,k})$ for $(t_k, y_{j,k}) \in \Gamma_p$.
\IF{$\mbox{RMSE}_p$ defined in \eqref{eq:rmse} $>$ $\tau_{\mbox{\scriptsize{rmse}},p}$}
\STATE Perform Algorithm 3
\ENDIF
\ENDFOR
\STATE \textbf{Return: } $\omega_p, B_p, d_p$ and $\gamma_p$ for $j = 1, \ldots, P$.
 \caption{Parameter estimation}
 \end{algorithmic}
 \label{alg:estimation}
 \end{algorithm}

Given the minimum separation, we can separate these approximations for different signal components using DBSCAN on Line~3 of
Algorithm~2.
The remainder of Algorithem~2 then focuses on this approximation in each component of the signal, and finds the corresponding parameters using linear regression.
The result of this operation is shown in Figure~\ref{fig:ssofound}.

\begin{figure}[H]
\begin{center}
\includegraphics[scale=.2]{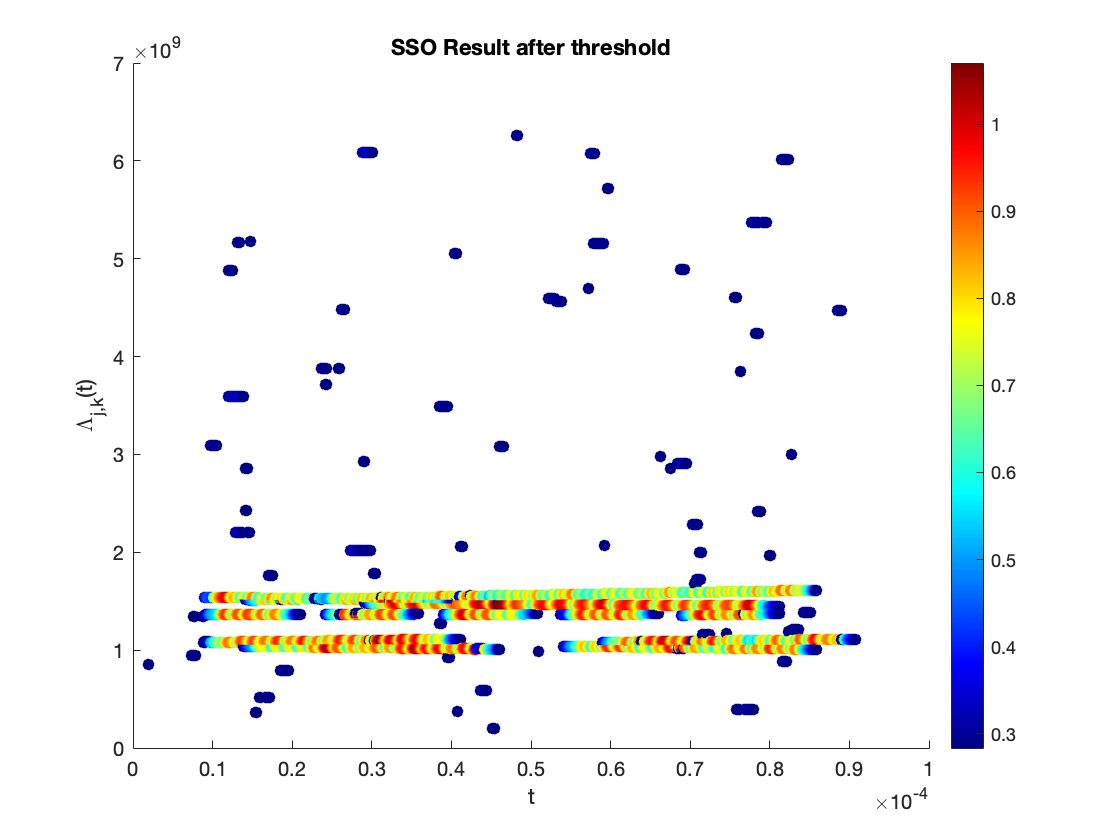}
\includegraphics[scale=.2]{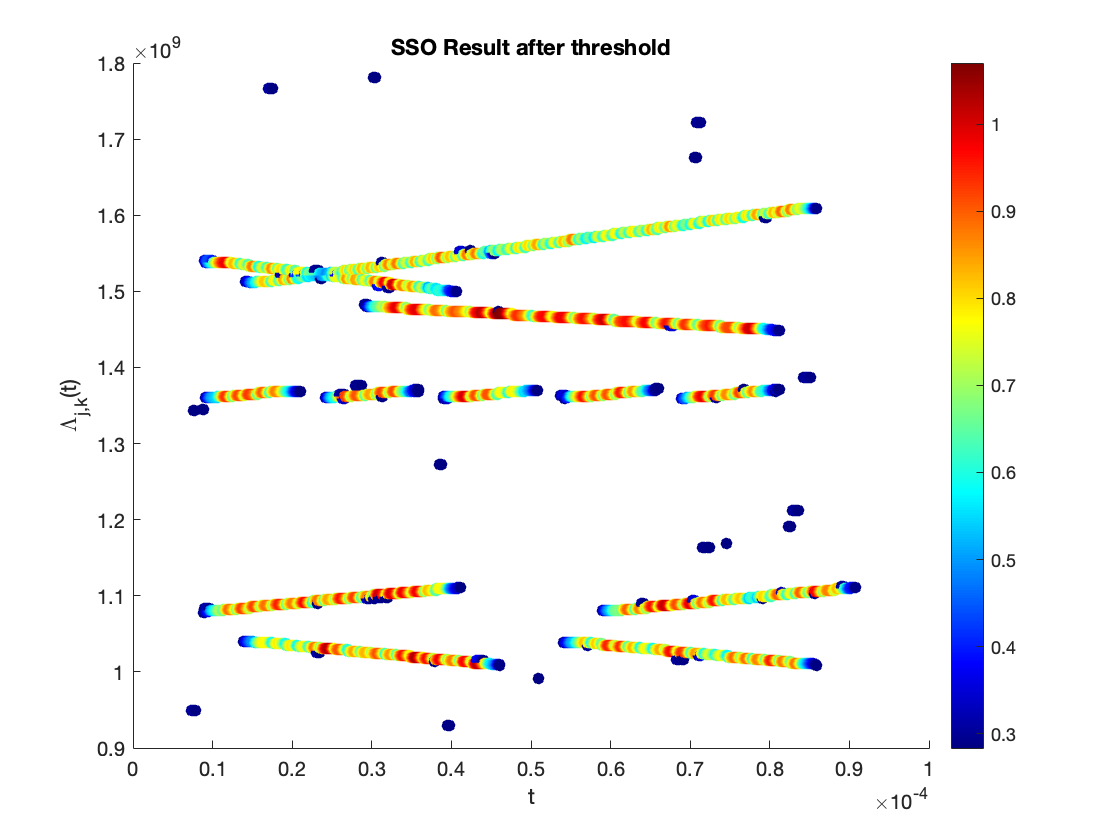}
\end{center}
\caption{(Left) The plot of threshold SSO result $\{\Lambda_{j,k}\}$ for $j=1,\ldots,J_k$, $k=1,\ldots,D$ by choosing $\Delta = 2\times 10^{-6}, D=2500, D_1=D/2, D_2=D/100,$ and $t_k$ are equidistant samples from 0 to $1\times 10^{-4}$. 
The peaks represented by blue dots result from noise, and may be subject to aliasing. Our sampling rate ensures that the signal of interest is not affected by aliasing. (Right) The plot of threshold SSO result after step 1 of Algorithm 2.}
\label{fig:rawsso}
\end{figure}

\subsection{Refinements}\label{bhag:refine}

Algorithm~3  focuses on working on the clusters in the SSO diagram where the signals crossover; i.e. the distance between two or more signals is less than the minimum separation $\eta$.

\begin{figure}[H]
\begin{center}
\includegraphics[scale=.2]{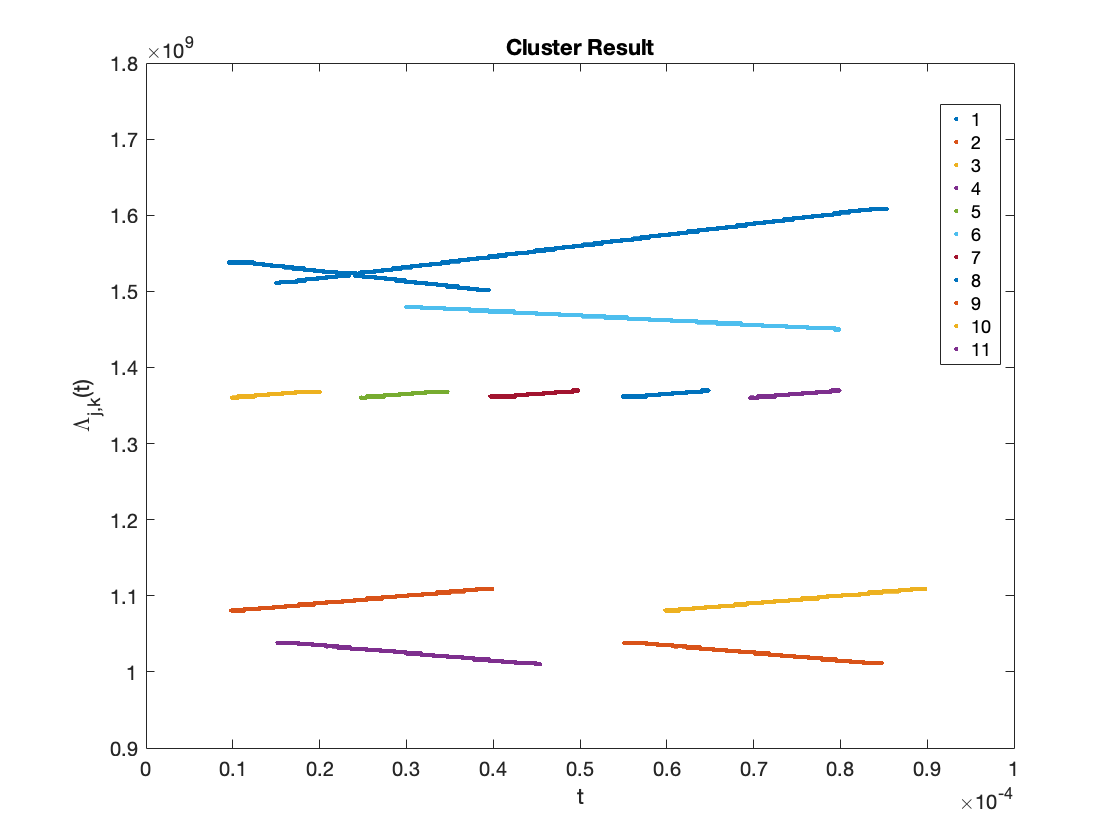}
\includegraphics[scale=.2]{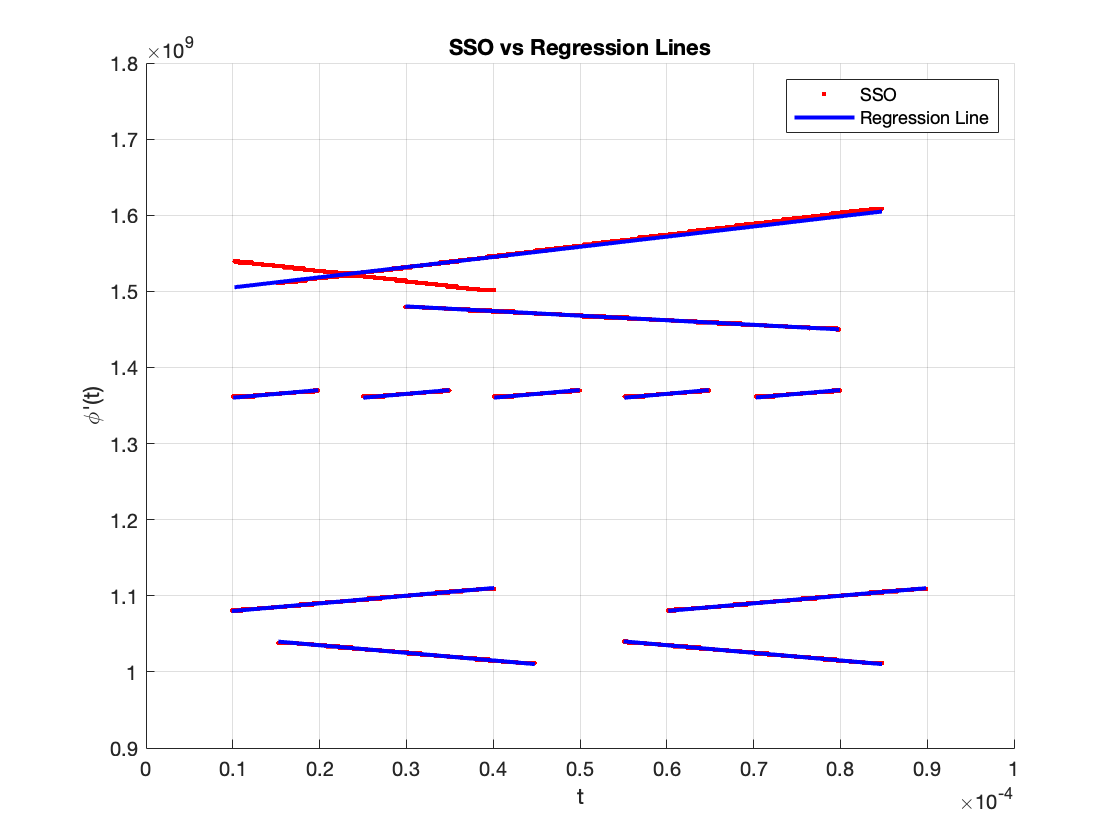}
\end{center}
\caption{(Left) The plot of signal components after using DBSCAN on Line~3 of Algorithm~2. (Right) The comparison plot between the SSO result vs the regression lines of $((\{\Lambda_{j,k}\},1), p)$. This is the result of Algorithm 2 without performing Algorithm 3; i.e. skipping steps 8 - 10 of Algorithm 2.}
\label{fig:ssofound}
\end{figure}

\begin{algorithm}[ht]
\begin{algorithmic}[1]
\item[{\rm a)}] \textbf{Input:} the minimum separation $\eta, (\{\Lambda_{j,k}\}, p)$ for $p$ when signals crossover, and minimum number of neighbors $D_2$, partition number $M$.
\item[{\rm b)}] \textbf{Output:} Estimation of $\omega_p, B_p, d_p$ and $\gamma_p$ for $p$ when signals crossover.
\STATE $\{\Lambda_{j,k}\}$ $\gets$ $(\{\Lambda_{j,k}\},p)$
\STATE partitions $\gets$ partition $\{\Lambda_{j,k}\}$ into $M$ parts.
\FOR{partition $m$ in  partitions}
\STATE $(\{\Lambda_{j,k}\}, -1)$, $(\{\Lambda_{j,k}\}, 1)$, $\ldots$, $(\{\Lambda_{j,k}\}, P_m)$ $\gets$ DBSCAN($\{\Lambda_{j,k}\}$, $\eta$, $D_2$) to separate the $\{\Lambda_{j,k}\}$ into $P_m$ clusters.
\FOR{$p_m=1$ to $P_m$}
\STATE $\Gamma_{p_m} \gets [(t_k, y_{j,k}) \mbox{ for } y_{j,k} \in (\{\Lambda_{j,k}\}, {p_m})]$
\STATE $\gamma_{p_m} \gets$ $\min(t_k)$ for $(t_k, y_{j,k}) \in \Gamma_{p_m}$.
\STATE $d_{p_m} \gets$ $\max(t_k) - \min(t_k)$ for $(t_k, y_{j,k}) \in \Gamma_{p_m}$
\STATE $B_{p_m} \gets$ $(\max(y_{j,k}) - \min(y_{j,k}))/2$ for $(t_k, y_{j,k}) \in \Gamma_{p_m}$
\STATE Use linear regression on 50\% of $\Gamma_{p_m}$ that has highest $|\sigma_{n,k}(x)|$ to obtain the parameter estimation for $\omega_{p_m}$.
\STATE $\tau_{\mbox{\scriptsize{rmse}},{p_m}}$ $\gets$ $1\%$ of $\mbox{mean}(y_{j,k})$ for $(t_k, y_{j,k}) \in \Gamma_{p_m}$.
\IF{$\mbox{RMSE}_{p_m}$ defined in \eqref{eq:rmse} $>$ $\tau_{\mbox{\scriptsize{rmse}},{p_m}}$}
\STATE Remove $\omega_{p_m}, B_{p_m}, d_{p_m}$, and $\gamma_{p_m}$ from estimate parameters and return fail to detect this part of the signal.
\ENDIF
\ENDFOR
\ENDFOR
\STATE Compute $\phi'_{p_m}$ for each partition. If any two or more of the $\phi'_{p_m}$'s are within $10\%$ each other,  then average them and update these parameters to $\omega_{p}, B_{p}, d_{p}$ and $t_{p}$.
\STATE \textbf{Return: } $\omega_{p}, B_{p}, d_{p}$, and $\gamma_{p}$ for $p$ when signals crossover.


 \caption{Handling signals crossover}
 \end{algorithmic}
 \label{alg:crossover}
 \end{algorithm}

The part where signals crossover is detected by computing the RMSE between $y_{j,k}$ for $(t_k, y_{j,k}) \in \Gamma_{p}$ and the regression line using the formula
\bea\label{eq:rmsedef}
\mbox{RMSE}_p &=& \sqrt{\frac{1}{D}\sum_{k=1}^D\sum_{j=1}^{J_k} |y_{j,k} - (\omega_{p}+B_{p}/d_{p}(t_k-\gamma_{p}))|^2}, \label{eq:rmse}
\eea
where the parameters $\omega_p, B_p$ etc. are determined on lines~5-9 of Algorithm~2.
If the $\mbox{RMSE}_p$ is larger that a small threshold, we then proceed to the refinement algorithm, Algorithm 3.

To refine this part of the signal, we partition the signal into small intervals as shown in Figure \ref{fig:algo3}. Finally, we set a threshold on the $\phi_{p_m}'(t)$ and connect them if the differences between two or more $\phi_{p_m}'(t)$'s are within the threshold. 
The principles for doing this are the same as discussed in Section~\ref{bhag:paramest}.
However, at this stage, the noise is practically eliminated, making the choices easier.
The final result is shown in Figure \ref{fig:final}.

\begin{figure}[H]
\begin{center}
\includegraphics[scale=.15]{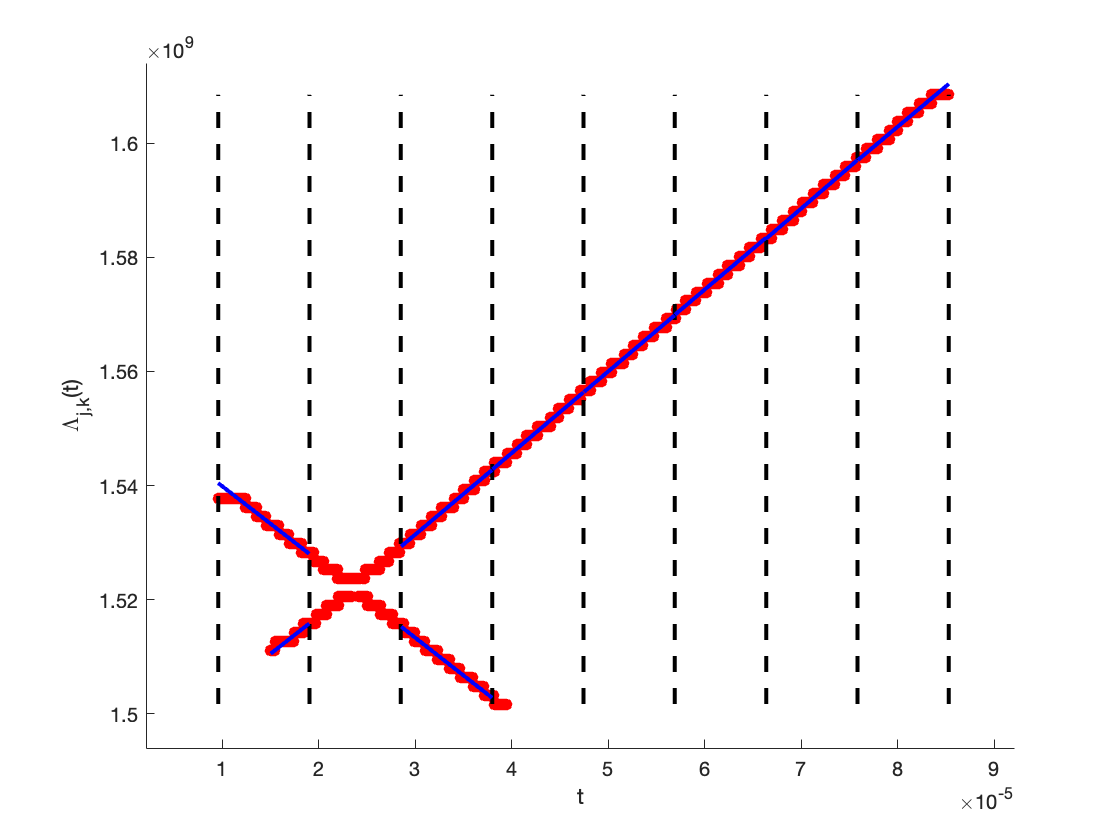}
\includegraphics[scale=.15]{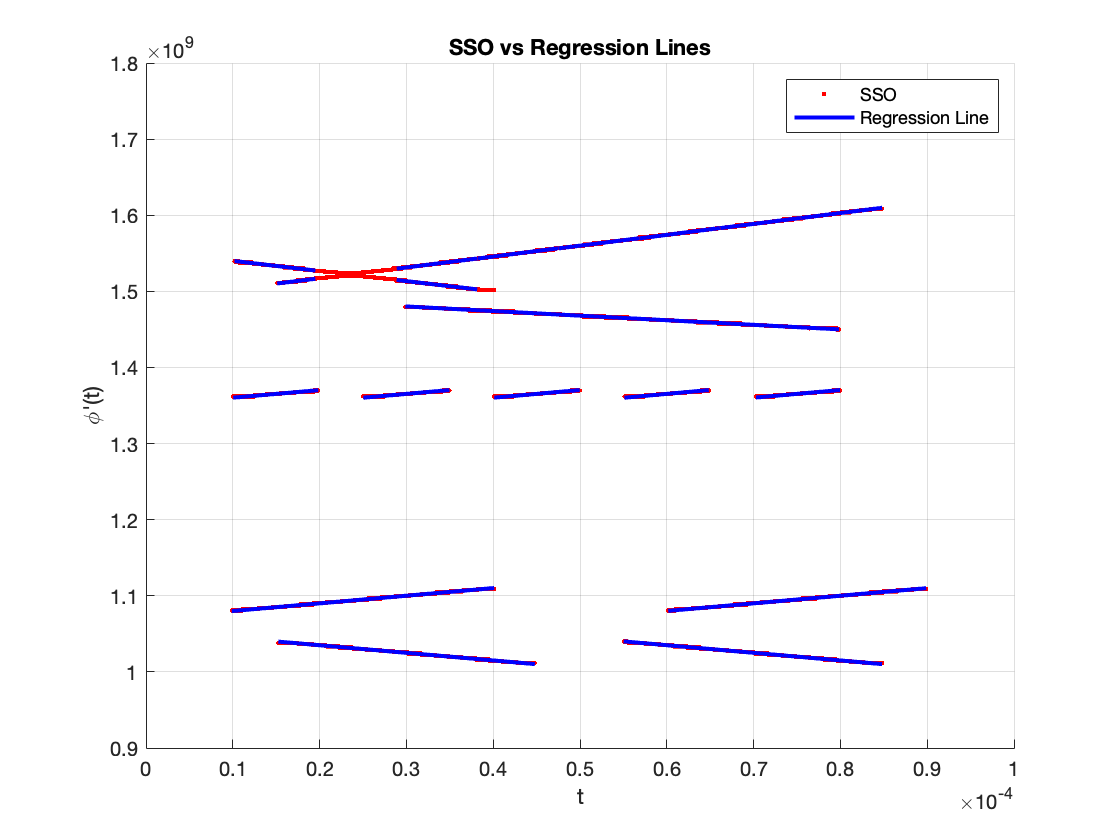}
\end{center}
\caption{(Left) The plot of crossing signals after performing step 1 - 13 of the Algorithm 3 by choosing $M = 8$ partitions. (Right) The comparison plot between the SSO result vs the regression lines of $\{\Lambda_{j,k}\}$ after performing step 1 - 13 of the Algorithm 3.}
\label{fig:algo3}
\end{figure}

\begin{figure}[H]
\begin{center}
\includegraphics[scale=.15]{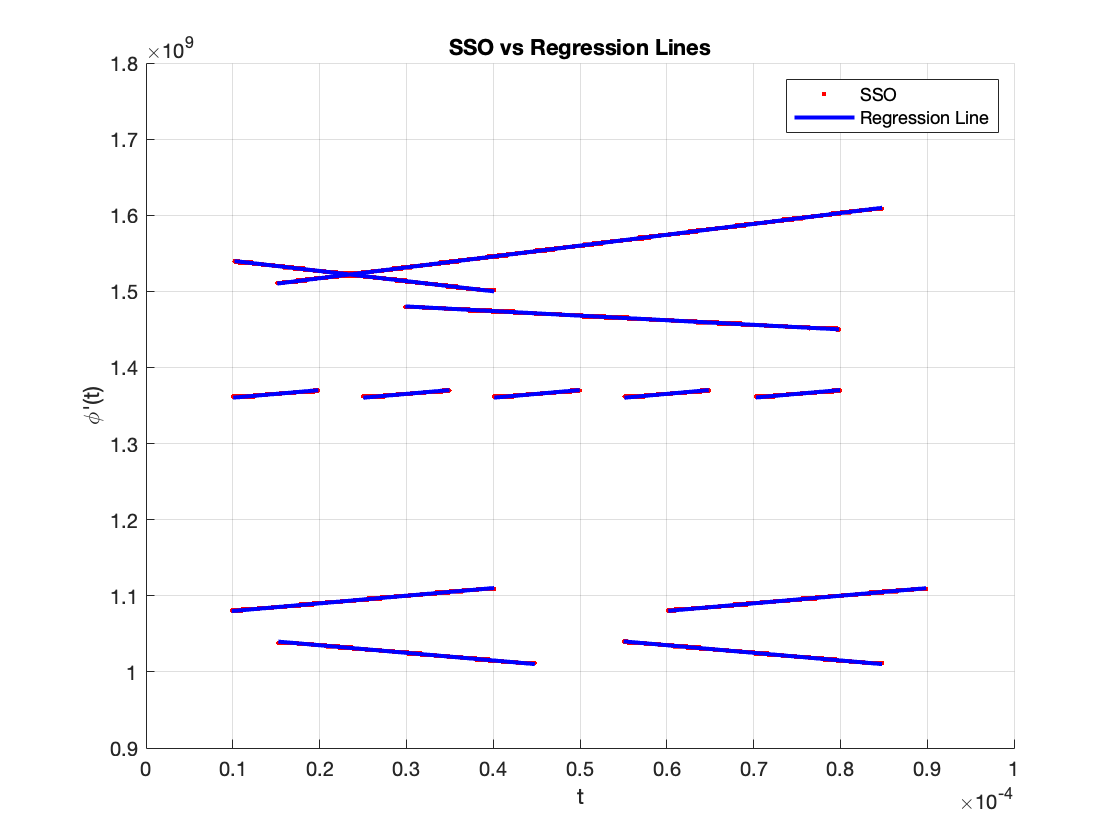}
\includegraphics[scale=.15]{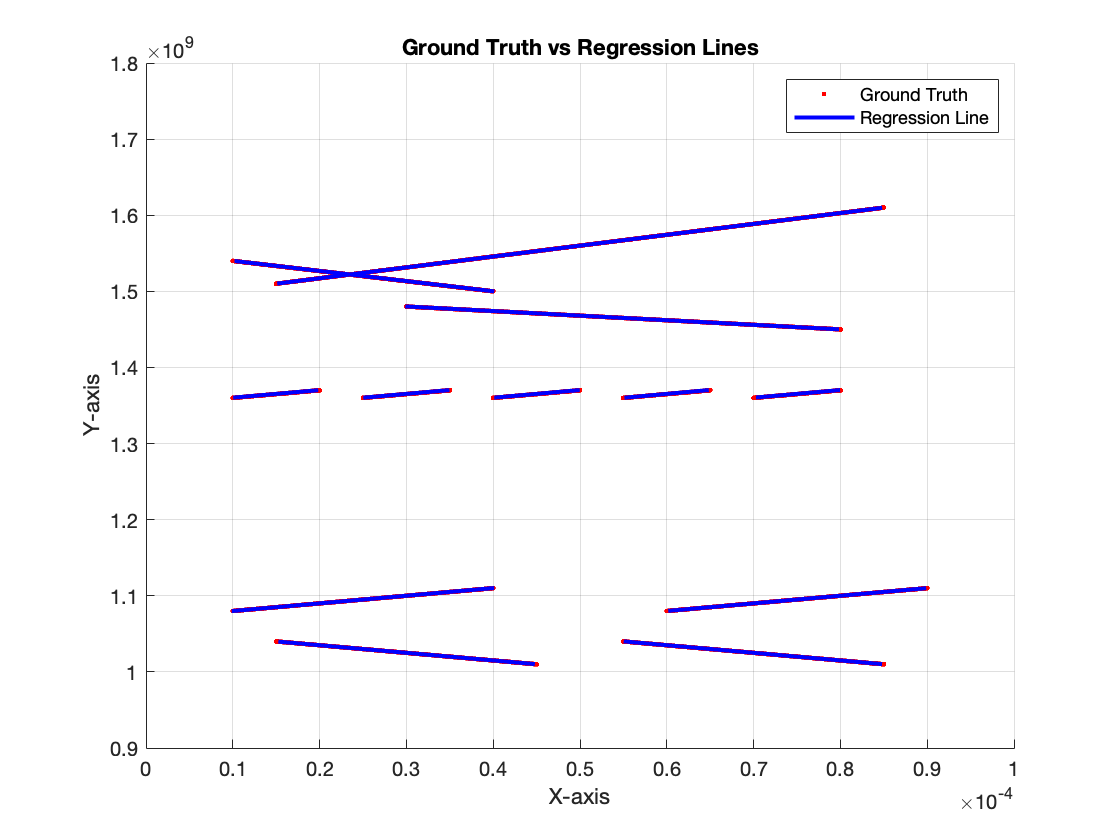}
\end{center}
\caption{The final comparison plot (after all algorithms are executed) between the SSO result vs the regression lines of $\{\Lambda_{j,k}\}$ (left) and the ground truth vs the regression lines (right).}
\label{fig:final}
\end{figure}

\vspace*{-0.5cm}

\bhag{Data generation}\label{bhag:datagen}
In order to test our methodology, we generated a dataset of 15 examples, each consisting of 6 sub-signals in a 0.1 millisecond time window. 
This will be made available at a web location after the paper is published.  
The parameters of the 6 signals are defined in Table~\ref{tab:signal_table} used to illustrate the algorithm in Section~\ref{bhag:algorithm}.  
The additional  examples are generated by keeping the parameter set $(A_j, B_j, d_j, t^{0,j}, \mbox{PRI}_j)$ fixed for  all the sub-signals, but in each example, $\omega_j$ is  selected randomly for each sub-signal in the range $[1,1.6]$ GHz, thereby shifting the signals in the spectrum, each with a different realization of i.i.d white Gaussian noise and sample rates.


We add noise to our generated original signal using the formula
\be\label{eq:add_noise}
\boldsymbol{F} = \boldsymbol{f} + \boldsymbol{\epsilon},
\ee
where $\boldsymbol{F}$ is a vector formed by uniformly sampling the noisy signal $F(t)$ at discrete time points; $\boldsymbol{f}$ is a vector formed by uniformly sampling the original signal $f(t)$ at discrete time points; $\boldsymbol{n}$ is a vector of white Gaussian noise i.i.d samples; and $\boldsymbol{\epsilon}$ is a vector of the noise $\epsilon(t)$ defined by
\be
\boldsymbol{\epsilon} = \frac{||\boldsymbol{f}||_2}{10^{\mbox{\scriptsize{SNR}}/20} * ||\boldsymbol{n}||_2} * \boldsymbol{n}.
\ee
In the numerical experiments shown in Section~\ref{bhag:numerical_results} we use this formula to generate 7 sets of signals, each with 6 sub-signals, and with SNRs $\{-30, -20, -10, 0, 10\}$ dB.  
Furthermore, we also consider three different sample rates $\{ 0.5, 1, 25.0 \}$ GHz leading to 50,000, 100,000, and 2,500,000 samples for the length of $\boldsymbol{F}$ and $\boldsymbol{\epsilon}$ respectively.   
The same process is used to generate the data used in Section~\ref{bhag:refine} with one set of signals given by the parameters in Table~\ref{tab:signal_table}.

\bhag{Choice of tunable parameters}\label{bhag:tunable}
In this section, we describe some heuristic strategies to select the main tunable parameters and their interdependence: the minimal separation $\eta$, the radius $\Delta$ of the snippets,  the degree $n$ in \eqref{eq:sigmadef}, and the threshold in Algorithm~1. 

\subsection{Parameter $\eta$ (minimal separation)}\label{bhag:etaselect}

Ideally we want to pick parameters $\eta$ and $n$ so that $\eta$ is a little bit bigger than the distance between the main prominent peak and its side lobe peak of the kernel $\Phi_n$. 
If $\eta$ is too small, we will have very noisy result. If $\eta$ is too big, we won't detect 2 or more signals that are within the distance $\eta$. 

Here is the table of the recommended relationship between $n$ and $\eta$.

\begin{table}[H]
\begin{center}
\begin{tabular}{ |c|c|c|c|c|c|c|c|c| } 
 \hline
 $n$ & 128 & 256 & 512 & 1000 & 2000 & 5000 & 10000 & 50000 \\
 \hline
 $\eta$ & 0.070204 & 0.035114 & 0.017545 & 0.008988 & 0.004482 & 0.001798 & 0.000887 & 0.000192 \\
 \hline
\end{tabular}
\end{center}
 	\caption{The tables above shows the relationship between $n$ and $\eta$.} \label{tab:n_vs_eta}
\end{table}

Figure~\ref{fig:n_vs_eta}  shows the effect of different $n$ and $\eta$ vs the peaks detected for the signal 
$$
\hat{\mu}(\ell) = \exp(0.02i\ell) + \exp(-0.01i\ell) + \exp(-0.015i\ell)
$$ at sampling rate 1 GHz {without noise}.
We note that the resolution improves as $n$ increases, and for any given $n$, setting the parameter $\eta$ to be smaller than what is indicated in Table~\ref{tab:n_vs_eta} results in the detection of false signals. 
The effect of choosing a small eta might be mitigated by an appropriate choice of the threshold, but the purpose of this section is only to illustrate the connection between the choice of $n$ and $\eta$.

\begin{figure}[ht]
\begin{center}
\begin{minipage}{0.3\textwidth}
\includegraphics[width=\textwidth]{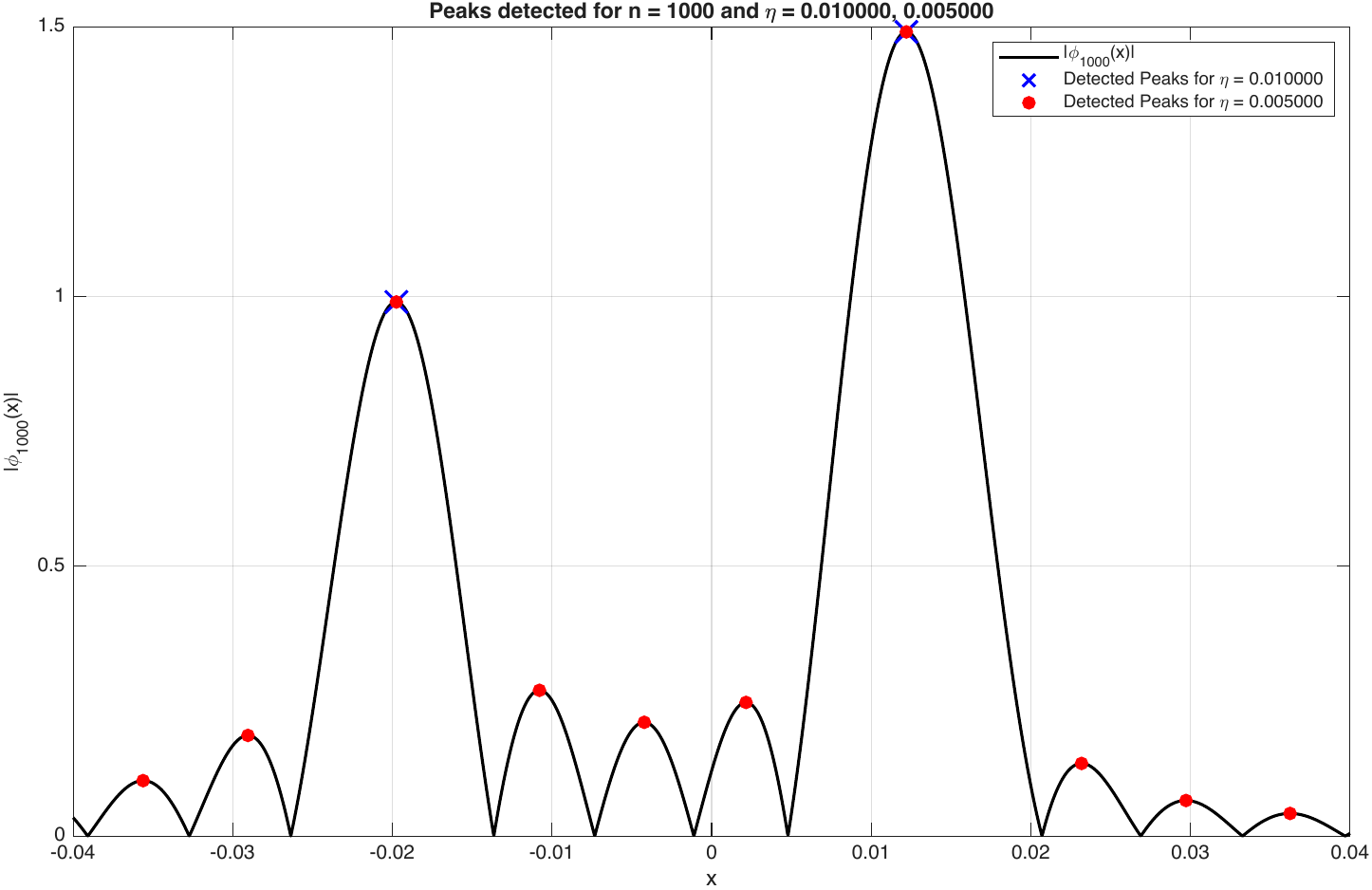}
\end{minipage}
\begin{minipage}{0.3\textwidth}
\includegraphics[width=\textwidth]{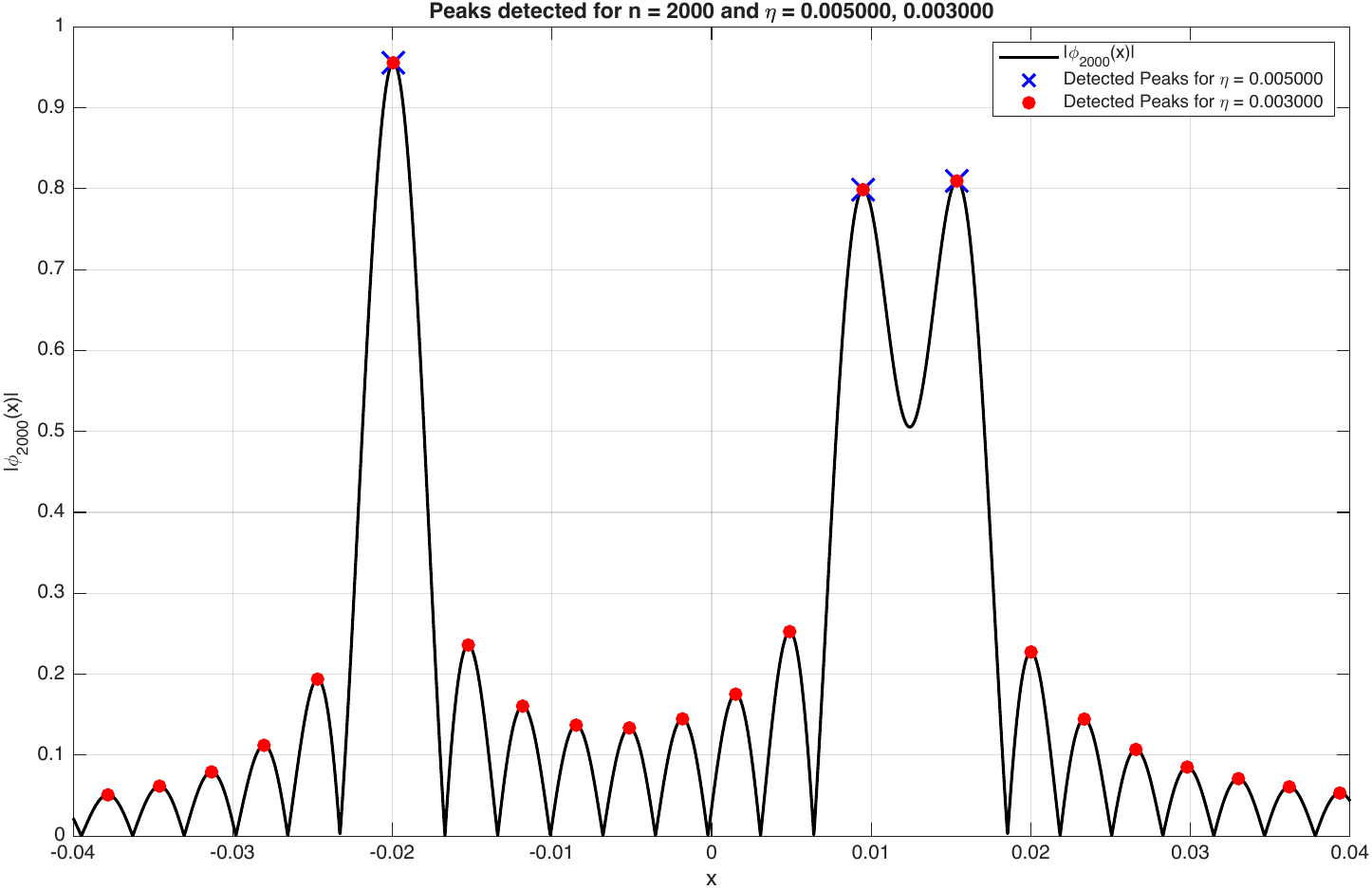}
\end{minipage}
\begin{minipage}{0.3\textwidth}
\includegraphics[width=\textwidth]{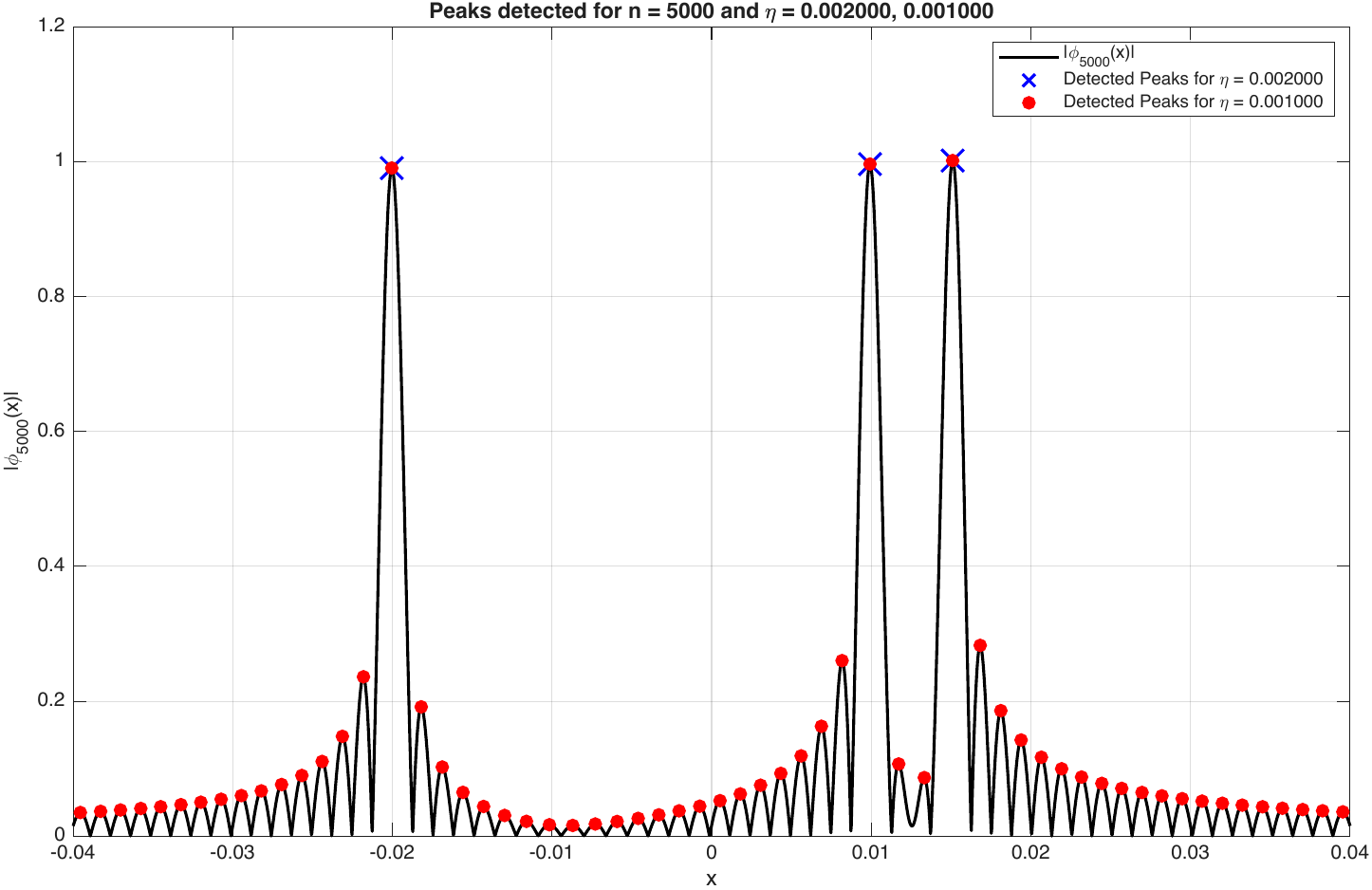}
\end{minipage}
\end{center}
\caption{The effect of choosing $\eta$ and $n$. In each figure, a value of $\eta$ smaller than the value indicated Table~\ref{tab:n_vs_eta} results in false detection of signals. Also, the higher the value of $n$, the better is the resolution for the right value of $\eta$.}
\label{fig:n_vs_eta}
\end{figure}

\subsection{Parameter $\Delta$ (half-length of each snippet)}\label{bhag:deltaselect}

The parameter $\Delta$ is perhaps the most delicated parameter to choose, as remarked in Remark~\ref{rem:ncond}.
For a given sampling rate, a small value of $\Delta$ results in a smaller value of $n$, and hence, less noise reduction as well as ability to resolve close-by signals as discussed in Section~\ref{bhag:etaselect}.
The purpose of this section is to demonstrate with an example the connection between the choice of the parameter $\Delta$ which governs the length of the snippets and the chirp factor.
For this purpose, we consider  the ground truth signal parameters are described in Table \ref{tab:signal_high_slope_table}. 
We note that the sub-signals $1$ and $4$ have a high chirp rate while the others have a relatively low chirp rate ($j=2$ and $3$).
According to \eqref{eq:slowvary}, we should use a small $\Delta$ to detect the high chirp rate signals and a larger $\Delta$ to detect the low chirp rate signals.
The experiments in this section illustrate this, and provides some heuristics for different ranges of the chirp rates.
To get the final result, we  start with a small $\Delta$ first to detect the total number of signals. 
Then, we perform our Algorithm 1-3 and measure the RMSE between the SSO result and its corresponding regression line.
 For signals with the RMSE below a threshold, we have a good estimation. 
 For signals with the RMSE above a threshold, we continue repeat the process with bigger and bigger values of $\Delta$.
 This is illustrated in Figure~\ref{fig:Deltaselect}, using heatmaps of the error rather than an overlap of estimated and actual frequencies.

\begin{table}[H]
\begin{center}
\begin{tabular}{ |c|c|c|c|c|c|c|c|c| }
 \hline
 $j$ & $A_j$ & $\omega_j$ & $B_j$ & $d_j$ & $t^*_{0,j}$ & $\mbox{PRI}_j$ & Chirp Sign & Number of Pulses \\
 \hline
 1 & 1 & 1080000000 & 150000000 & 3.0e-06 & 1.0e-05 & 5.0e-05 & + & 2 \\
 2 & 1 & 1360000000 & 5000000 & 1.0e-05 & 3.0e-05 & 0 & + & 1 \\
 3 & 1 & 1740000000 & 20000000 & 3.0e-05 & 1.0e-05 & 0 & - & 1 \\
 4 & 1 & 1510000000 & 20000000 & 1.0e-05 & 4.5e-05 & 0 & + & 1 \\
 \hline
\end{tabular}
\end{center}
 	\caption{The tables above shows an example of signal ground truth parameters.} \label{tab:signal_high_slope_table}
\end{table}

\begin{figure}[H]
\begin{center}
\begin{minipage}{0.3\textwidth}
\includegraphics[width=\textwidth]{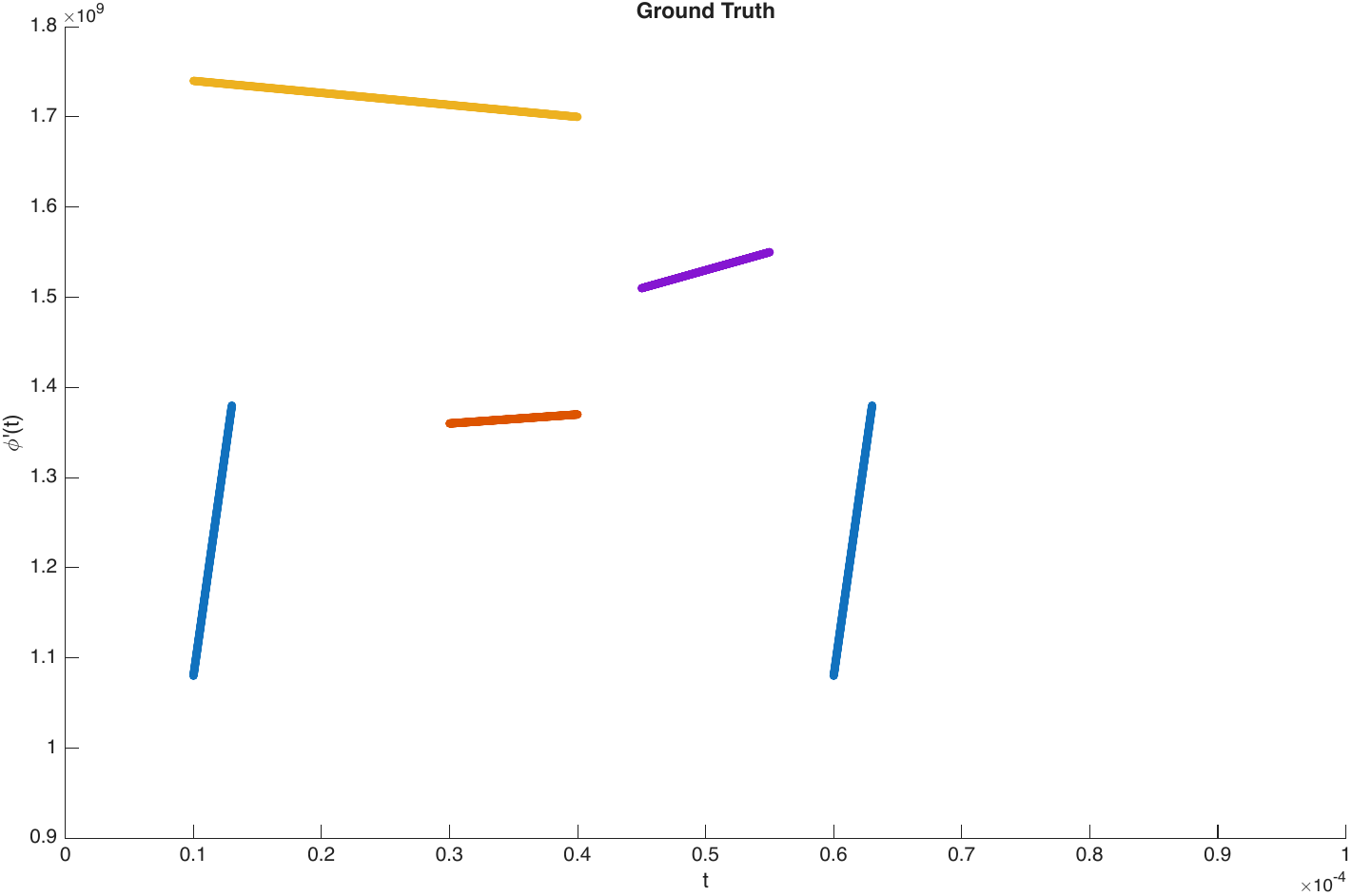}
\end{minipage}
\begin{minipage}{0.3\textwidth}
\includegraphics[width=\textwidth]{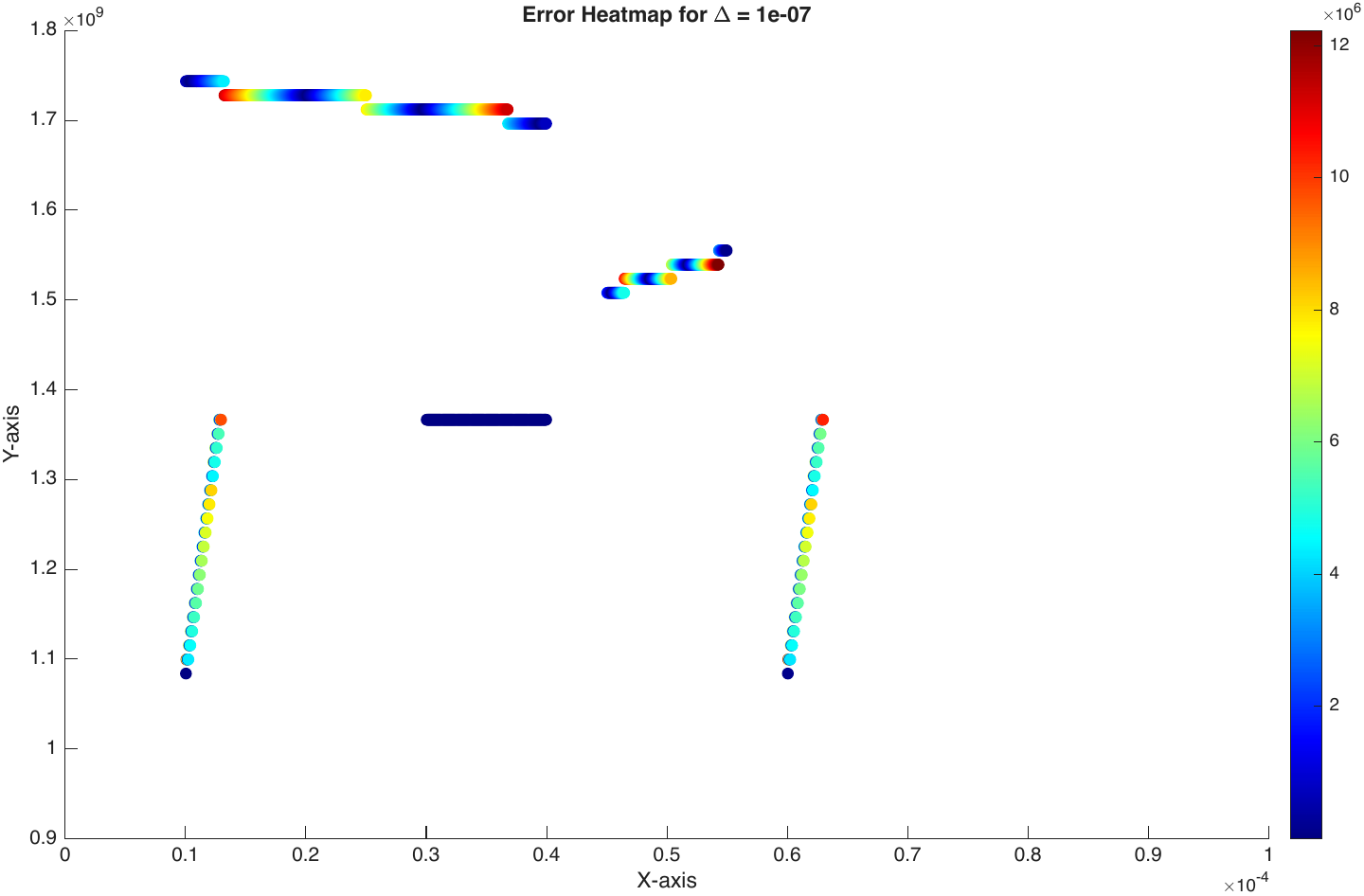}
\end{minipage}
\begin{minipage}{0.3\textwidth}
\includegraphics[width=\textwidth]{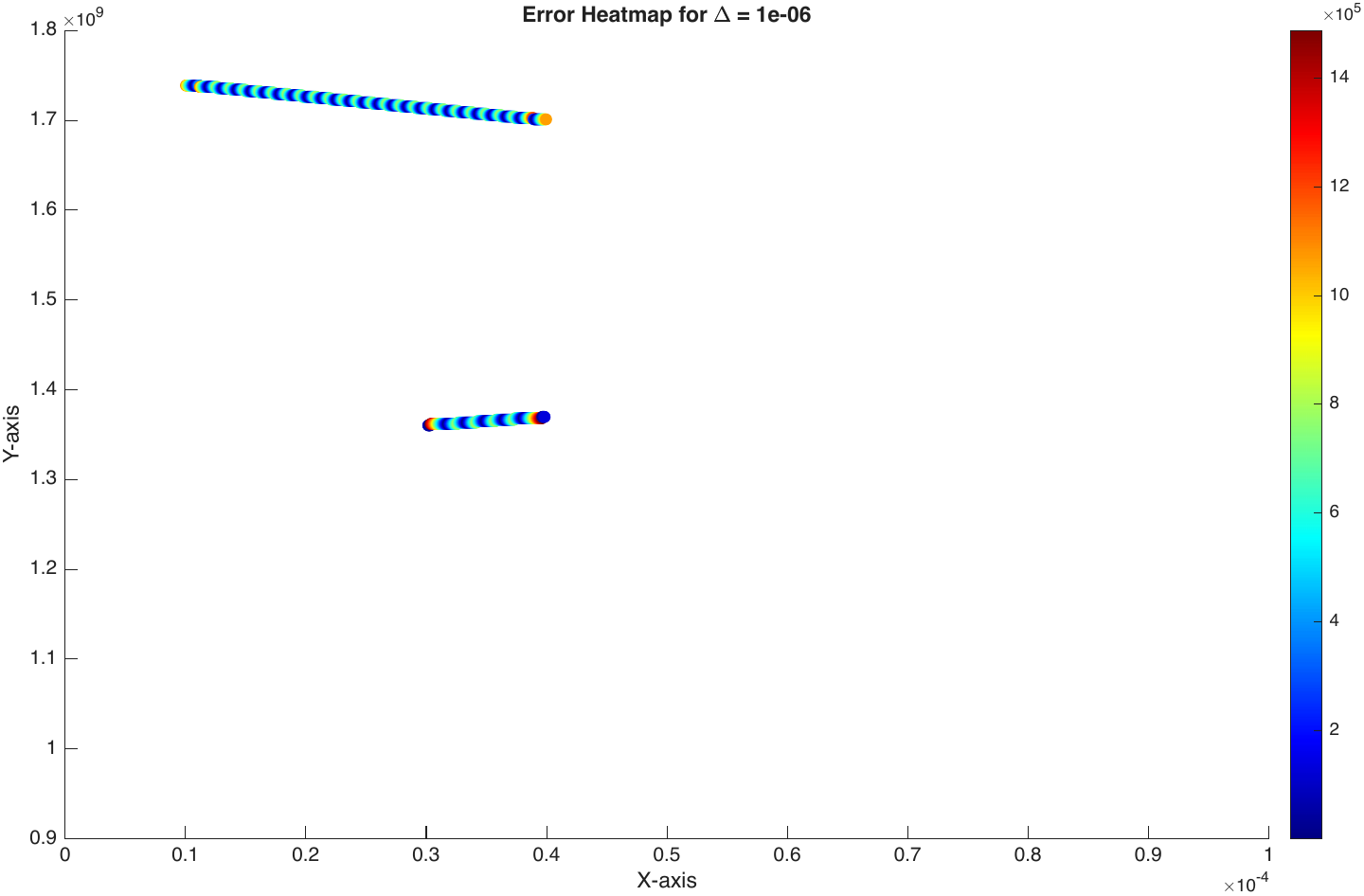}
\end{minipage}
\end{center}
\caption{Left: the ground truth IFs for Table~\ref{tab:signal_high_slope_table}. Middle: With $\Delta=10^{-7}$, the high chirp rate signals are detected accurately, while the low chirp rate signals are not, as indicated by {the error heatmap between the SSO results and the corresponding regression line}. Right: With $\Delta=10^{-6}$, the low chirp rate signals are detected accurately, but not the high chirp rate signals.}
\label{fig:Deltaselect}
\end{figure}

Here is the table of the relationship between $\Delta$ and suitable chirp rate $B_j/d_j$ range.

\begin{table}[H]
\begin{center}
\begin{tabular}{ |c|c|c|c|c|c|c|c|c| } 
 \hline
 $\Delta$ & 1e-06 & 5e-07 & 1e-07 \\
 \hline
 $B_j/d_j$ & [1e+11, 1e+12] & [1e+12, 1e+13] & [1e+13, 1e+14] \\
 \hline
\end{tabular}
\end{center}
 	\caption{The tables above shows the relationship between $\Delta$ and suitable chirp rate $B_j/d_j$ range.} \label{tab:Delta_vs_slope}
\end{table}

\subsection{Selection of the threshold}\label{bhag:thresholdselect}

In this section, we discuss different strategies for selecting the threshold to be used on Line~3 of Algorithm~1.

One strategy is described in Section~\ref{bhag:alg1}, where we select the threshold for each snippet to be 99.9 percentile of the histogram of the SSO applied to that snippet, which is the approach we use in the numerical simulation provided in Section~\ref{bhag:numerical_results}.
The selection of the percentile can be delicate as it depends on the amount of signals and their power relative to the noise SNR.
This is particularly impactful when using the DBSCAN algorithm in Line~1 of Algorithm~2 to identify different sub-bands of signals for further processing when the receiver bandwidth is large; this step is demonstrated in Figure~\ref{fig:twoband}.
Thus, before further processing, using the sub-bands identified by Line-1 of Algorithm-2 one can use a bandpass filter to isolate the different collections of signals and processed individually in the rest of the algorithm. 

\begin{figure}[H]
\begin{center}
\begin{minipage}{0.32\textwidth}
\includegraphics[width=\textwidth]{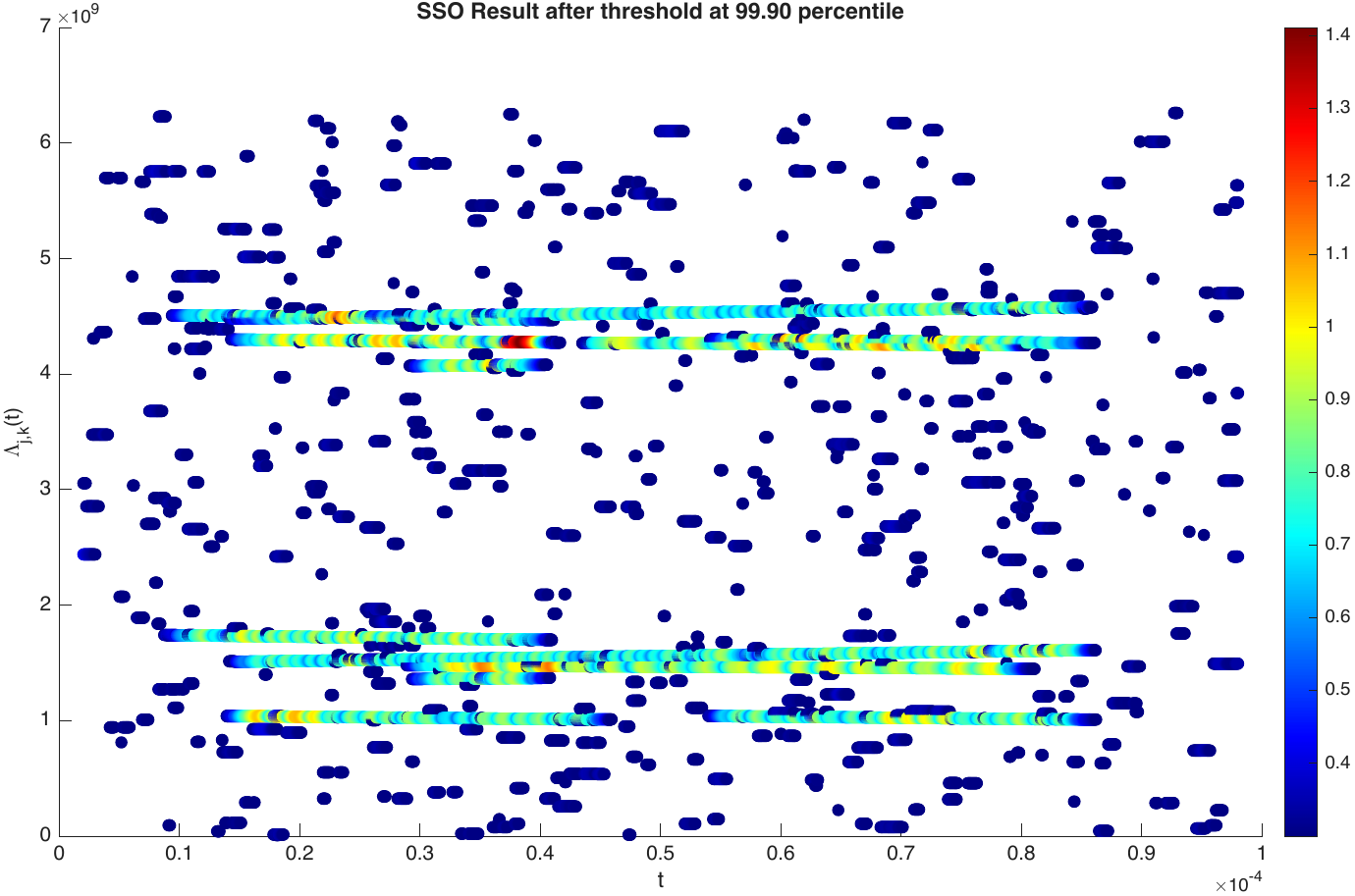} 
\end{minipage}
\begin{minipage}{0.32\textwidth}
\includegraphics[width=\textwidth]{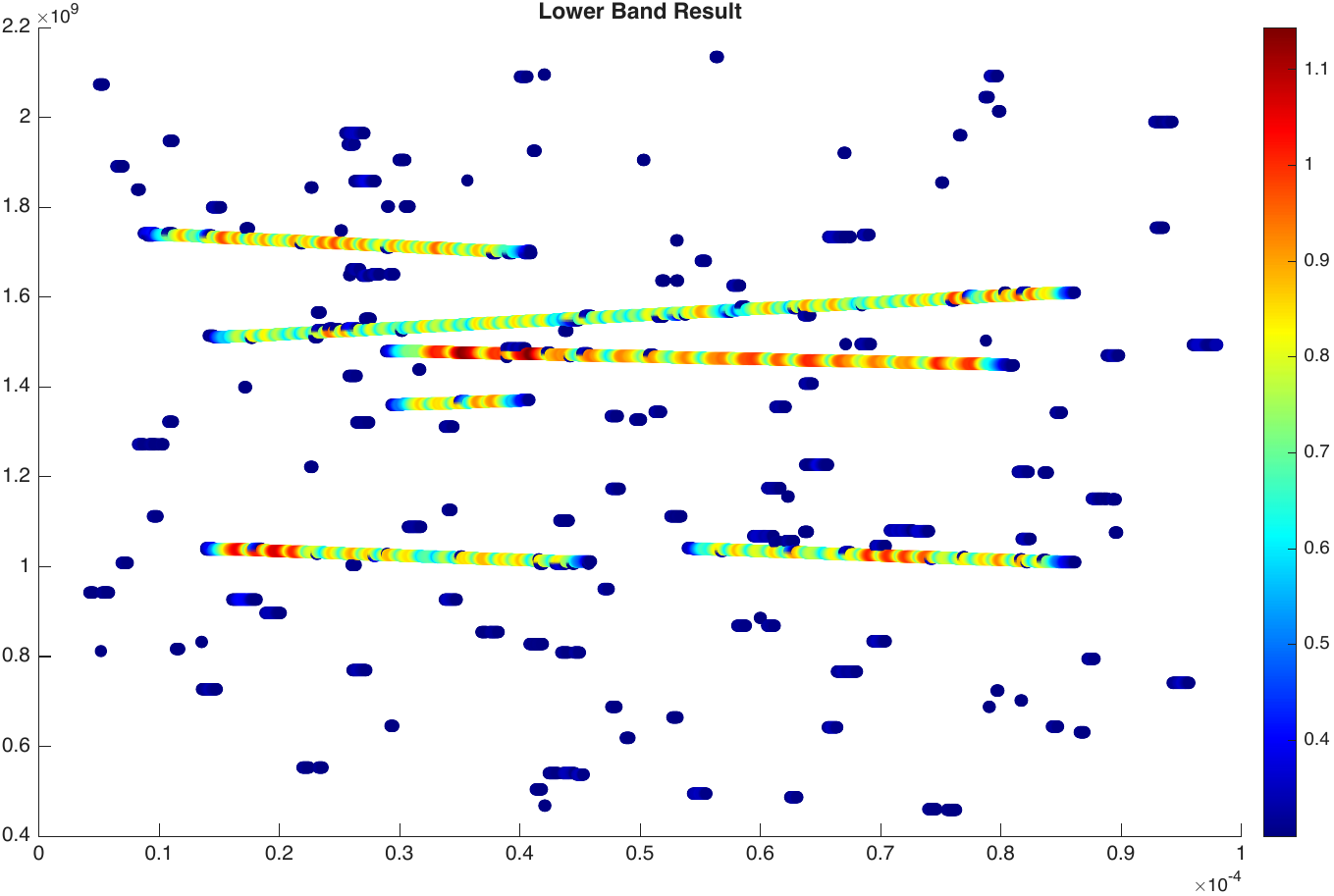} 
\end{minipage}
\begin{minipage}{0.32\textwidth}
\includegraphics[width=\textwidth]{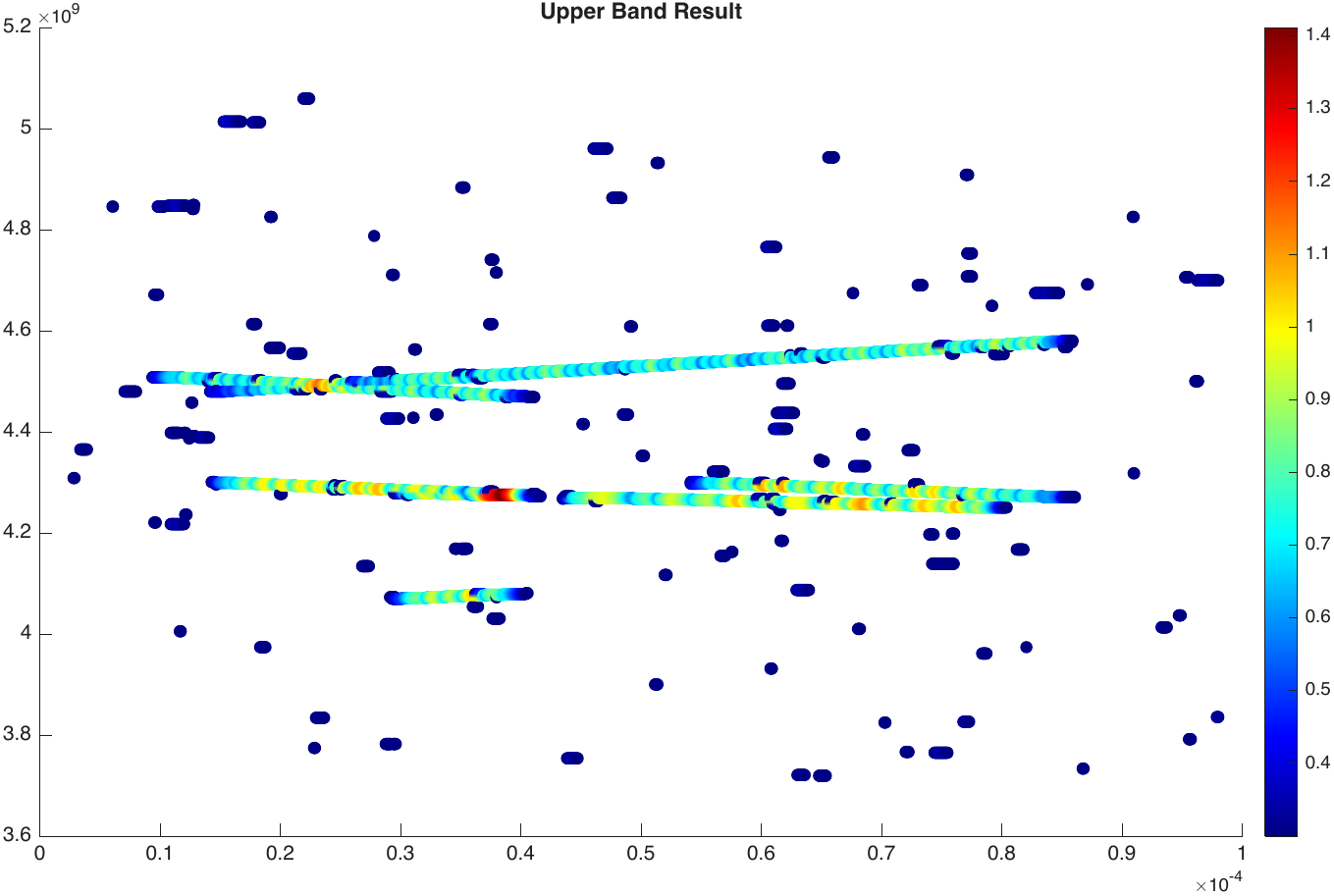} 
\end{minipage}
\end{center}
\caption{Example of a signal with two receiver bandwidths. Left: raw SSO diagram showing the two bandwidths, Middle: Isolation of the lower band, Right: isolation of the upper band.}
\label{fig:twoband}
\end{figure}

One can also skip this first clustering step if they have prior knowledge of the sub-band of the signal, specifically assuming knowledge of $\max_{j,k,t}\phi_{j,k}'(t)$ and $\min_{j,k,t}\phi_{j,k}'(t)$ separately.  To demonstrate the effect of different percentiles we assume knowledge of the signal content to identify the sub-band and evaluate the thresholding step on the result, shown in Figure~\ref{fig:thresholdselect}.

\begin{figure}[H]
\begin{center}
\begin{minipage}{0.3\textwidth}
\includegraphics[width=\textwidth]{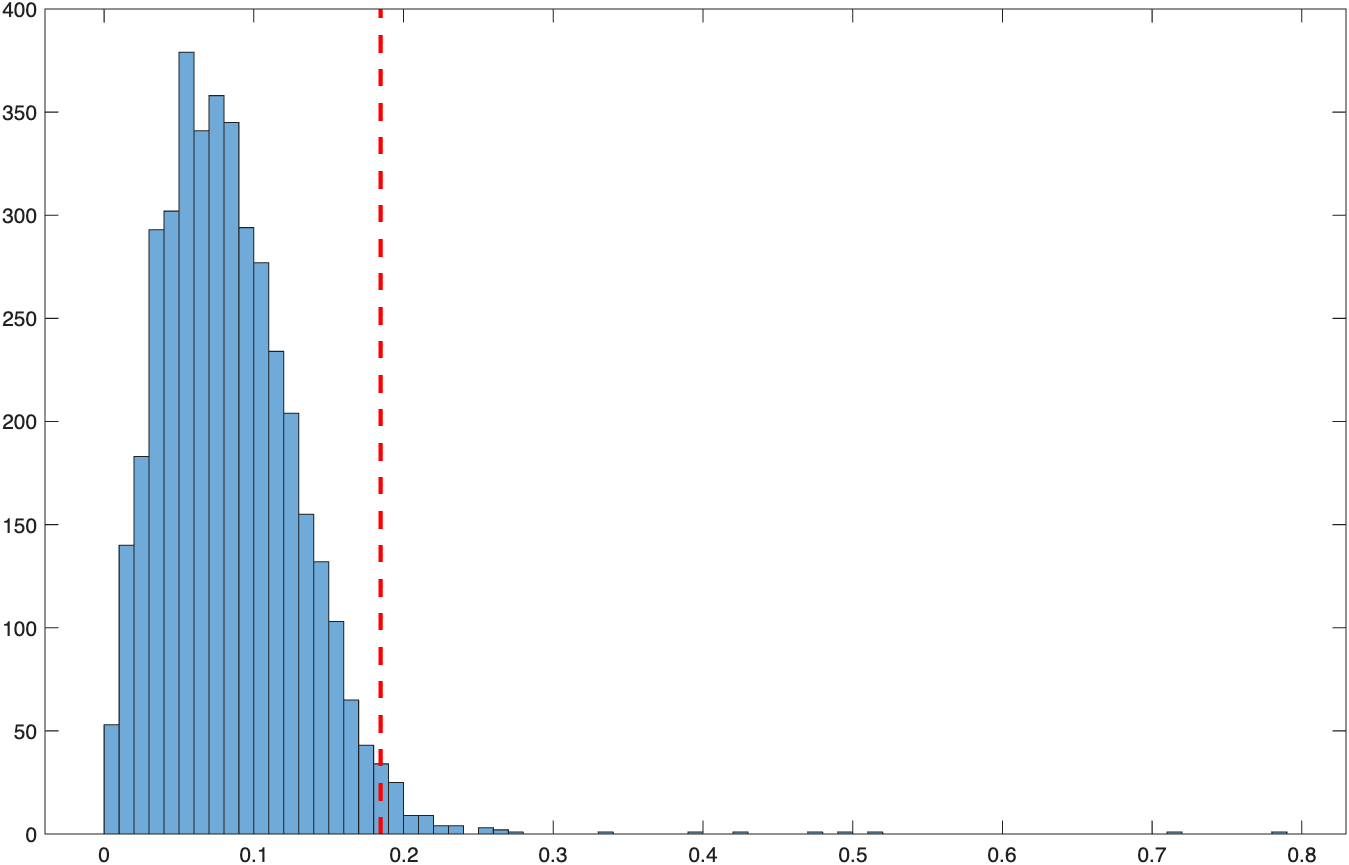}
\end{minipage}
\begin{minipage}{0.3\textwidth}
\includegraphics[width=\textwidth]{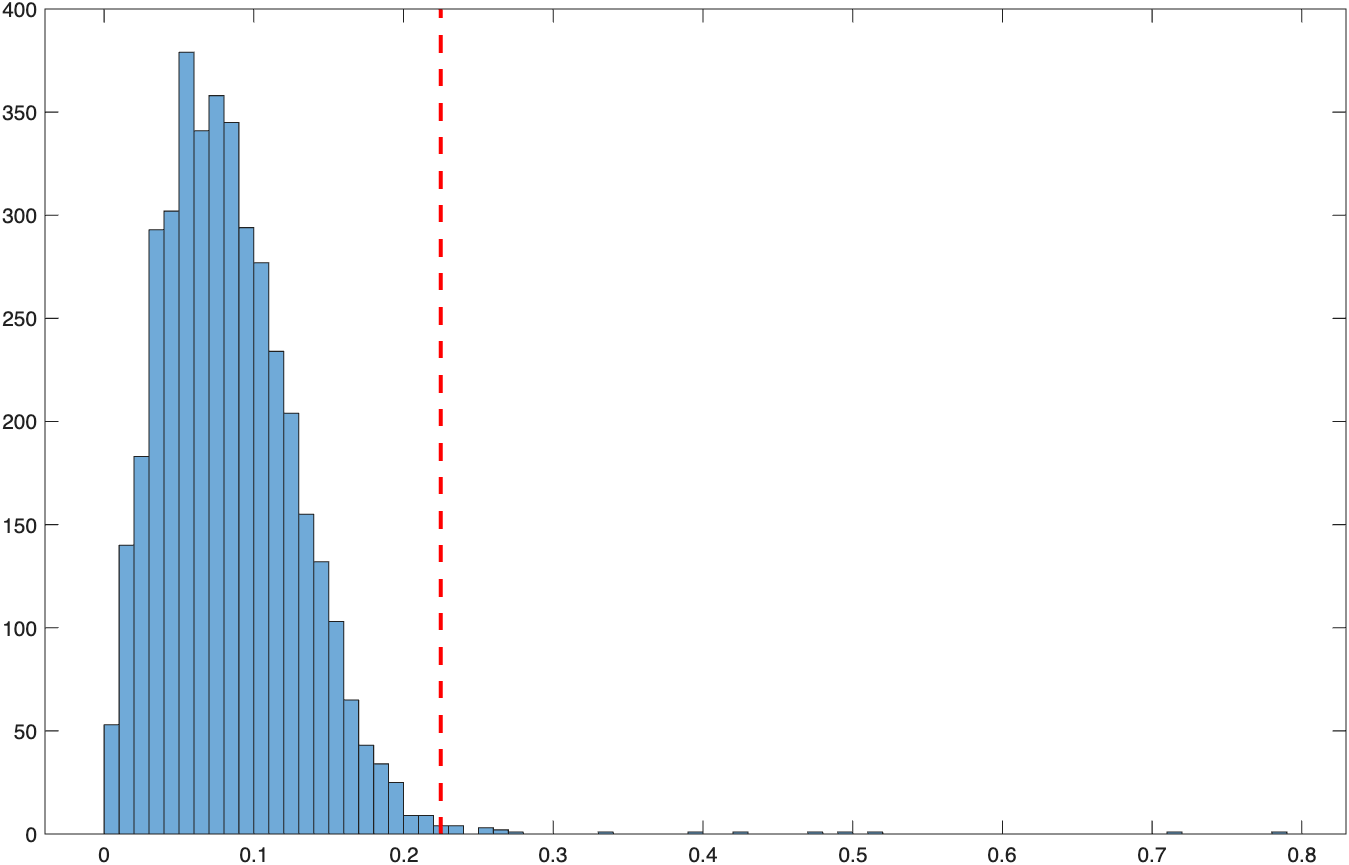}
\end{minipage}
\begin{minipage}{0.3\textwidth}
\includegraphics[width=\textwidth]{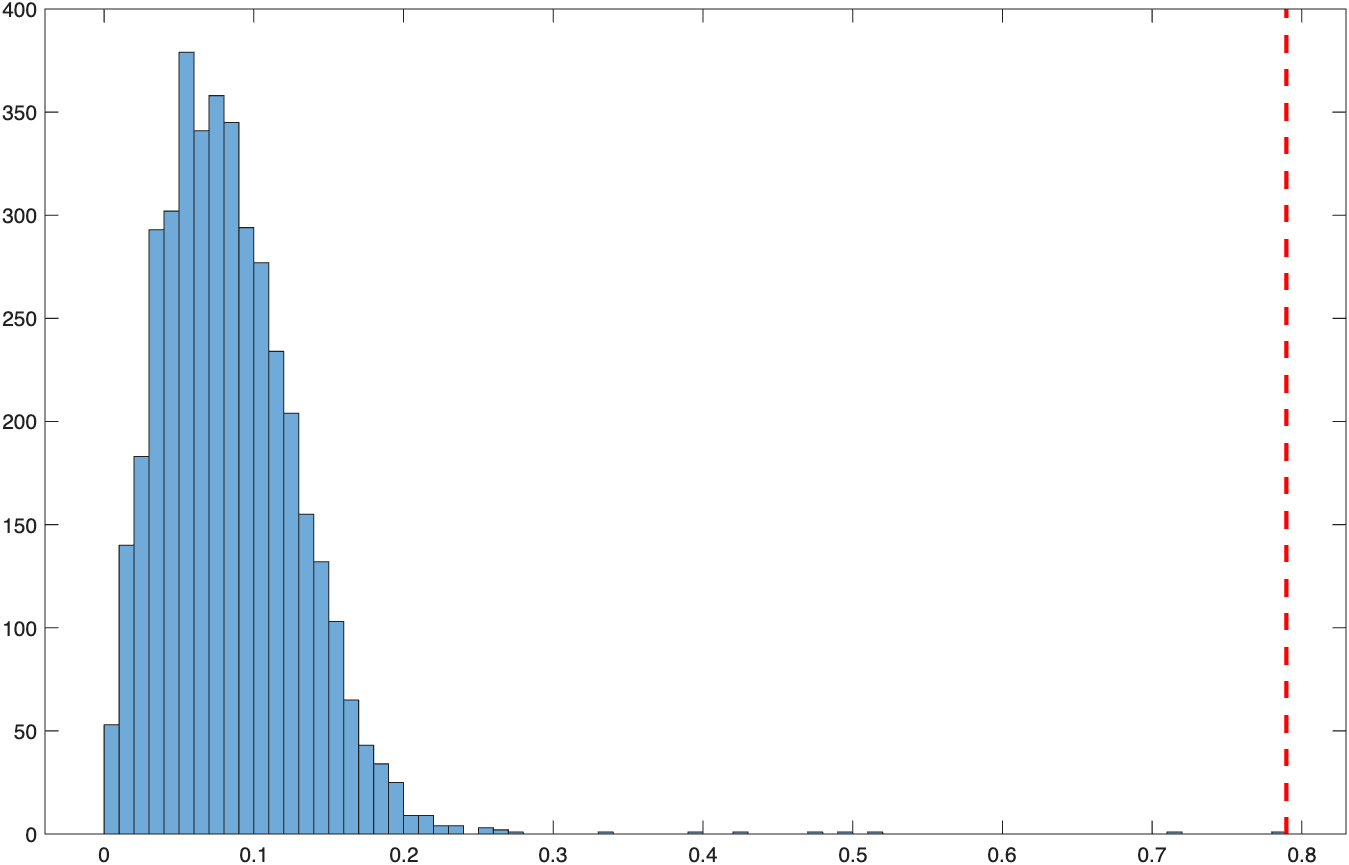}
\end{minipage}\\ \vspace*{1em}
\begin{minipage}{0.3\textwidth}
\includegraphics[width=\textwidth]{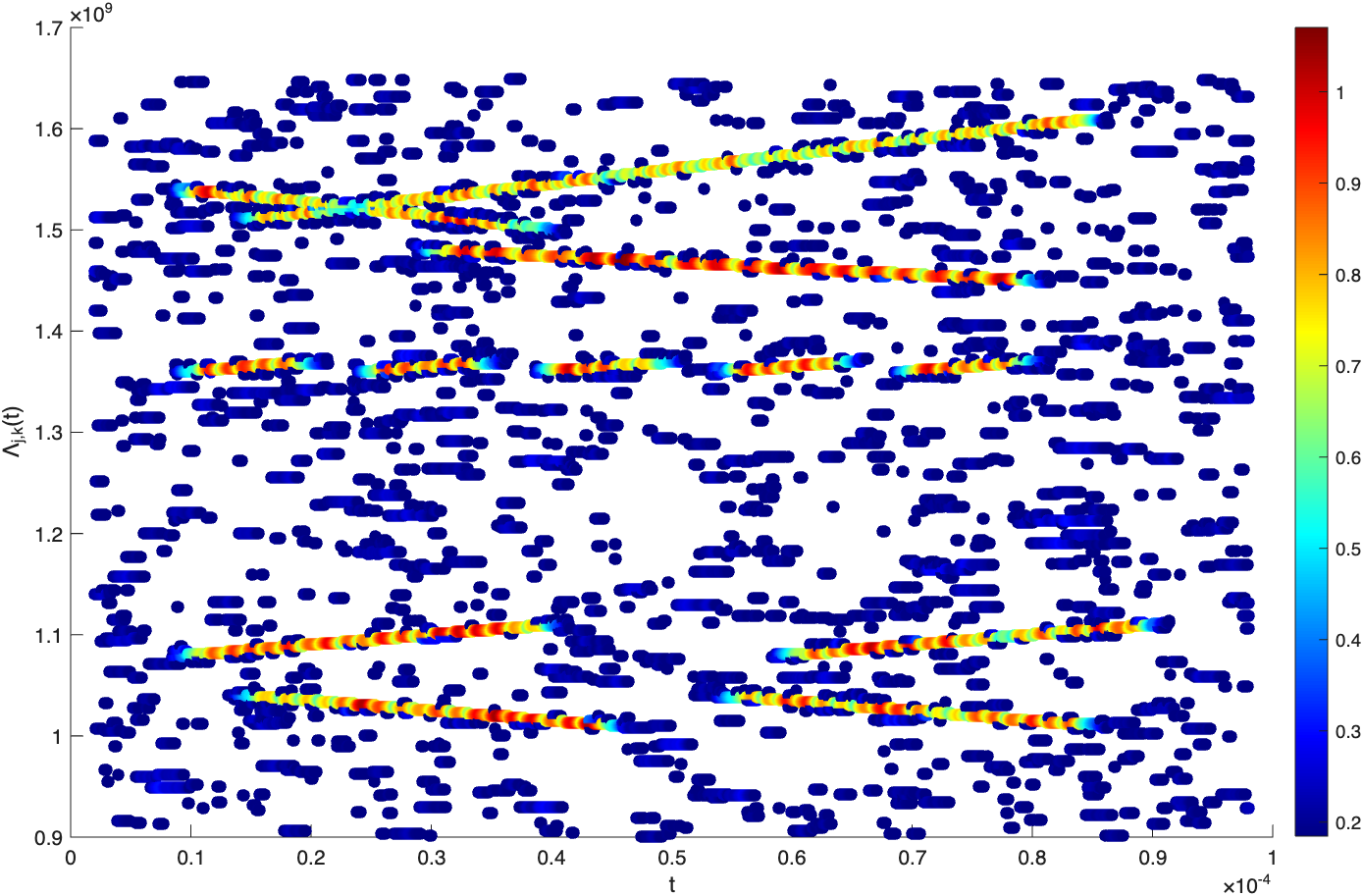}
\end{minipage}
\begin{minipage}{0.3\textwidth}
\includegraphics[width=\textwidth]{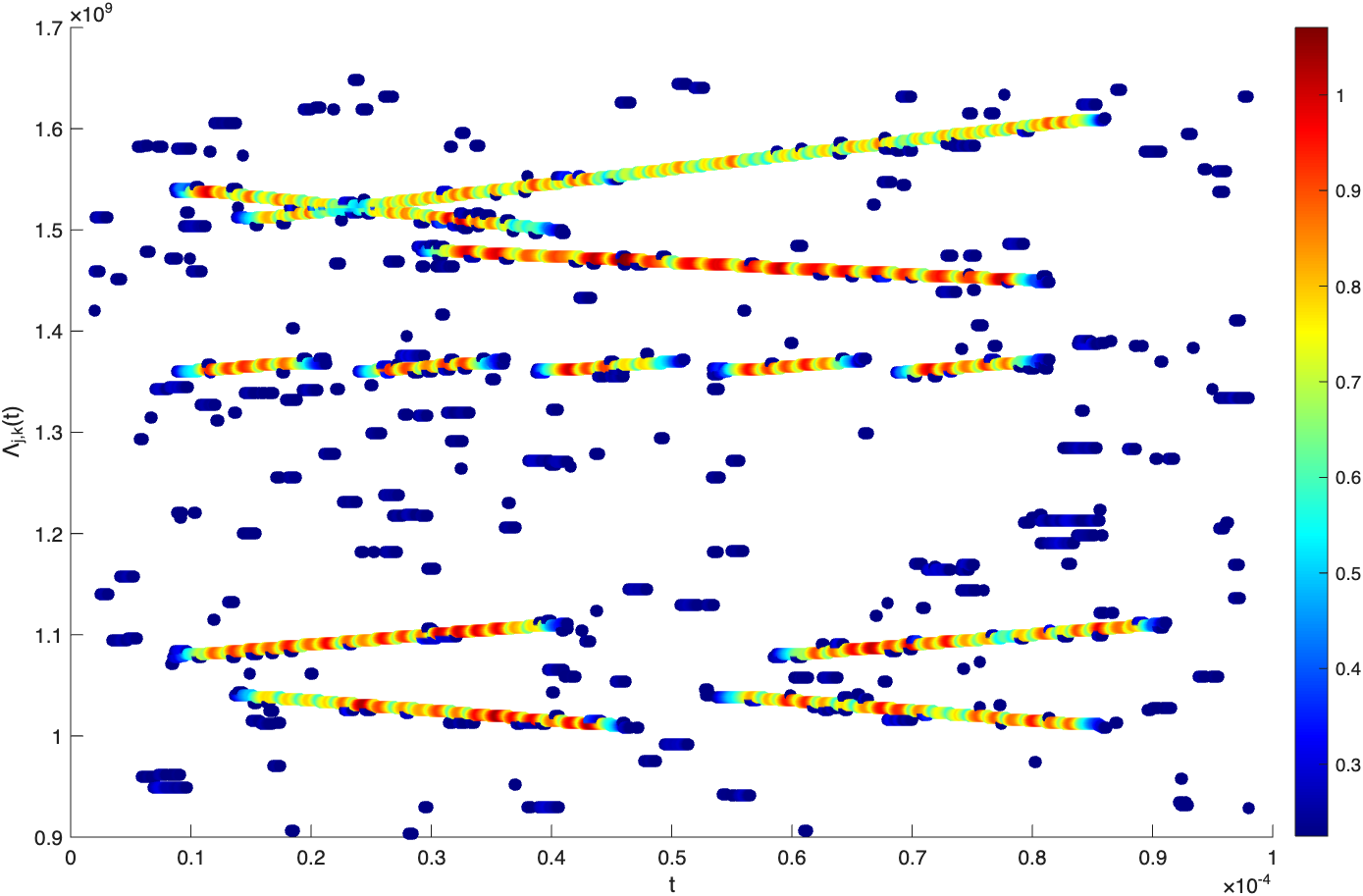}
\end{minipage}
\begin{minipage}{0.3\textwidth}
\includegraphics[width=\textwidth]{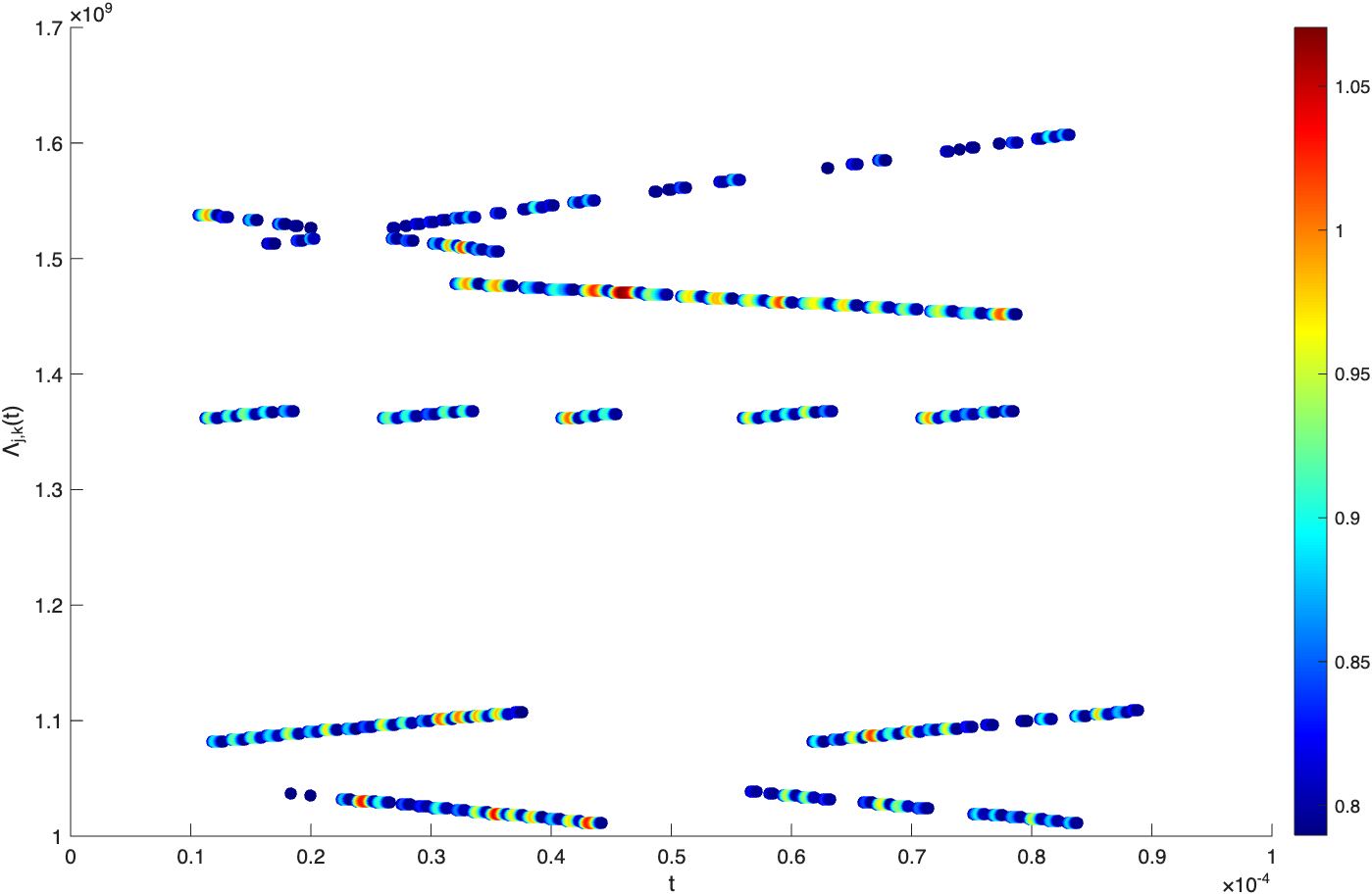}
\end{minipage}
\end{center}



\caption{The effect of selecting different thresholds on signals with 1 GHz sampling rate at -10 dB SNR. On the top, there are three threshold selections with the corresponding SSO diagrams under each selection. It is clear from the left two pair of figures that selecting too low a threshold leads to a very small noise reduction, while the rightmost pair of figures shows that selecting it too high leads to a loss of signal. In either case, the DBSCAN on Line 3 in Algorithm 2 will fail to detect the signal clusters.  }
\label{fig:thresholdselect}
\end{figure}

An alternative approach to select the threshold is to take the histogram of the maximum peaks of the SSO figure in each snippet (maximum value for each $\Delta$), then use Kernel Density Estimation (KDE) to estimate the corresponding probability density. 
If the signal dominates the noise level, the histogram will clearly show the noise part and the signal part in the density estimator appearing as a bimodal distribution.
The lowest point of the density estimator can serve as a globally determined threshold or an alternative method can be used to determine a suitable point separating the noise from the signal.
In Figure~\ref{fig:kde_vs_local}, we demonstrate that this KDE-based approach leads to the same threshold as described in Section~\ref{bhag:alg1} and hence, the same clustering as before. 
When the noise level dominates the signal, the histogram is unimodal, and this approach will not work, in which case a suitably high threshold can be chosen, and iterated on using visual inspection of the signal or some application specfic evaluation metric.

\begin{figure}
\begin{center}
\begin{minipage}{0.3\textwidth}
\includegraphics[width=\textwidth]{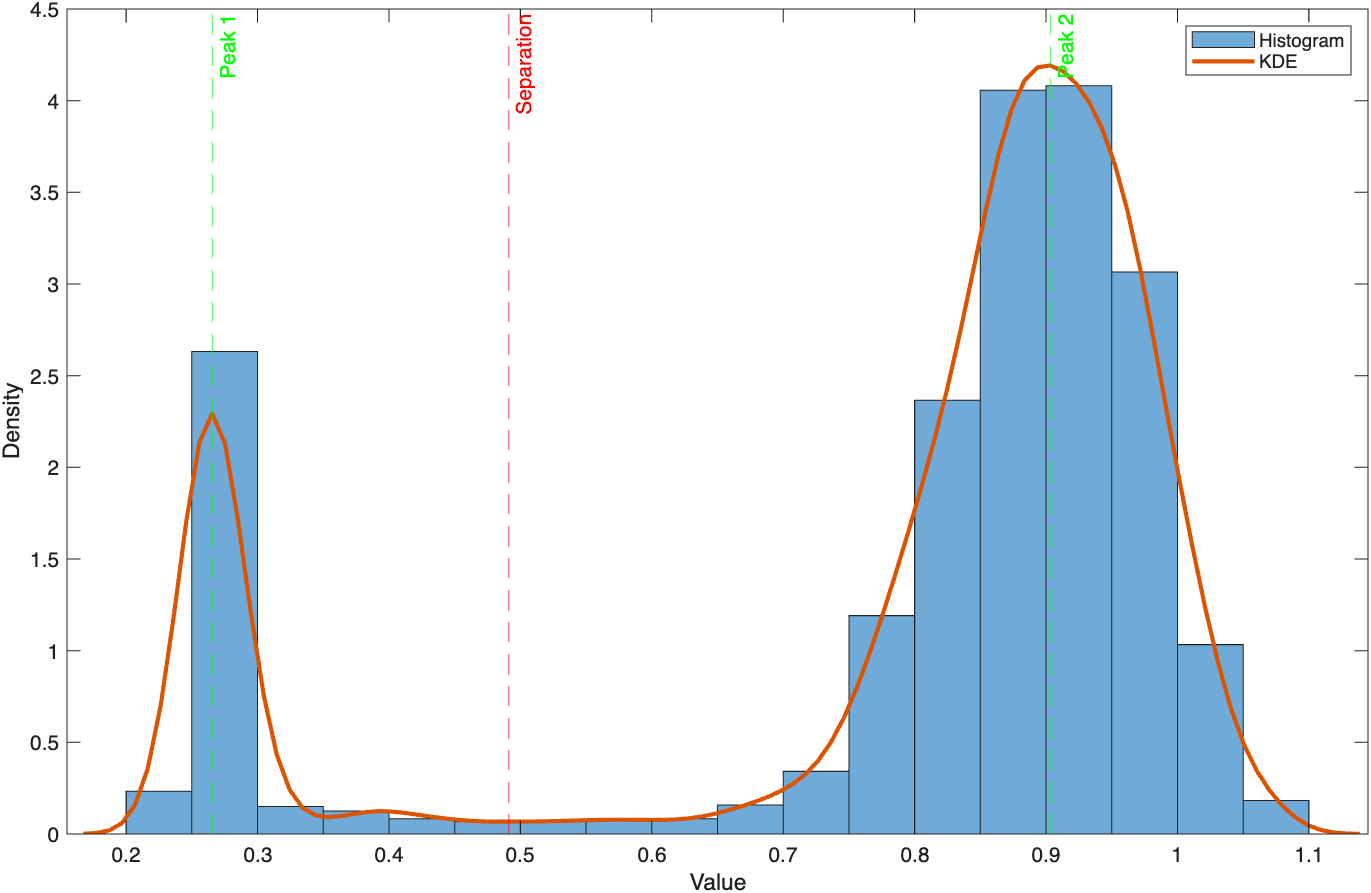}
\end{minipage}
\begin{minipage}{0.3\textwidth}
\includegraphics[width=\textwidth]{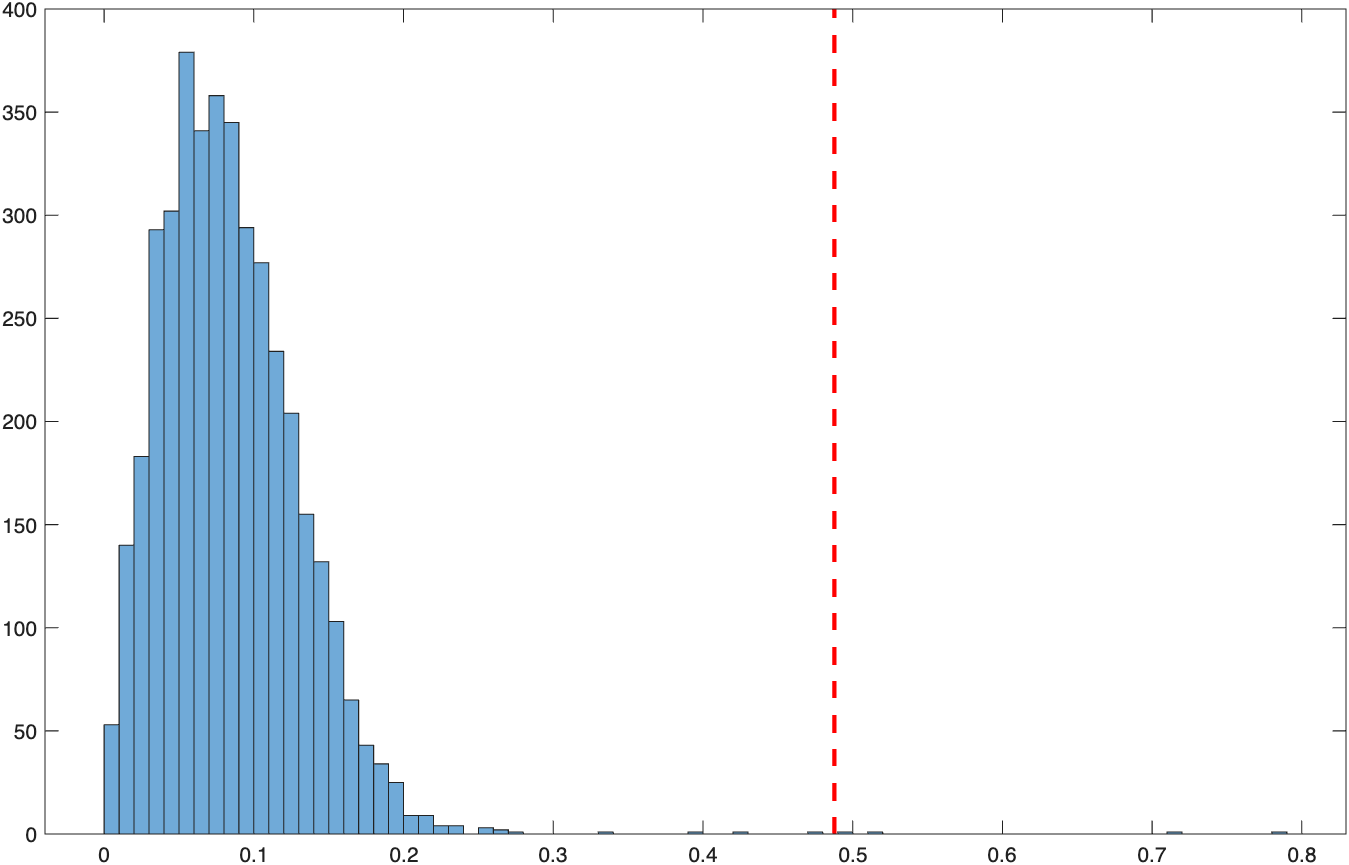}
\end{minipage}
\begin{minipage}{0.3\textwidth}
\includegraphics[width=\textwidth]{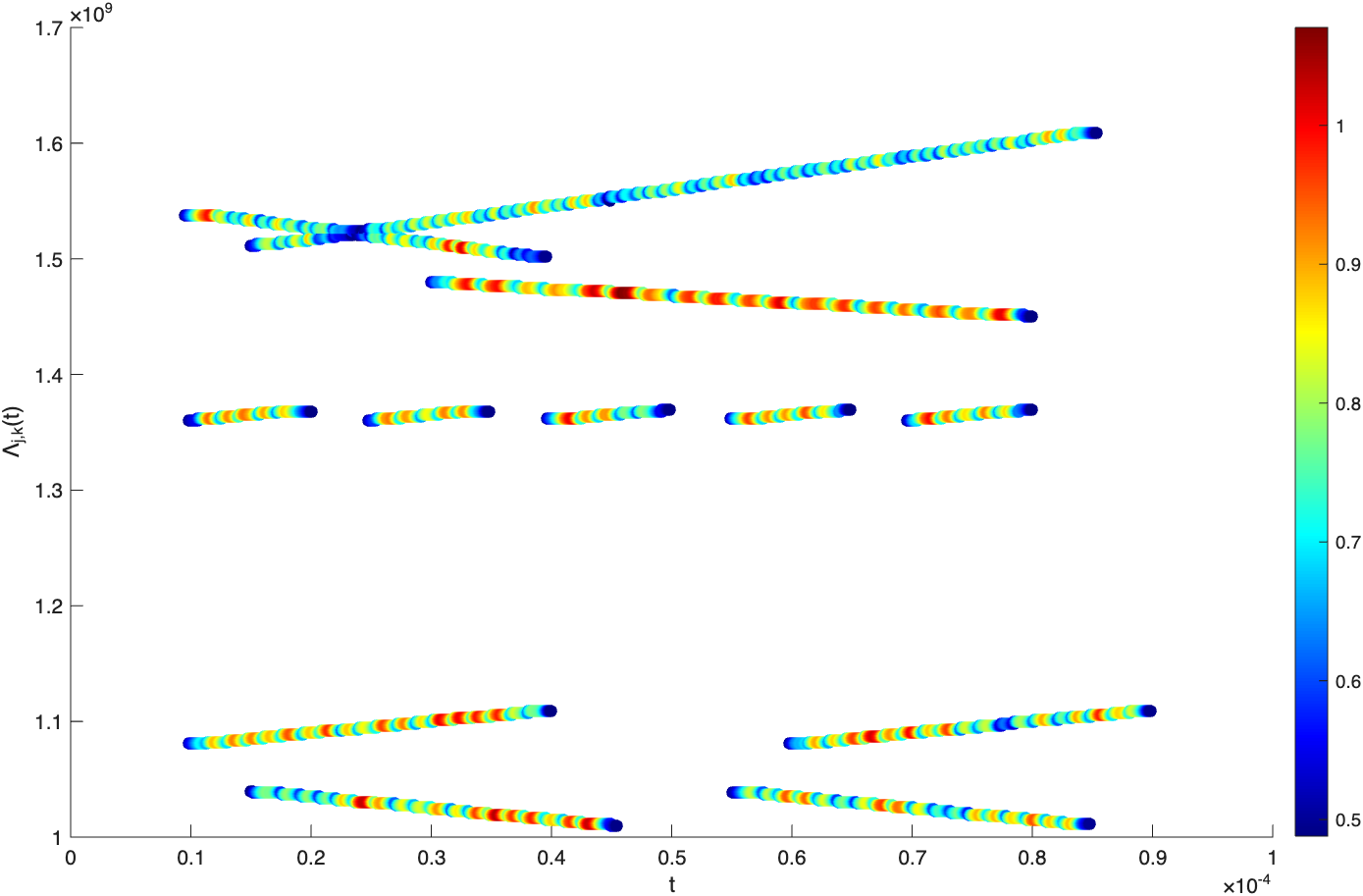}
\end{minipage}
\begin{minipage}{0.44\textwidth}

\end{minipage}
\end{center}
\begin{center}

\end{center}
\caption{Selection of the threshold globally based on a kernel density estimation. With the sampling rate of 1GHz, the threshold can be determined satisfactorily for SNR=-10dB as the figure on the  left shows, The middle figure shows that this threshold is the same as the one described in Section~\ref{bhag:alg1}, and the right figure shows the resulting clusters same as the one described in Section~\ref{bhag:alg1}. }
\label{fig:kde_vs_local}
\end{figure} 


\newpage
\color{black}
\bhag{Numerical results}\label{bhag:numerical_results}
To assess the robustness and performance of the proposed Signal Separation Operator (SSO) method, we conducted a series of controlled numerical experiments using simulated data consisting of LFM chirp signals described in Section~\ref{bhag:datagen}, the most commonly used type of radar waveform.  These experiments were designed to evaluate the impact of key signal characteristics, including minimal frequency separation, SNR, sampling rate, and the presence of frequency crossovers.  The details for all the experiments are listed in Tables~\ref{tab:result_table_1} and \ref{tab:result_table_2} in the Appendix.

While the focus of our experiments is on LFM chirp waveforms, we do remark that there are other types of modulated waveforms that are sometimes used, such as non-linear frequency modulation (NLFM), phase-shift keying (PSK), and noise waveforms.  These are beyond the scope of this paper and require a different theoretical treatment; however, the SSO method will handle these cases to different degrees of success.  For example, in the case of NLFM waveforms, the method will accurately recover the instantaneous frequency, and the clustering steps will perform the same as they do in this case.  The area where our algorithm will require modification is the parameterization of the pulses since a non-linear function will be required to fit the pulses. The detection of signal crossings will become more sensitive to the minimal separation and snippet size studied in Section~\ref{bhag:tunable}.


\begin{figure}[H]
\begin{center}
\includegraphics[scale=.145]{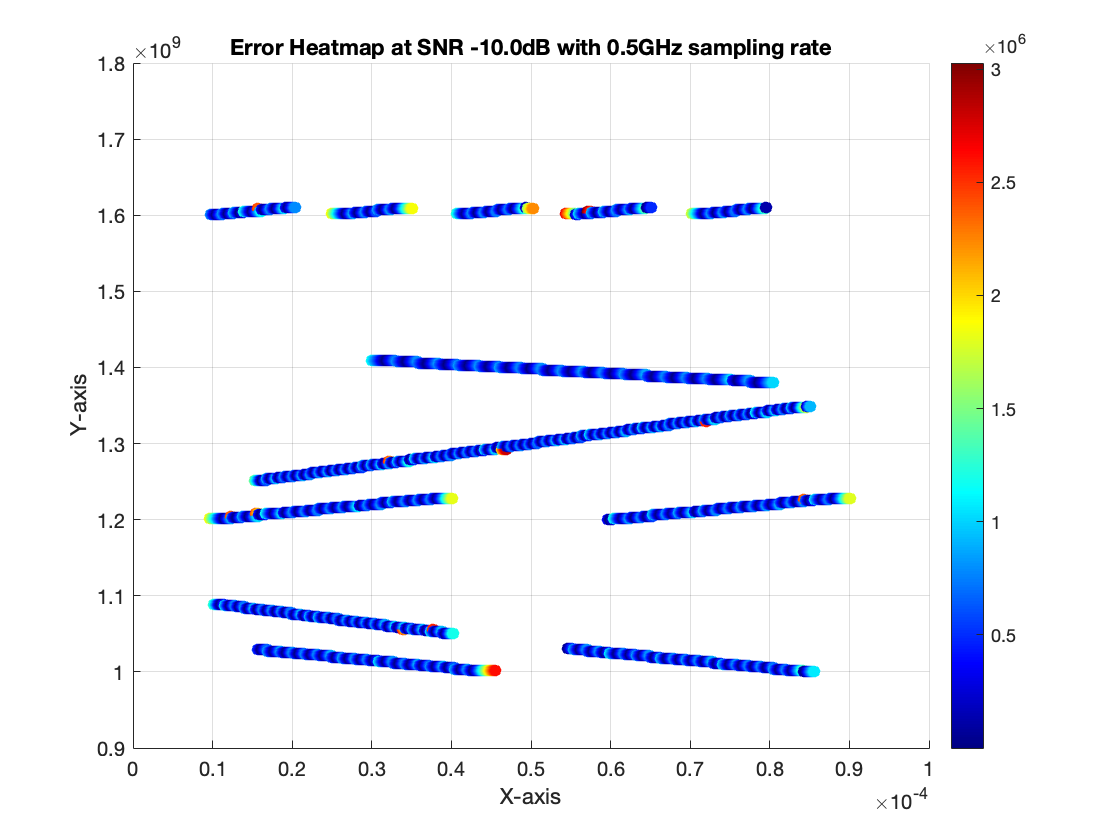}
\includegraphics[scale=.145]{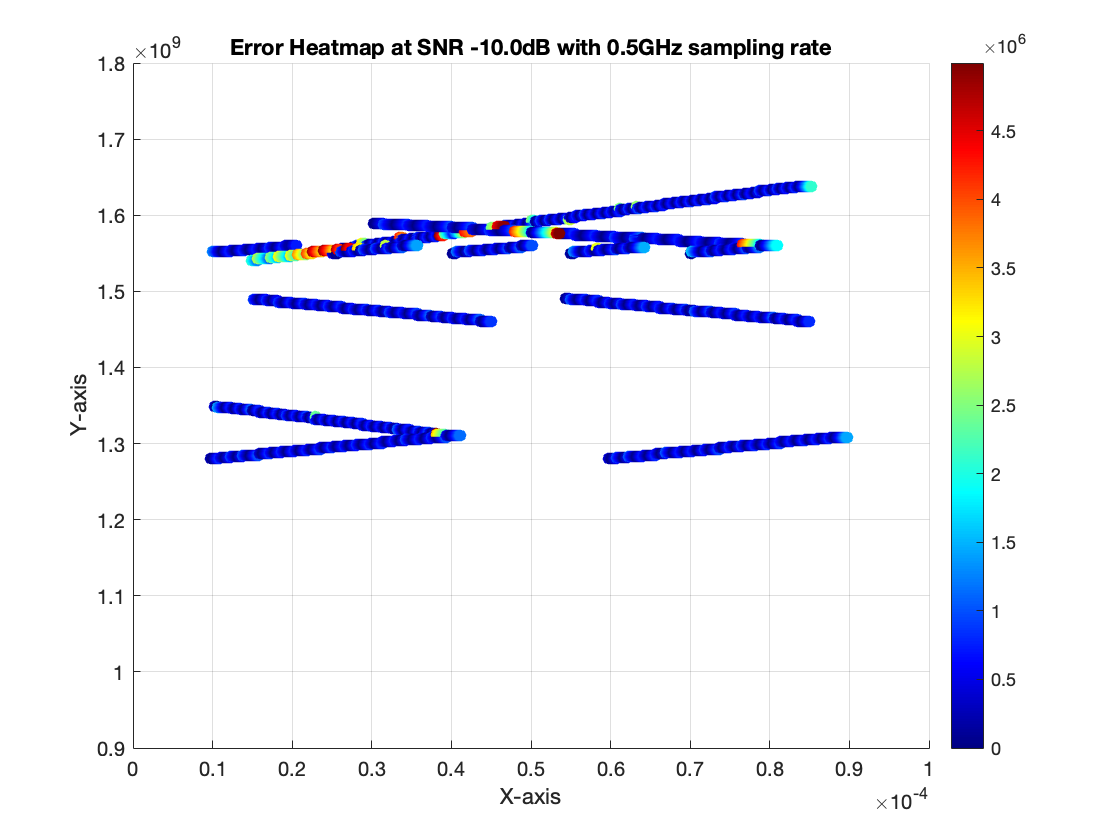}
\includegraphics[scale=.145]{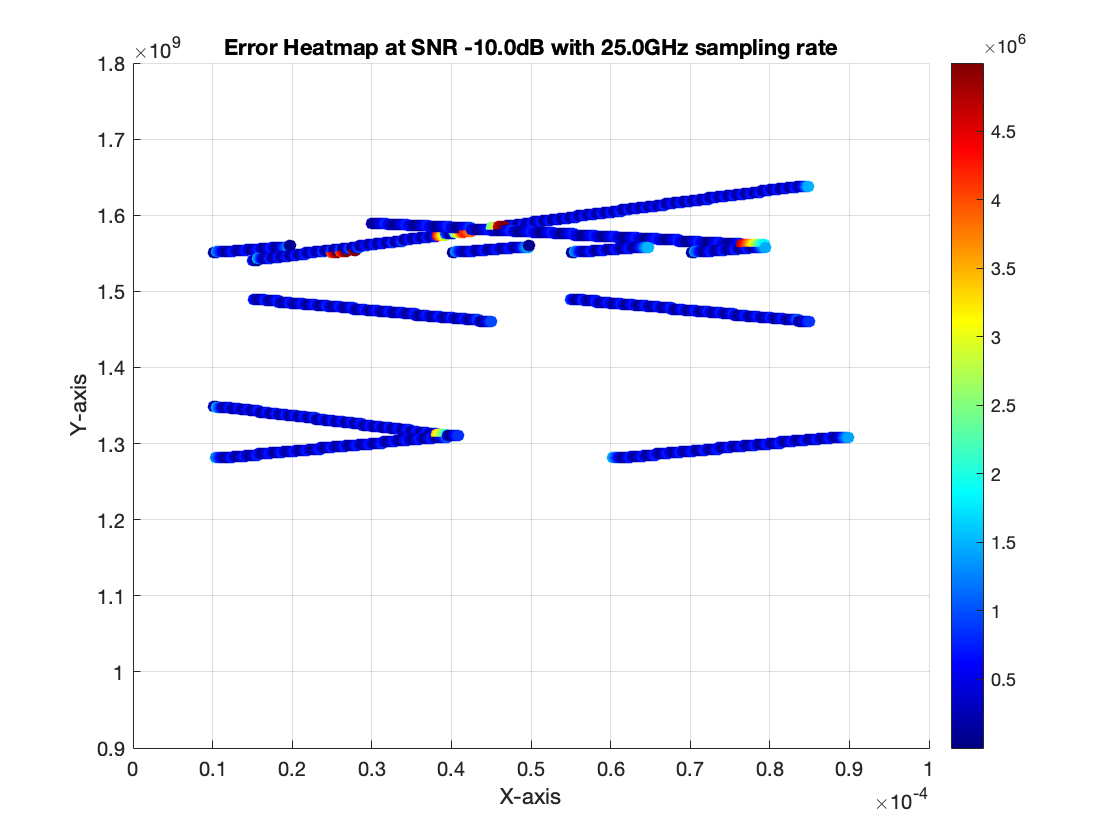}
\end{center}
\caption{Left and middle: The error heatmap plots between SSO results and the corresponding regression line for example  2 and 4 with SNR -10 dB and sampling rate 0.5 GHz respectively. Right: Example 4 with the sampling rate increased to 25GHz.
 Our algorithm works well in case of no crossover and clear crossover signals, but does not work so well when the minimum separation among signals is low. 
 The accuracy improves with a higher sampling rate.}
\label{fig:4results}
\end{figure}

\subsection{Minimal separation}\label{bhag:minsep}
The algorithm demonstrates strong performance when signal components are sufficiently separated in frequency. As illustrated in Figure \ref{fig:4results}, the method accurately recovers the instantaneous frequencies and associated parameters in both non-intersecting and cleanly intersecting cases. However, performance degrades in scenarios where the minimal separation between components becomes too small (i.e. for 0.5 GHz sampling rate, we use $\Delta = 10^{-6}$ as suggested in Table \ref{tab:Delta_vs_slope}, $n=1000,$ and $\eta = 0.01$ as suggested in Table \ref{tab:n_vs_eta}). In such cases, the spectral peaks associated with different components may become indistinct, leading to inaccurate clustering and parameter estimation. This limitation is primarily due to the inherent resolution bounds imposed by the kernel size and the sampling interval. To address this issue, a higher sampling rate is required in order to increase $n$ and decrease $\eta$ so that different components become distinct. In Figure \ref{fig:4results}, we increase the sampling rate to 25 GHz. As a result, $n$ increases to 50000 and $\eta$ decreases to 0.0002.

\subsection{Robustness to noise}\label{bhag:robust}

The method remains effective even under high noise conditions. 
To illustrate, we compute the the root mean square error (RMSE) in each experiment by
\be\label{eq:exptrmse}
\mbox{RMSE}=\mbox{mean}\left(\sqrt{\frac{1}{D}\sum_{k=1}^D\sum_{j=1}^{J_k}\left(\frac{\phi_{j,k}'(t_k)-\widehat{\phi_{j,k}'(t_k)}}{\phi_{j,k}'(t_k)}\right)^2}\right).
\ee
In our experiments, we took the mean over 16 trials for each choice of the signal, the sampling rate, and SNR.
As shown in Figure \ref{fig:4rmse}, the root mean square error (RMSE) remains low for SNR levels ranging from $10$ dB to $-30$ dB. 
The algorithm maintains RMSE values within acceptable bounds, provided that the sampling rate is sufficiently large. 
The combination of localized kernel averaging and peak detection in the frequency domain allows the method to suppress noise effectively.

\begin{figure}[H]
\includegraphics[scale=.109]{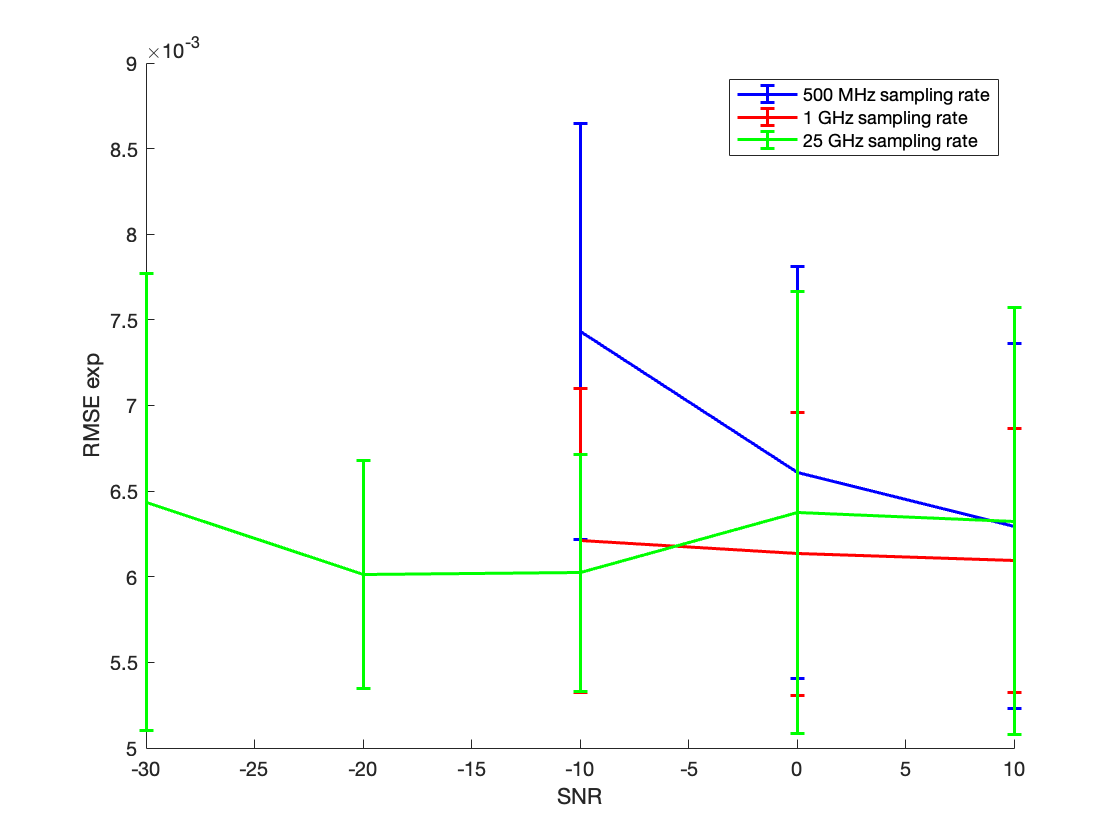}
\includegraphics[scale=.109]{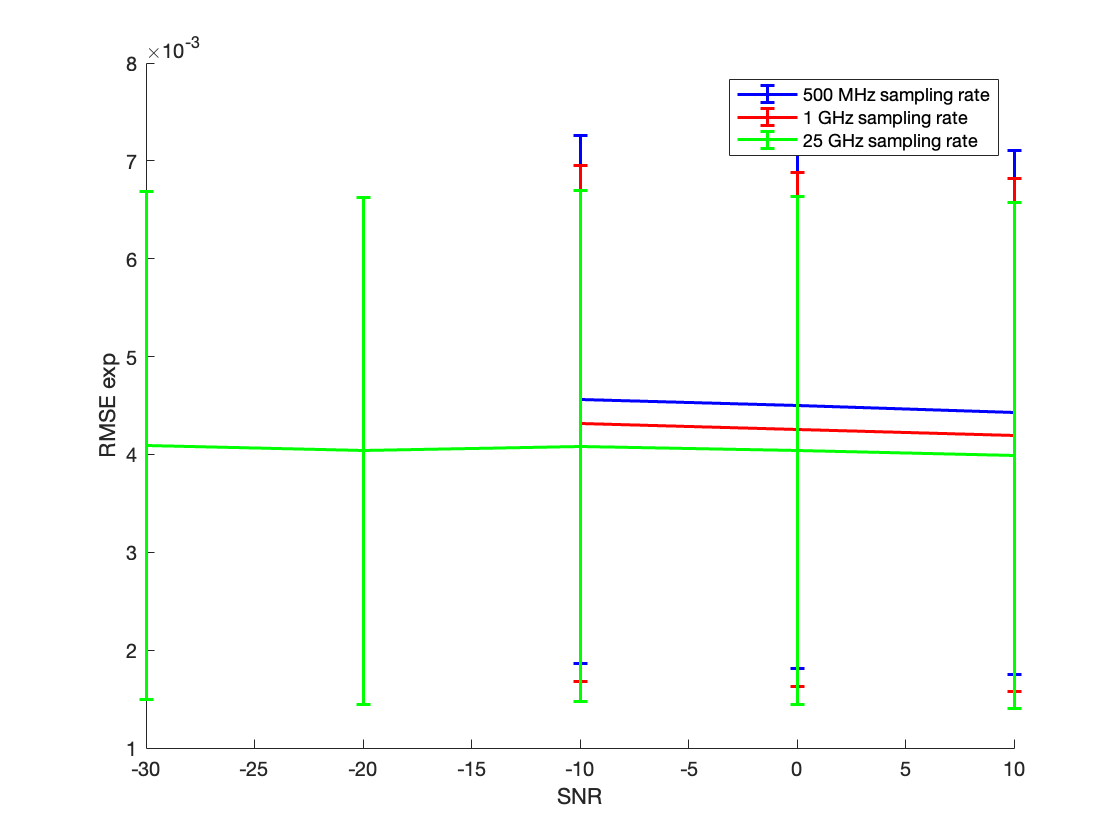}
\includegraphics[scale=.109]{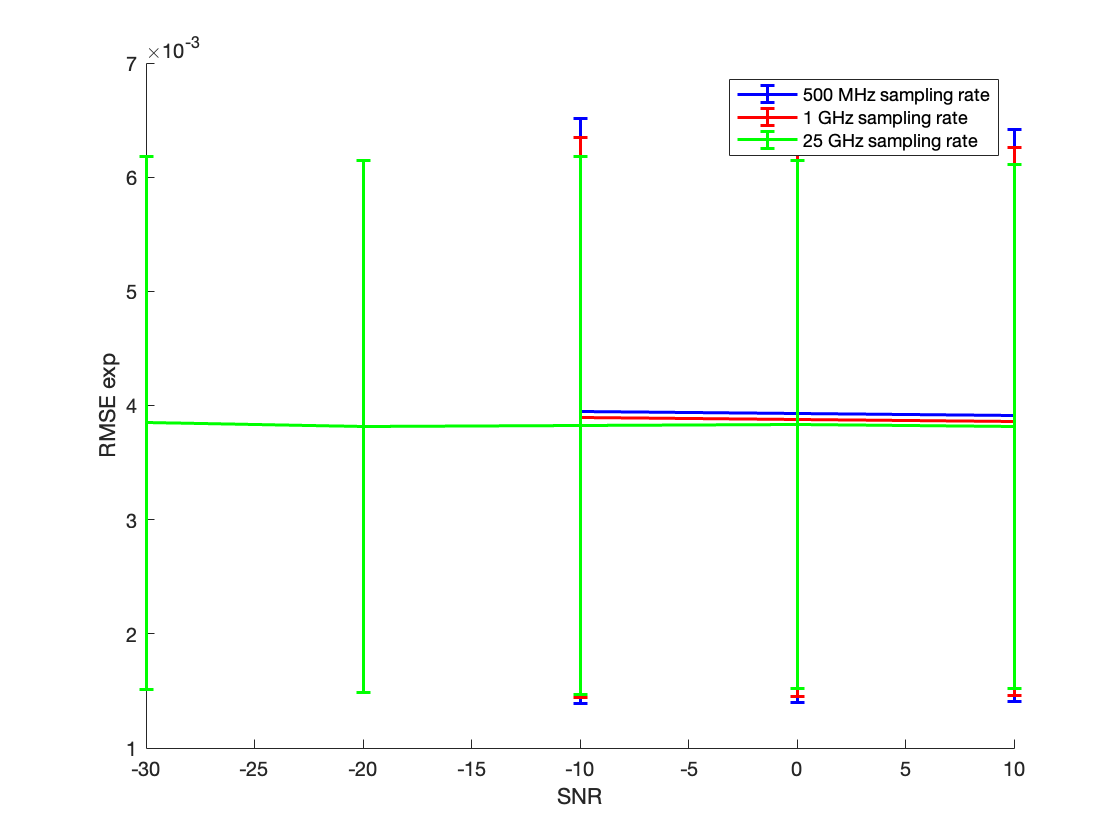}
\includegraphics[scale=.109]{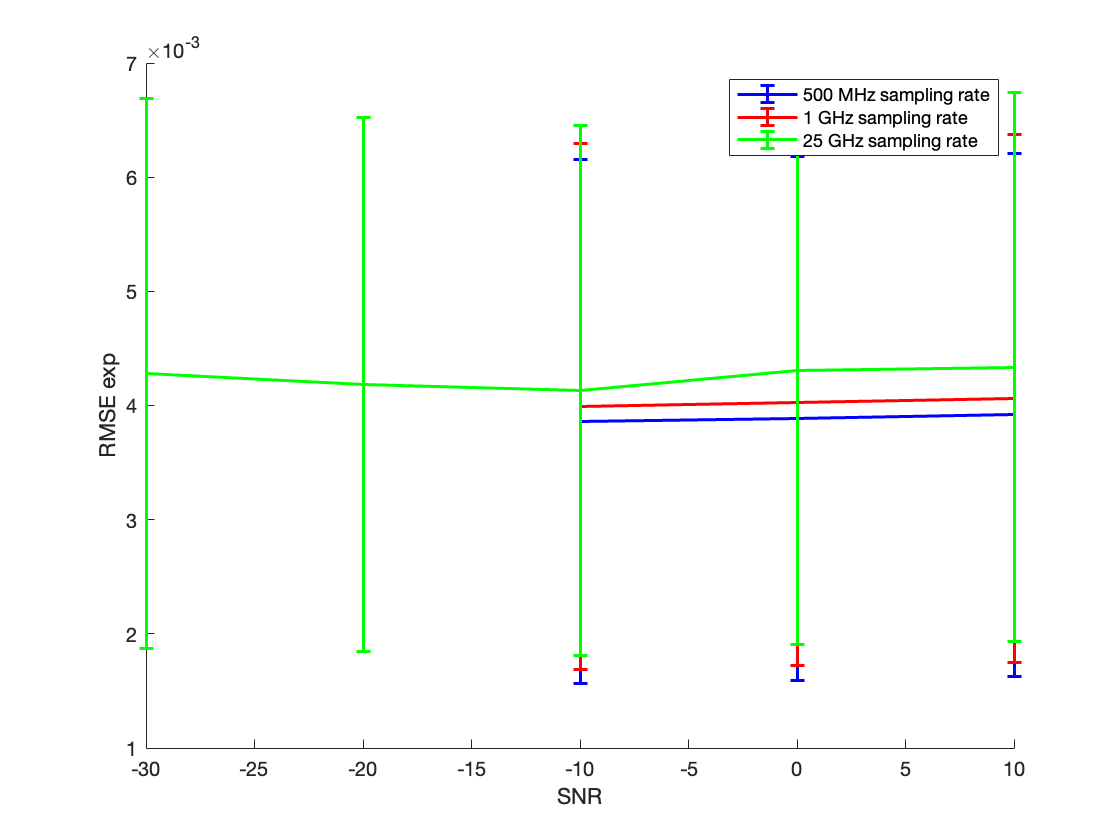}
\caption{The sampling rate performance plots of vary SNR vs RMSE for example 1, 2, 3, and 4. The RMSE was comparing between the ground truth and the estimation parameters of $\phi(t)$ in \eqref{eq:pulsephasedef}. }
\label{fig:4rmse}
\end{figure}

\subsection{Comparison with SST}\label{bhag:sst}
We recall that the most often used method for finding IF's is EMD, which is purely heuristic.
The method known as the synchrosqueezing transform method (SST) is based on a solid mathematical foundation and works better than EMD.
In this section, we compare our results with those obtained by using SST.
As mentioned earlier, there are many versions of SST. 
We use the implementation of the SST algorithm as described in Section III of \cite{thakur2013synchrosqueezing}.

In the experiment, we use the code from GitHub repository (\texttt{https://github.com/ebrevdo/synchrosqueezing}) from \cite{thakur2013synchrosqueezing}.
We use the same parameters used by the authors in their examples, provided in the code, which use the shifted Gaussian (or Morlet) wavelet with zero mean.
We refer the reader to \cite{thakur2013synchrosqueezing} for a more detailed discussion of the approach and parameters.

Figure \ref{fig:base_sst} shows that our implementation of this algorithm works for certain signals at high SNR levels and deteriorates at lower levels.
We observed that the maximum frequency that the SST method can take is $0.2$ GHz at sampling rate $0.5$ GHz while our dataset has maximum frequency at $1.6$ GHz.
Figure \ref{fig:sst_0.5ghz} illustrates that the SST fails to work in the regime in which we are interested in this paper.
In Figure \ref{fig:sst_1ghz}, we extend the experiment sampling rate from 0.5 GHz to 1 GHz. We observe that the SST can detect some hints of the signal, but not as satisfactory level as our SSO method. We also have to note that the computational time for SST is very high that we are not able to run SST on the input signal more than an array of length greater than 10000 (e.g. $1\times 10^{-5}$ seconds for 1 GHz sampling rate).

\begin{figure}[H]

\begin{minipage}{0.3\textwidth}
\includegraphics[width=\textwidth]{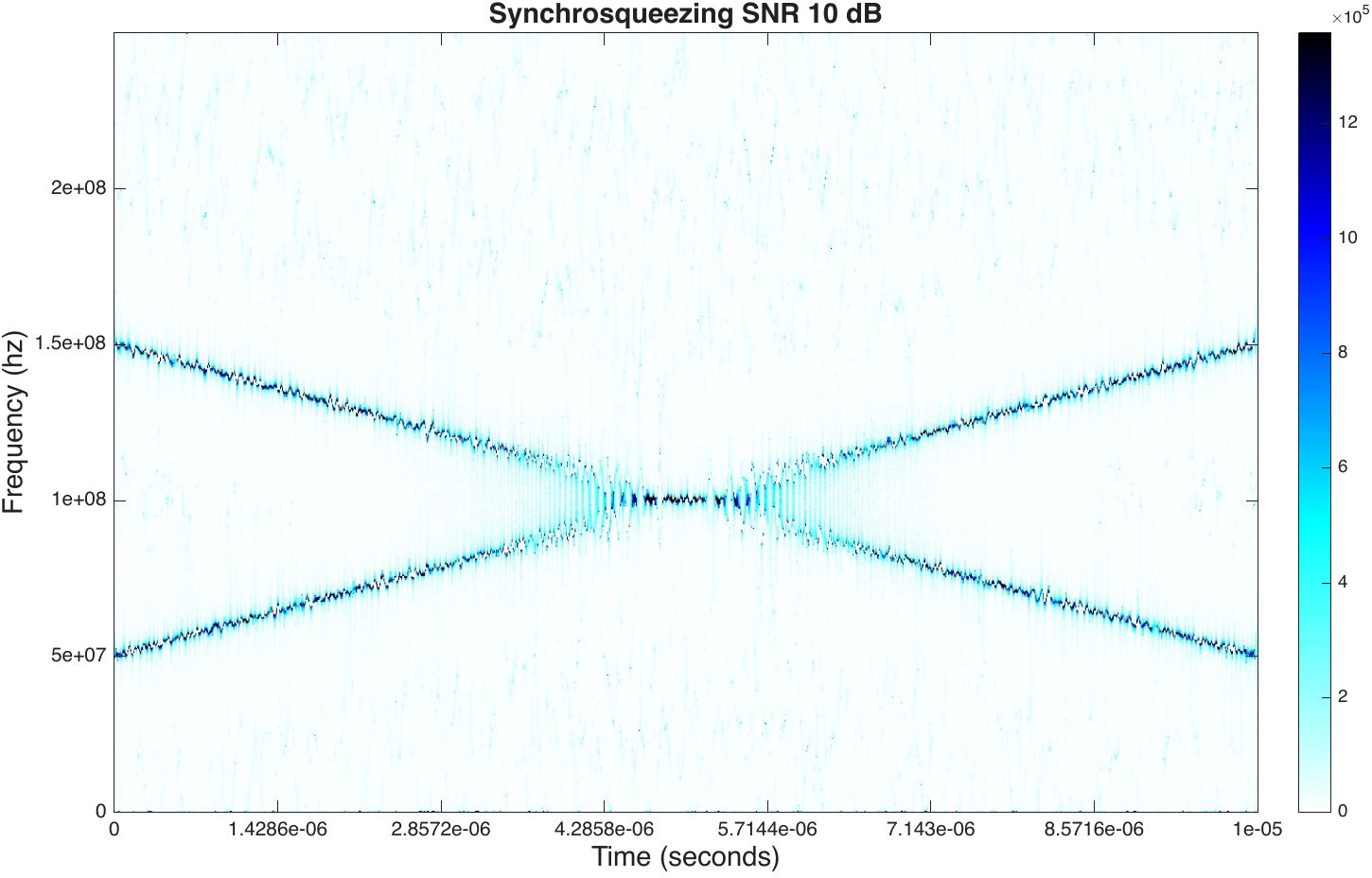}
\end{minipage}
\begin{minipage}{0.3\textwidth}
\includegraphics[width=\textwidth]{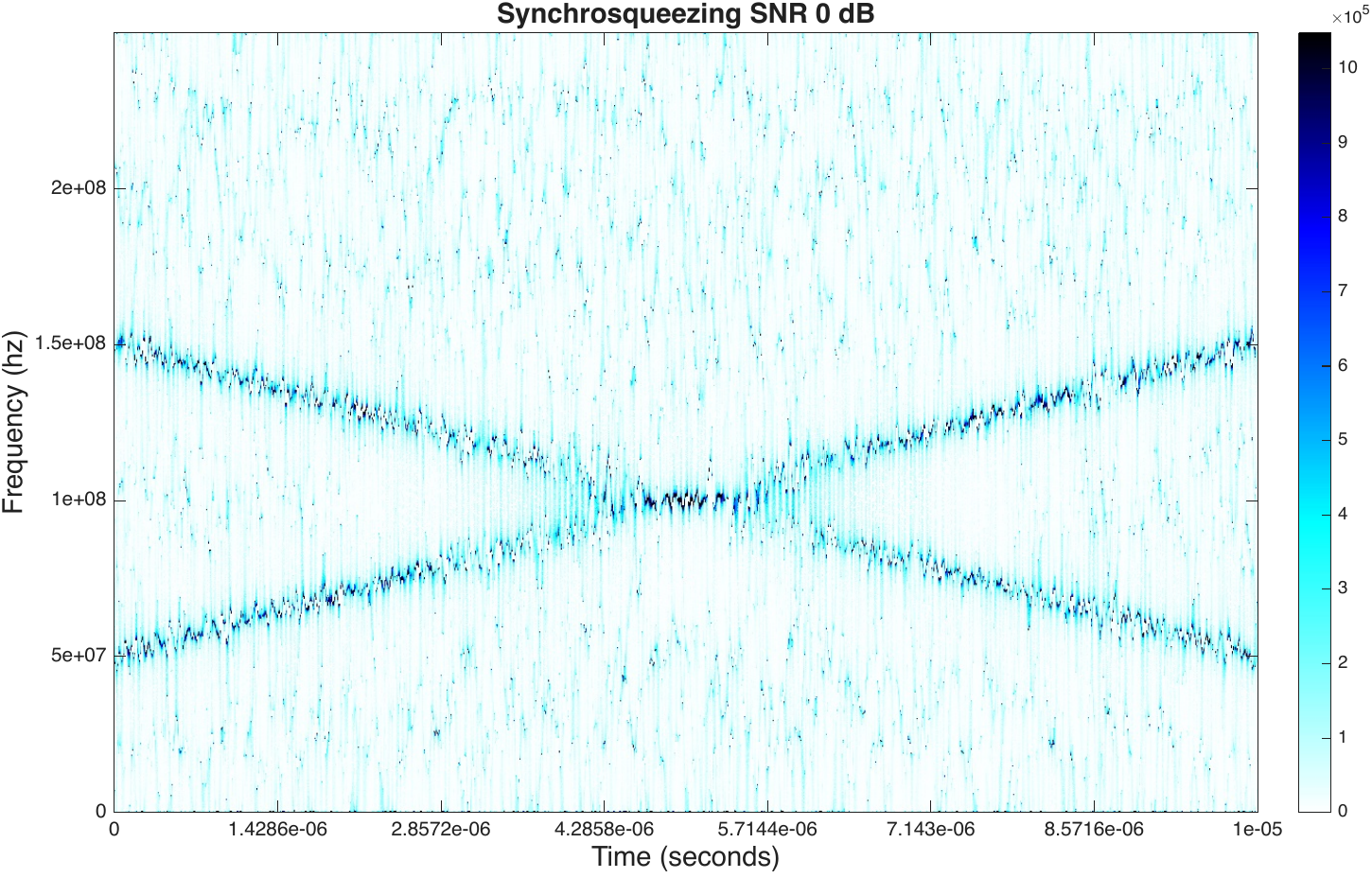}
\end{minipage}
\begin{minipage}{0.3\textwidth}
\includegraphics[width=\textwidth]{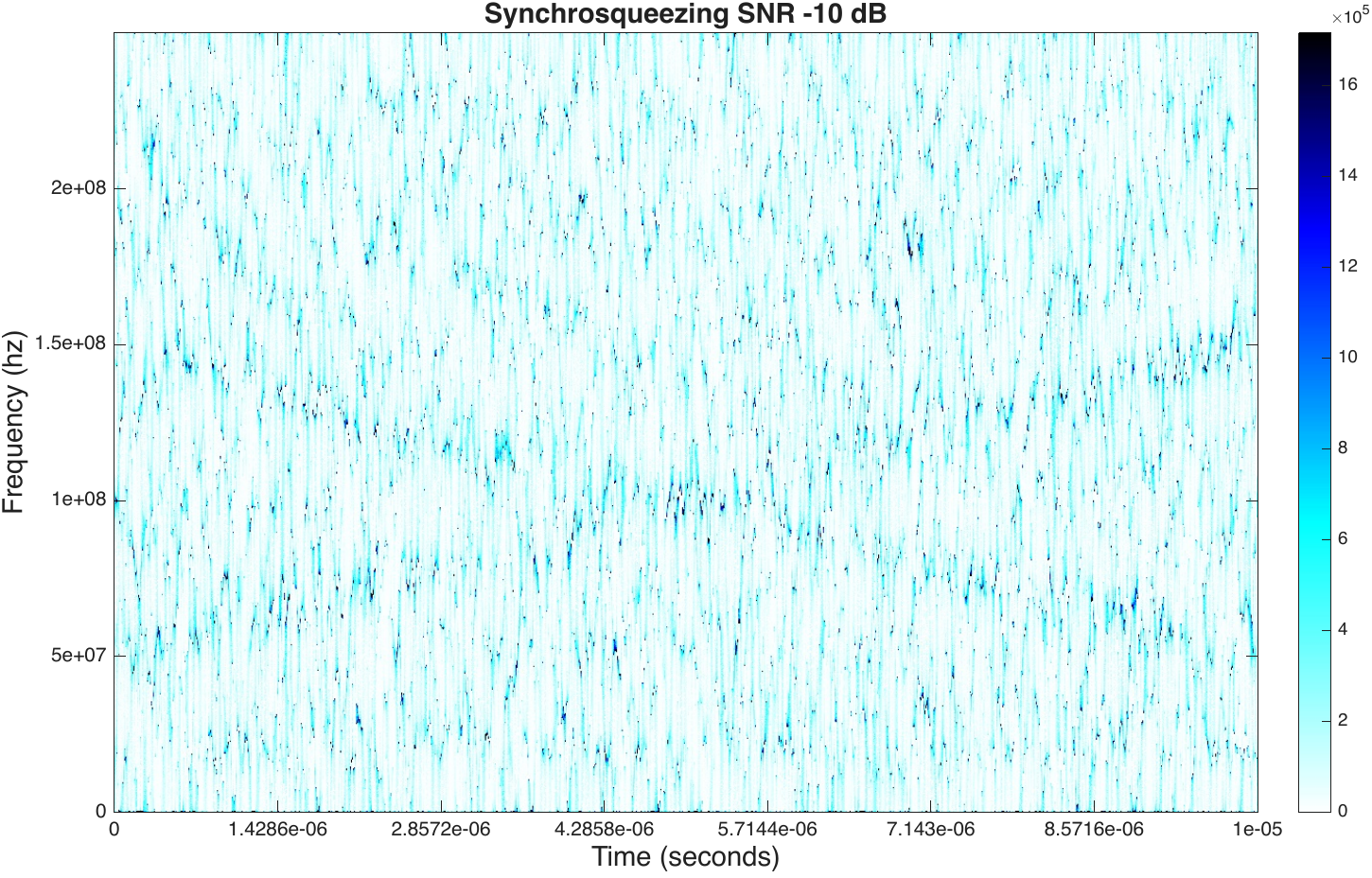}
\end{minipage}

\caption{The examples of SST plot with sampling rate 0.5 GHz at SNR 10 dB (left), SNR 0 dB (middle), and SNR -10 dB (right) for signals $f(t) = e^{2\pi i(15\times 10^7 t-5\times10^{12}t^2)} + e^{2\pi i(5\times 10^7t+5\times10^{12}t^2)}$.}
\label{fig:base_sst}
\end{figure}


\begin{figure}[H]
\begin{minipage}{0.3\textwidth}
\includegraphics[width=\textwidth]{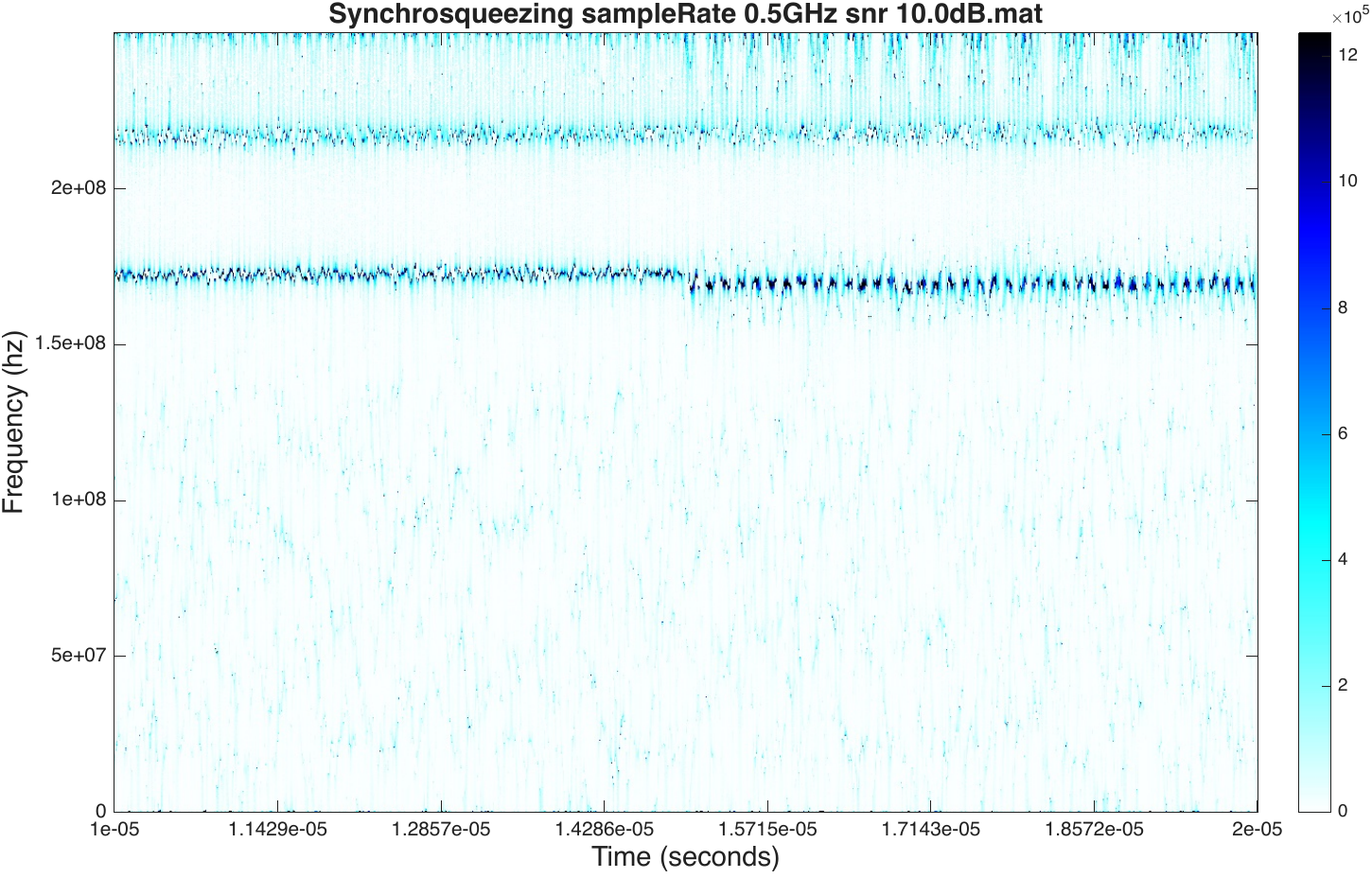}
\end{minipage}
\begin{minipage}{0.3\textwidth}
\includegraphics[width=\textwidth]{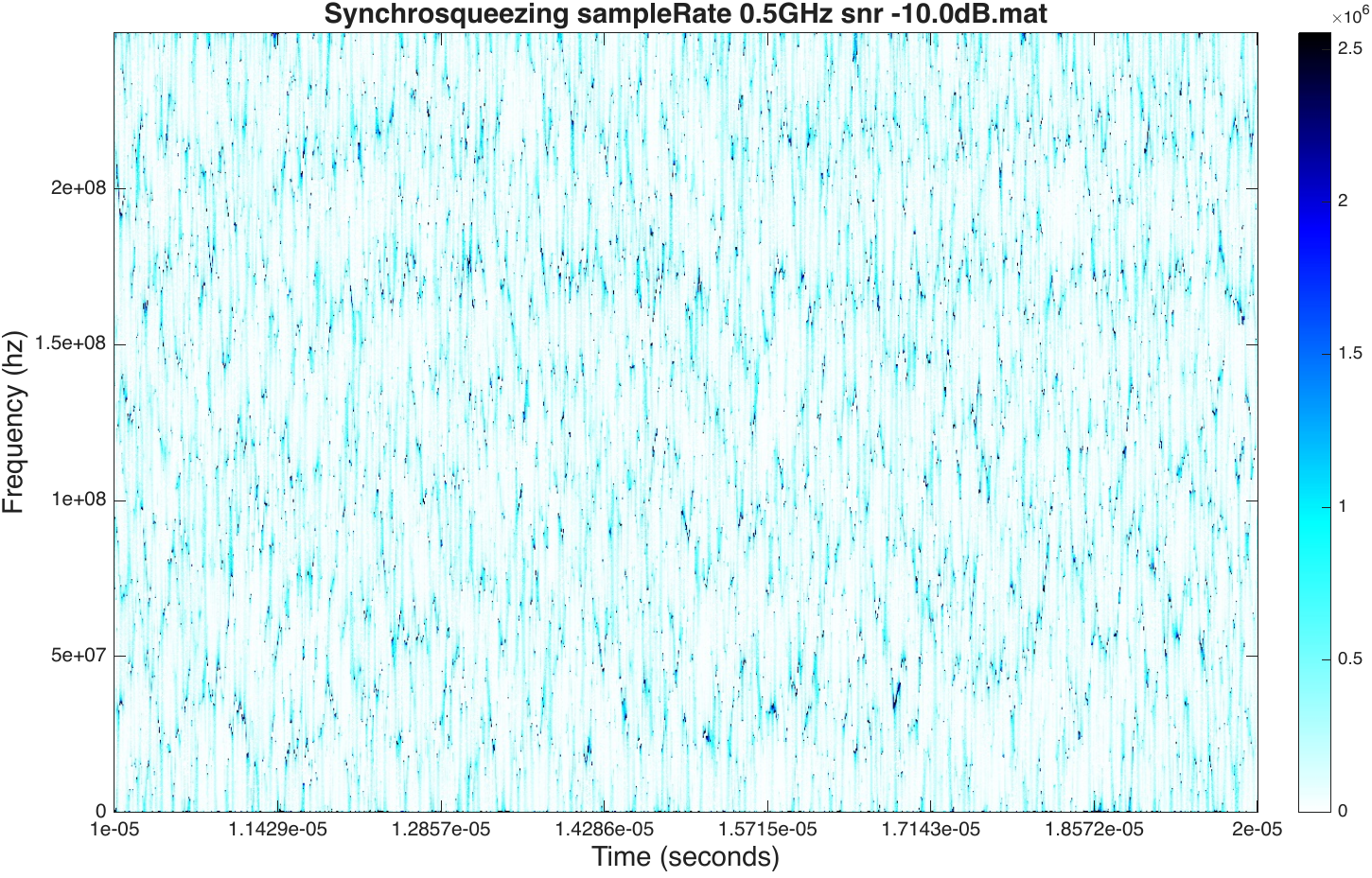}
\end{minipage}
\begin{minipage}{0.3\textwidth}
\includegraphics[width=\textwidth]{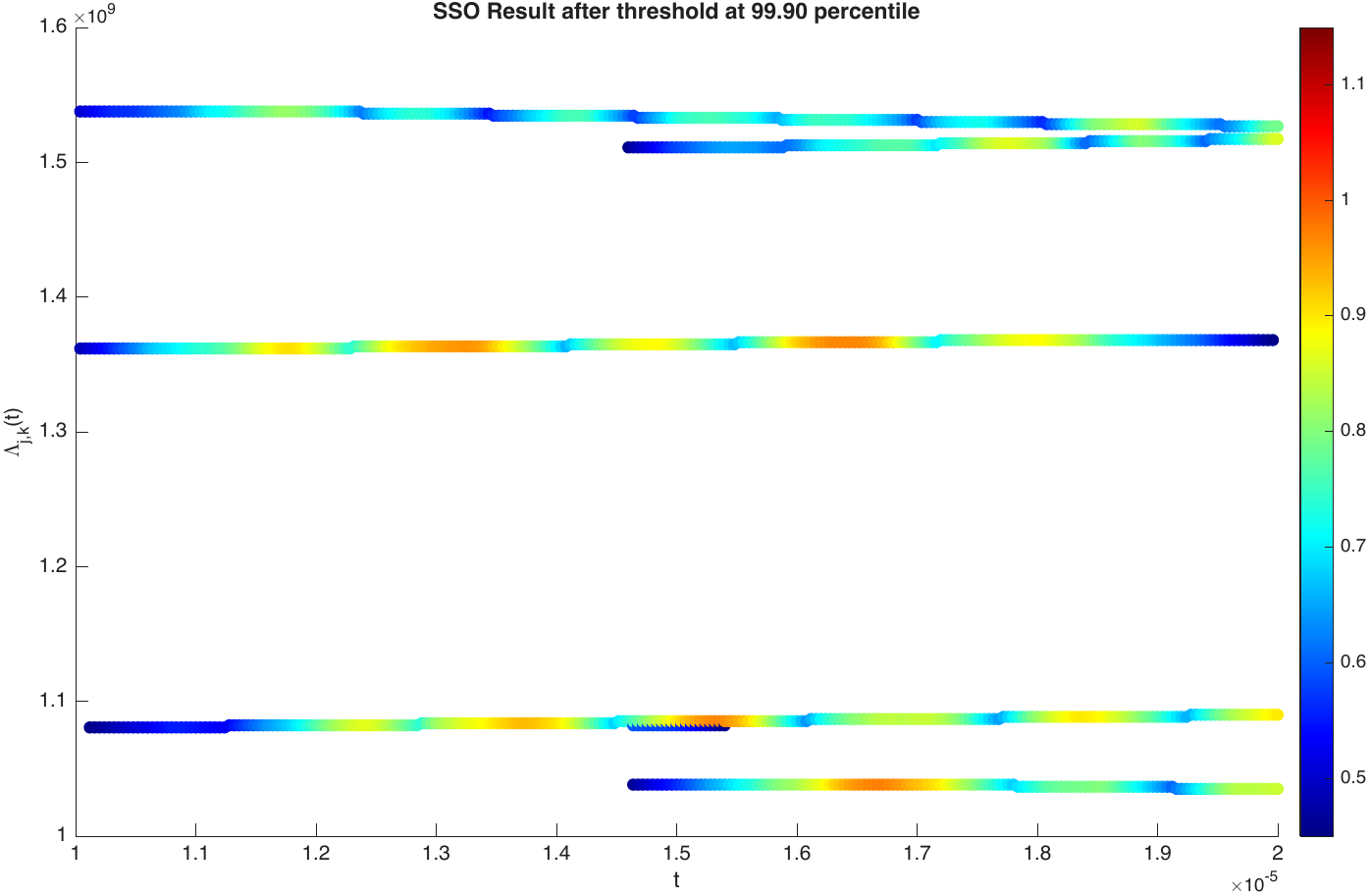}
\end{minipage}

\caption{This figure shows the comparison between SST and SSO method for $t$ between $1 \times 10^{-5}$ and $2 \times 10^{-5}$. The SST plots for example 1 with sampling rate 0.5 GHz at SNR 10,  -10 dB and the SSO plot for example 1 with sampling rate 0.5 GHz at SNR -10 dB. }
\label{fig:sst_0.5ghz}
\end{figure}

\begin{figure}[H]
\begin{minipage}{0.3\textwidth}
\includegraphics[width=\textwidth]{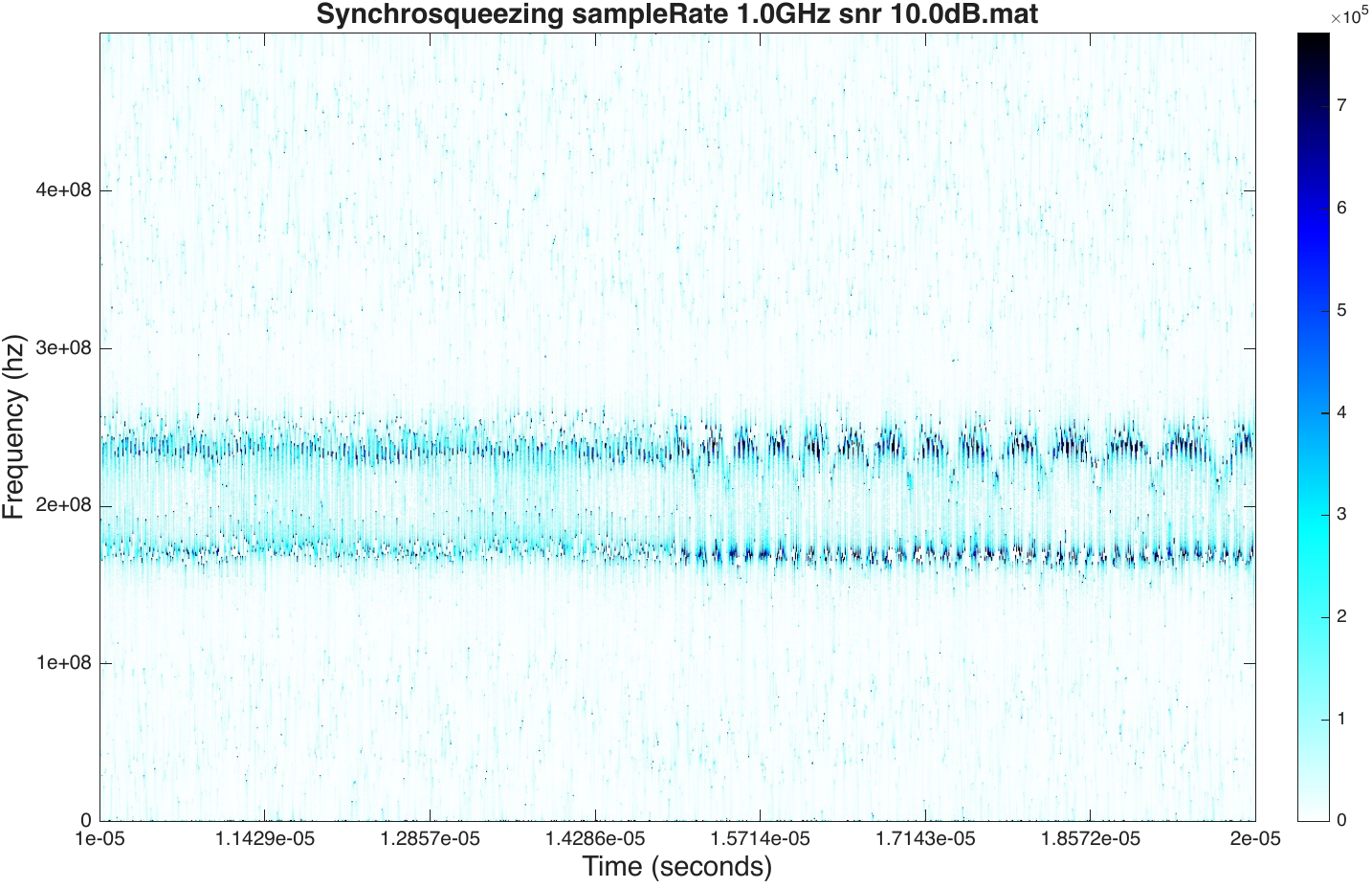}
\end{minipage}
\begin{minipage}{0.3\textwidth}
\includegraphics[width=\textwidth]{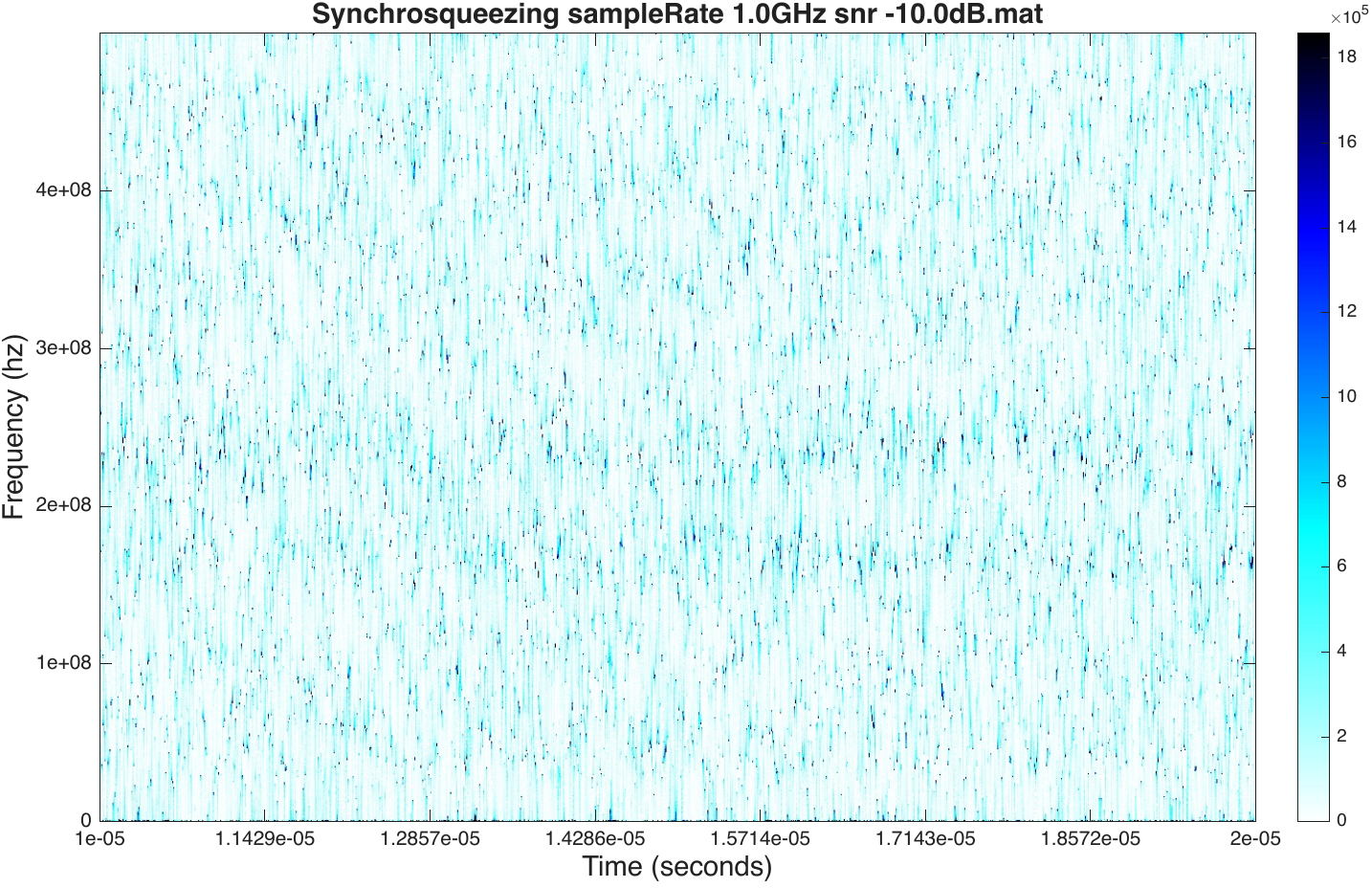}
\end{minipage}
\begin{minipage}{0.3\textwidth}
\includegraphics[width=\textwidth]{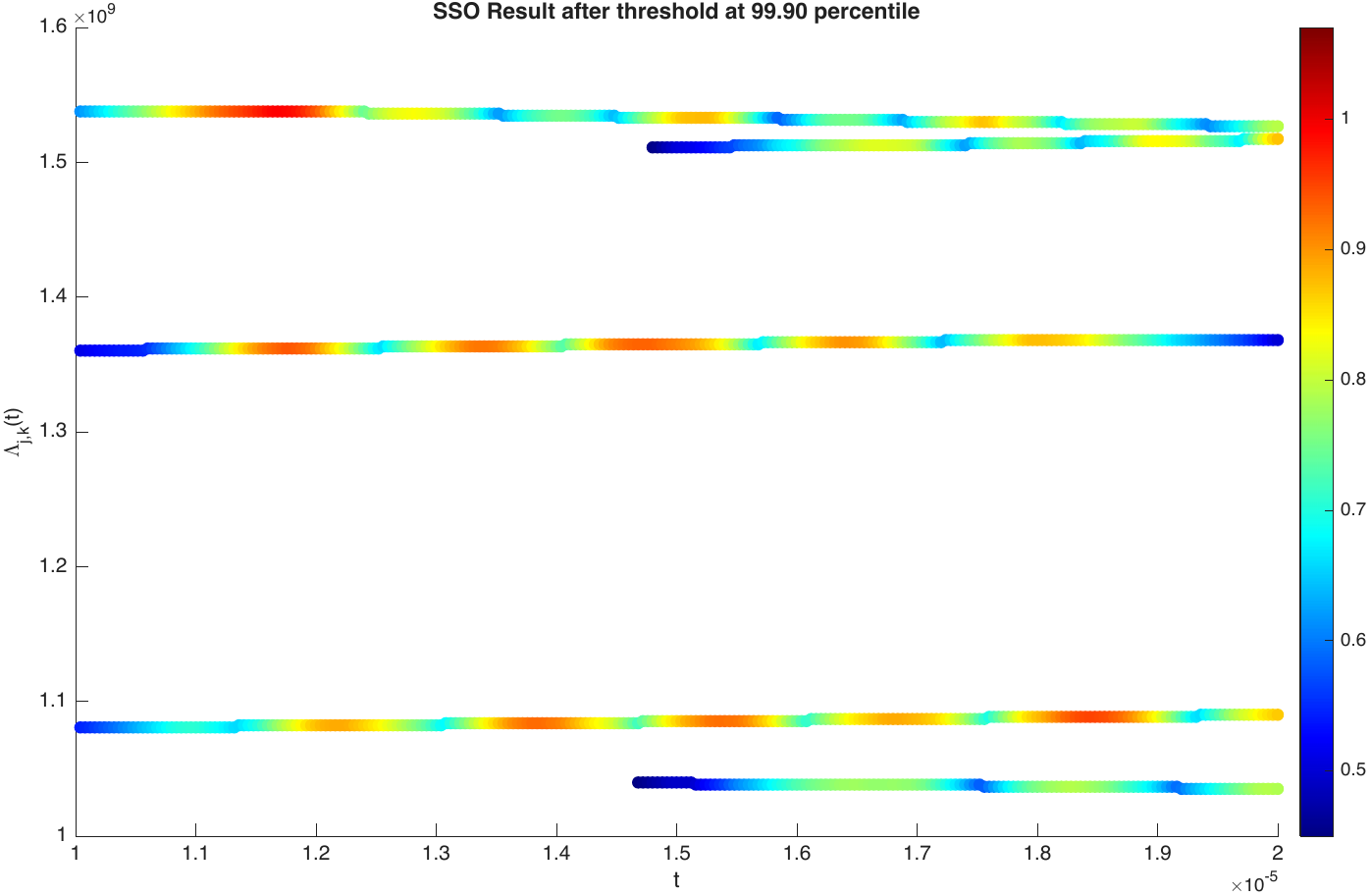}
\end{minipage}

\caption{This figure shows the comparison between SST and SSO method for $t$ between $1 \times 10^{-5}$ and $2 \times 10^{-5}$. The SST plots for example 1 with sampling rate 1 GHz at SNR 10,  -10 dB and the SSO plot for example 1 with sampling rate 1 GHz at SNR -10 dB.}
\label{fig:sst_1ghz}
\end{figure}

\color{black}
\bhag{Conclusions}\label{bhag:conclusion}
Separation of a superposition of blind source signals is an important problem in audio and radar signal processing, with many applications, including electronic warfare and electronic intelligence.
In this paper, we have modified a method proposed in \cite{bspaper} for solving this problem in the case when the instantaneous frequencies of the components are linear.
The method is based on a filtered FFT and hence is very fast.
While the paper \cite{bspaper} assumes the frequencies to be slowly and continuously varying, our modification includes the case when the signals are discontinuous and when there are crossover frequencies present. 
We analyze theoretically the behavior of SSO in the  presence of noise, giving insight into the relationship between the sampling frequency, desired  accuracy, and noise.
We have illustrated the effectiveness of our method with a few well chosen examples and tested it on a small database of 7 signals, demonstrating consistent behavior. 
Our method outperforms the SST algorithm given in \cite{thakur2013synchrosqueezing}.

 \printnomenclature

\bhag{Appendix}\label{bhag:appendix}
We report the statistics for 16 trials for each of the SNR and sampling rates for each signal, where the RMSE is computed using \eqref{eq:exptrmse}.
\begin{table}[H]
\begin{center}
\begin{tabular}{ |c|c|c|c|c|c|c|c|c| } 
 \hline
 SNR & Example & Sampling & Total & Detected  & RMSE & Standard \\
 (dB) & label & rate (GHz) & signals & signals & & Deviation \\
 \hline
10 & 1 & 0.500000 & 12 & 12 & 0.006295 & 0.001066 \\
10 & 1 & 1.000000 & 12 & 12 & 0.006095 & 0.000770 \\
10 & 1 & 25.000000 & 12 & 12 & 0.006325 & 0.001249 \\
0 & 1 & 0.500000 & 12 & 12 & 0.006610 & 0.001203 \\
0 & 1 & 1.000000 & 12 & 12 & 0.006133 & 0.000828 \\
0 & 1 & 25.000000 & 12 & 12 & 0.006375 & 0.001291 \\
-10 & 1 & 0.500000 & 12 & 12 & 0.007433 & 0.001214 \\
-10 & 1 & 1.000000 & 12 & 12 & 0.006211 & 0.000888 \\
-10 & 1 & 25.000000 & 12 & 12 & 0.006022 & 0.000690 \\
-20 & 1 & 25.000000 & 12 & 12 & 0.006012 & 0.000666 \\
-30 & 1 & 25.000000 & 12 & 12 & 0.006436 & 0.001336 \\
10 & 2 & 0.500000 & 12 & 12 & 0.004428 & 0.002676 \\
10 & 2 & 1.000000 & 12 & 12 & 0.004197 & 0.002619 \\
10 & 2 & 25.000000 & 12 & 12 & 0.003990 & 0.002588 \\
0 & 2 & 0.500000 & 12 & 12 & 0.004496 & 0.002685 \\
0 & 2 & 1.000000 & 12 & 12 & 0.004255 & 0.002624 \\
0 & 2 & 25.000000 & 12 & 12 & 0.004038 & 0.002593 \\
-10 & 2 & 0.500000 & 12 & 12 & 0.004565 & 0.002698 \\
-10 & 2 & 1.000000 & 12 & 12 & 0.004315 & 0.002633 \\
-10 & 2 & 25.000000 & 12 & 12 & 0.004086 & 0.002607 \\
-20 & 2 & 25.000000 & 12 & 12 & 0.004038 & 0.002590 \\
-30 & 2 & 25.000000 & 12 & 12 & 0.004088 & 0.002597 \\
10 & 3 & 0.500000 & 12 & 11 & 0.003915 & 0.002504 \\
10 & 3 & 1.000000 & 12 & 11 & 0.003859 & 0.002402 \\
10 & 3 & 25.000000 & 12 & 11 & 0.003818 & 0.002293 \\
0 & 3 & 0.500000 & 12 & 12 & 0.003934 & 0.002532 \\
0 & 3 & 1.000000 & 12 & 11 & 0.003877 & 0.002427 \\
0 & 3 & 25.000000 & 12 & 11 & 0.003834 & 0.002314 \\
-10 & 3 & 0.500000 & 12 & 11 & 0.003950 & 0.002561 \\
-10 & 3 & 1.000000 & 12 & 11 & 0.003895 & 0.002452 \\
-10 & 3 & 25.000000 & 12 & 11 & 0.003825 & 0.002355 \\
-20 & 3 & 25.000000 & 12 & 11 & 0.003814 & 0.002332 \\
-30 & 3 & 25.000000 & 12 & 11 & 0.003850 & 0.002335 \\
10 & 4 & 0.500000 & 12 & 9 & 0.003918 & 0.002293 \\
10 & 4 & 1.000000 & 12 & 9 & 0.004061 & 0.002310 \\
10 & 4 & 25.000000 & 12 & 10 & 0.004335 & 0.002404 \\
0 & 4 & 0.500000 & 12 & 9 & 0.003886 & 0.002292 \\
0 & 4 & 1.000000 & 12 & 9 & 0.004025 & 0.002306 \\
0 & 4 & 25.000000 & 12 & 10 & 0.004310 & 0.002406 \\
-10 & 4 & 0.500000 & 12 & 9 & 0.003861 & 0.002296 \\
-10 & 4 & 1.000000 & 12 & 9 & 0.003994 & 0.002305 \\
-10 & 4 & 25.000000 & 12 & 10 & 0.004132 & 0.002319 \\
-20 & 4 & 25.000000 & 12 & 10 & 0.004184 & 0.002339 \\
-30 & 4 & 25.000000 & 12 & 10 & 0.004283 & 0.002408 \\
 \hline
\end{tabular}
\end{center}
 	\caption{The tables above shows full performances of our algorithm.} \label{tab:result_table_1}
\end{table}

\begin{table}[H]
\begin{center}
\begin{tabular}{ |c|c|c|c|c|c|c|c|c| } 
 \hline
 SNR & Example & Sampling & Total & Detected  & RMSE & Standard \\
 (dB) & label & rate (GHz) & signals & signals & & Deviation \\
 \hline
10 & 5 & 0.500000 & 12 & 12 & 0.005752 & 0.002110 \\
10 & 5 & 1.000000 & 12 & 12 & 0.005199 & 0.002565 \\
10 & 5 & 25.000000 & 12 & 12 & 0.004693 & 0.002723 \\
0 & 5 & 0.500000 & 12 & 12 & 0.005926 & 0.001830 \\
0 & 5 & 1.000000 & 12 & 12 & 0.005325 & 0.002474 \\
0 & 5 & 25.000000 & 12 & 12 & 0.004781 & 0.002699 \\
-10 & 5 & 0.500000 & 12 & 10 & 0.006114 & 0.001440 \\
-10 & 5 & 1.000000 & 12 & 12 & 0.005460 & 0.002364 \\
-10 & 5 & 25.000000 & 12 & 12 & 0.004971 & 0.002684 \\
-20 & 5 & 25.000000 & 12 & 12 & 0.004873 & 0.002693 \\
-30 & 5 & 25.000000 & 12 & 12 & 0.004874 & 0.002666 \\
10 & 6 & 0.500000 & 12 & 10 & 0.004385 & 0.002373 \\
10 & 6 & 1.000000 & 12 & 10 & 0.004349 & 0.002319 \\
10 & 6 & 25.000000 & 12 & 10 & 0.004416 & 0.002305 \\
0 & 6 & 0.500000 & 12 & 10 & 0.004383 & 0.002387 \\
0 & 6 & 1.000000 & 12 & 9 & 0.004365 & 0.002332 \\
0 & 6 & 25.000000 & 12 & 10 & 0.004406 & 0.002313 \\
-10 & 6 & 0.500000 & 12 & 10 & 0.004380 & 0.002400 \\
-10 & 6 & 1.000000 & 12 & 9 & 0.004380 & 0.002345 \\
-10 & 6 & 25.000000 & 12 & 10 & 0.004337 & 0.002298 \\
-20 & 6 & 25.000000 & 12 & 10 & 0.004344 & 0.002290 \\
-30 & 6 & 25.000000 & 12 & 8 & 0.004395 & 0.002322 \\
10 & 7 & 0.500000 & 12 & 10 & 0.004440 & 0.002269 \\
10 & 7 & 1.000000 & 12 & 10 & 0.004427 & 0.002231 \\
10 & 7 & 25.000000 & 12 & 11 & 0.004435 & 0.002193 \\
0 & 7 & 0.500000 & 12 & 9 & 0.004436 & 0.002280 \\
0 & 7 & 1.000000 & 12 & 10 & 0.004441 & 0.002241 \\
0 & 7 & 25.000000 & 12 & 11 & 0.004436 & 0.002204 \\
-10 & 7 & 0.500000 & 12 & 9 & 0.004435 & 0.002290 \\
-10 & 7 & 1.000000 & 12 & 9 & 0.004453 & 0.002252 \\
-10 & 7 & 25.000000 & 12 & 10 & 0.004409 & 0.002211 \\
-20 & 7 & 25.000000 & 12 & 11 & 0.004409 & 0.002202 \\
-30 & 7 & 25.000000 & 12 & 9 & 0.004441 & 0.002214 \\
 \hline
\end{tabular}
\end{center}
 	\caption{The tables above shows full performances of our algorithm.} \label{tab:result_table_2}
\end{table}

\bibliographystyle{plain}
\bibliography{mason,refs,hrushikesh}

\end{document}